\documentclass{article}
\usepackage{geometry}
\usepackage[utf8]{inputenc}
\usepackage[T1]{fontenc}
\usepackage{microtype}
\usepackage{braket}
\usepackage{qcircuit}
\usepackage{amsthm}
\usepackage{amsmath,amssymb}
\usepackage{mathrsfs}
\usepackage{eufrak}
\usepackage{soul}
\usepackage{array}
\usepackage{makecell}
\usepackage{mdframed}
\usepackage{stmaryrd}
\usepackage{hyperref}
\usepackage[bottom]{footmisc}
\usepackage{scrextend}
\usepackage{tikz}
\usepackage{tikz-cd}
\usetikzlibrary{calc,arrows,decorations.pathreplacing,decorations.pathmorphing,intersections}
\usetikzlibrary{positioning}
\usetikzlibrary{fit}
\theoremstyle{plain}
\newtheorem{thm}{Theorem}[section]
\newtheorem{lem}[thm]{Lemma}

\newtheorem{cor}[thm]{Corollary}

\newtheorem{fact}{Fact}

\theoremstyle{definition}
\newtheorem{prtl}{Protocol}
\newtheorem{conv}{Convention}
\newtheorem{toyprtl}{Toy Protocol}
\newtheorem{setup}{Set-up}
\newtheorem{defn}{Definition}[section]
\newtheorem{exmp}{Example}[section]
\newtheorem{nota}{Notation}[section]
\theoremstyle{remark}

\DeclareMathOperator{\tr}{tr}

\newcommand{\cC}{{\mathcal{C}}}

\newcommand{\cE}{{\mathcal{E}}}

\newcommand{\cH}{{\mathcal{H}}}

\newcommand{\btheta}{{\boldsymbol{\theta}}}

\newcommand{\bS}{{\boldsymbol{S}}}
\newcommand{\baux}{{\boldsymbol{aux}}}
\newcommand{\bD}{{\boldsymbol{D}}}

\newcommand{\bx}{{\boldsymbol{x}}}
\newcommand{\bbt}{{\boldsymbol{t}}}

\newcommand{\bR}{{\boldsymbol{R}}}
\newcommand{\bbbr}{{\boldsymbol{r}}}

\newcommand{\boldy}{{\boldsymbol{y}}}
\newcommand{\bd}{{\boldsymbol{d}}}

\newcommand{\bflag}{{\boldsymbol{flag}}}

\newcommand{\bbC}{{\boldsymbol{C}}}

\DeclareMathOperator{\bN}{\mathbb{N}}

\newcommand{\fM}{{\sf M}}

\newcommand{\fI}{{\sf I}}

\newcommand{\fDisgard}{{\sf Disgard}}

\newcommand{\fHash}{{\mathsf{Hash}}}

\newcommand{\fSAoK}{{\mathsf{SAoK}}}
\newcommand{\fsuccZKAoK}{{\mathsf{succZKAoK}}}
\newcommand{\fSFVP}{{\mathsf{SFVP}}}
\newcommand{\fSFVPTwo}{{\mathsf{SFVP2}}}
\newcommand{\fCommit}{{\mathsf{Commit}}}

\newcommand{\fGarbleC}{{\mathsf{GarbleC}}}
\newcommand{\fGarbleI}{{\mathsf{GarbleI}}}

\newcommand{\fDc}{{\mathsf{Dec}}}
\newcommand{\fEv}{{\mathsf{Eval}}}
\newcommand{\fKg}{{\mathsf{KeyGen}}}

\newcommand{\fneg}{{\mathsf{negl}}}
\newcommand{\fpoly}{{\mathsf{poly}}}

\newcommand{\fAdv}{{\mathsf{Adv}}}

\mdfdefinestyle{figstyle}{ 
  linecolor=black!7, 
  backgroundcolor=black!7
}

\mdfdefinestyle{figstyle}{ 
  linecolor=black!7, 
  backgroundcolor=black!7
}
\newcommand{\fSim}{{\mathsf{Sim}}}
\newcommand{\fExt}{{\mathsf{Ext}}}

\newcommand{\fCHK}{{\mathsf{CHK}}}
\newcommand{\ftrue}{{\mathsf{true}}}
\newcommand{\ffalse}{{\mathsf{false}}}

\newcommand{\fpass}{{\mathsf{pass}}}
\newcommand{\ffail}{{\mathsf{fail}}}

\newcommand{\fCollapse}{{\mathsf{Collapse}}}

\newcommand{\fHadamardTest}{{\mathsf{HadamardTest}}}
\newcommand{\fComputationalTest}{{\mathsf{ComputationalTest}}}
\newcommand{\fFirstTwoStep}{{\mathsf{FirstTwoStep}}}
\newcommand{\fAlignR}{{\mathsf{AlignR}}}
\newcommand{\fSecondStep}{{\mathsf{SecondStep}}}
\newcommand{\fThirdStepComp}{{\mathsf{ThirdStepComp}}}

\newcommand{\tind}{{\text{ind}}}

\newcommand{\sk}{{\text{sk}}}

\newcommand{\pk}{{\text{pk}}}

\newcommand{\fHybrid}{{\mathsf{Hybrid}}}

\newcommand{\tSFSIdeal}{{\text{SFSIdeal}}}
\newcommand{\tSFVPIdeal}{{\text{SFVPIdeal}}}
\newcommand{\tSFVPTwoIdealS}{{\text{SFVP2IdealS}}}
\newcommand{\tSFVPTwoIdealC}{{\text{SFVP2IdealC}}}
\newcommand{\fSFSTest}{{\mathsf{SFSTest}}}
\newcommand{\fsuccTest}{{\mathsf{succTest}}}
\newcommand{\fSFSComp}{{\mathsf{SFSComp}}}
\newcommand{\fSFS}{{\mathsf{SFS}}}

\newcommand{\fCompressor}{{\mathsf{Compressor}}}
\newcommand{\fDecompressor}{{\mathsf{Decompressor}}}

\newcommand{\bX}{{\boldsymbol{X}}}

\newcommand{\bbE}{{\boldsymbol{E}}}
\newcommand{\bbs}{{\boldsymbol{s}}}
\newcommand{\bQ}{{\boldsymbol{Q}}}

\newcommand{\bregs}{{\boldsymbol{regs}}}
\newcommand{\breg}{{\boldsymbol{reg}}}
\newcommand{\tD}{{\text{D}}}
\newcommand{\tDist}{{\text{Dist}}}
\newcommand{\tDen}{{\text{Den}}}
\newcommand{\tTD}{{\text{TD}}}

\newcommand{\ttest}{{\text{test}}}
\newcommand{\tcomp}{{\text{comp}}}
\newcommand{\tmode}{{\text{mode}}}
\newcommand{\tvals}{{\text{vals}}}
\newcommand{\tReal}{{\text{Real}}}
\newcommand{\tIdeal}{{\text{Ideal}}}
\newcommand{\tTwoPCIdeal}{{\text{2PCIdeal}}}
\newcommand{\tRSPVIdeal}{{\text{RSPVIdeal}}}

\newcommand{\ttemp}{{\text{temp}}}

\newcommand{\tPRG}{{\text{PRG}}}

\usepackage{authblk}
\usepackage{subcaption}
\begin{document}
\title{A Quantum Approach For Reducing Communications in Classical Secure Computations with Long Outputs}
	 \author[1]{Jiayu Zhang\footnote{zhangjy@zgclab.edu.cn}}
	 \affil[1]{Zhongguancun Laboratory}
	\maketitle\thispagestyle{empty}
	\begin{abstract}
        How could quantum cryptography help us achieve what are not achievable in classical cryptography? In this work we study the classical cryptographic problem that two parties would like to perform secure computations \emph{with long outputs}. As a basic primitive and example, we first consider the following problem which we call \emph{secure function sampling} with long outputs: suppose $f:\{0,1\}^n\rightarrow \{0,1\}^m$ is a public, efficient classical function, where $m$ is big; Alice would like to sample $x$ from its domain and sends $f(x)$ to Bob; what Bob knows should be no more than $f(x)$ even if it behaves maliciously. Classical cryptography, like FHE and succinct arguments \cite{GentryFHE,FHELWE,Kilian92,HW15}, allows us to achieve this task within communication complexity $O(n+m)$; could we achieve this task with communication complexity independent of $m$?\par
        In this work, we first design a quantum cryptographic protocol that achieves secure function sampling with approximate security, within $O(n)$ communication (omitting the dependency on the security parameter and error tolerance). We also prove the classical impossibility using techniques in \cite{HW15}, which means that our protocol indeed achieves a type of quantum advantage. Building on the secure function sampling protocol, we further construct protocols for general secure two-party computations \cite{YaoGCOrigin,GB01} with approximate security, with communication complexity only depending on the input length and the targeted security. In terms of the assumptions, we construct protocols for these problems assuming only the existence of collapsing hash functions \cite{Unruh16}; what's more, we also construct a classical-channel protocol for these problems additionally assuming the existence of noisy trapdoor claw-free functions \cite{BCMVV,BKVV}.
    \end{abstract}
    \tableofcontents
\section{Introduction}\label{sec:1}
\subsection{Background}\label{sec:1.1}
The rapid development of quantum information science leads to the rapid development of quantum cryptography. One remarkable feature about quantum cryptography is that it could allow us to achieve what are not achievable in classical cryptography. \cite{PR22,Sattath22,Zhand17} For example:
\begin{itemize}\item Quantum key distribution \cite{PR22} achieves information-theoretic secure key exchange, which is not possible classically.
    \item Quantum cryptography allows us to make cryptographic primitives unclonable \cite{Zhand17,Sattath22}.
    \item Other examples include test of quantumness \cite{BCMVV,BKVV,KLVY22,YZ22}, certified deletion \cite{BI20,Unruh15}, position verification \cite{PV10}, etc.
\end{itemize}
Discovering new quantum advantage in cryptography is of significant theoretical interest and may even have practical impacts. In this work we focus on the quantum approach for a class of classical cryptographic tasks, which we introduce below.
\subsubsection{Secure computations with long outputs}\label{sec:1.1.1}
A central topic of cryptography is the secure computation. An example is the Yao's millionaires problem \cite{GB01}: two millionaires want to compare who is richer, but they do not want to reveal how wealthy they are; so they want to use a cryptographic protocol to solve the problem. The Yao's millionaires problem is an example of the general \emph{secure two-party computation} problem, which is as follows. Suppose Alice holds a private input $x_A$ and Bob holds a private input $x_B$, and Alice would like to learn $f_A(x_A,x_B)$ and Bob would like to learn $f_B(x_A,x_B)$. Here $f_A,f_B$ are efficient functions (which are described by, for example, Turing machines or polynomial-size circuits). We would like to design cryptographic protocols that at least satisfy three requirements, the correctness (or called completeness), security (or called soundness) and efficiency, which are roughly as follows:
\begin{itemize}
    \item (Correctness) When both parties follow the protocol (which is called the honest behavior or the honest setting), Alice gets $f_A(x_A,x_B)$ and Bob gets $f_B(x_A,x_B)$.
    \item (Security) When some party deviate from the protocol (which is called the malicious setting), the overall behavior of the protocol is as if the protocol is executed by some trusted third party, or called ideal functionality. An adversary is either caught cheating or only learns what it should learn from the trusted third party.\par
    We note that this discussion is informal and formalizing the security requires some works. And there are also other settings and definitions for the security; here we focus on the security in the malicious setting.
    \item (Efficiency) Both parties should run in polynomial time.
\end{itemize}
The study of secure two-party and multi-party computation has become a central topic in cryptography; people have developed many important theories and techniques for its different settings, including garbled circuits, fully homomorphic encryption (FHE), succinct arguments \cite{YaoGCOrigin,GentryFHE,FHELWE,Kilian92}, etc.\par
One important factor in protocol design is the \emph{communication complexity}, that is, the total size of communications. For simplicity of the introduction let's consider the setting where Alice and Bob want to perform secure two-party computation on a public, efficient classical function $f:\{0,1\}^n\times \{0,1\}^n\rightarrow \{0,1\}^m$. In classical cryptography, one typical solution for reducing communications is based on FHE and succinct arguments, which achieves communication complexity $O(n+m)$ (omitting the dependency on the security parameter); we refer to \cite{AJW12,DFH12,HW15} for its constructions. 
What's remarkable for these protocols is that the communication complexity is independent of the time complexity of evaluating $f$, which shows significant advantage when the time complexity of evaluating $f$ is much bigger than $n+m$.\par
However, there is still an undesirable output-length dependency in the classical solutions. This means that when the function to be evaluated has long outputs (that is, when $m$ is big), the communication complexity will still be big.  Indeed, as proved in \cite{HW15}, this output-length dependency is unavoidable if we want the security in the malicious setting. 
We note that this output-length dependency is avoided by the following trivially insecure protocol: in the example above, both parties simply reveal their inputs $x_A$ and $x_B$ and compute $f(x_A,x_B)$ on their own. This roughly means that, in the world of classical cryptography, we have to pay the price of the output-length dependency to get good security in general secure two-party computations.\par
However, quantum cryptography gives us the possibility for overcoming the limitations of classical cryptography. In this background, we ask:
\begin{center}
    \emph{Could quantum cryptography help to reduce communications in secure computations with long outputs?}
\end{center}
\subsection{Our Contributions}
In this work we achieve a type of quantum advantage in the problem above. In summary, we design quantum cryptographic protocols for classical secure computation problems with approximate security, with communication complexity independent of both the function evaluation time and the output length. We first focus on a special problem called secure function sampling, and then study the general secure two-party computations problem. Below we elaborate our results.
\subsubsection{A quantum advantage in reducing communications in secure function sampling}\label{sec:1.2.1}
We first formalize and study a problem called \emph{secure function sampling} (SFS). SFS is a special case of the general secure two-party computation problem (2PC).
\paragraph{Definitions of secure function sampling}
The SFS problem is informally defined as follows. We elaborate some missing details in Section \ref{sec:2.2} and 
the formal definition is given in Section \ref{sec:4.1}.
\begin{defn}[Secure function sampling (SFS), informal]Consider a protocol between two parties, the client and the server. Suppose $f:\{0,1\}^n\rightarrow  \{0,1\}^m$ is a public, efficient classical function (which is described by a polynomial time classical Turing machine or polynomial-size classical circuits). The protocol should satisfy the following requirements:
    \begin{itemize}
        \item (Correctness) When both parties are honest, in the end of the protocol the client gets a random $x\in \{0,1\}^n$ and the server gets the corresponding $f(x)$.
        \item (Security) Even if the server is malicious, what it could get from the protocol is no more than one of the following two cases: \begin{itemize}\item either the server only gets $f(x)$ as in the honest setting, \item or the server is caught cheating by the client and gets nothing.\end{itemize}
        In this work we consider an approximate variant of the security, where the error tolerance is denoted by $\epsilon$.
        \item (Efficiency) Both parties should run in polynomial time.
    \end{itemize}
\end{defn}
And our goal is to design a protocol for the SFS problem above so that the communication complexity only depends on the input length $n$ and the targeted security.\par
The setting of SFS is restricted compared to the general 2PC problem; one remarkable restriction is that the $x\in \{0,1\}^n$ is sampled during the protocol as the client-side outputs instead of controlled by the client as client-side inputs. Note that the fact that the primitive is restricted makes the classical impossibility result stronger. What's more, later we will construct protocols for 2PC based on our protocol for SFS.
\paragraph{Classical impossibility} By generalizing the results and techniques in \cite{HW15}, we show that the output length dependency is unavoidable for classical SFS protocol, even if we allow the security to be approximate:
\begin{thm}[Informal]
Assuming the existence of one-way functions, there exists a family of functions $f$ 
such that there is no classical protocol that achieves SFS for $f$ with security error tolerance $\epsilon<1-O(1)$ within communication complexity $\fpoly(n)$.
\end{thm}

\paragraph{Quantum protocol} Then we construct a quantum cryptographic protocol that solves the SFS problem above.
\begin{thm}[Informal]\label{thm:1.2}
Assuming the existence of collapsing hash functions, there exists a quantum cryptographic protocol (Protocol \ref{prtl:6}) for the SFS problem that is $\epsilon$-secure and has communication complexity (including quantum and classical communications) $\fpoly(n,1/\epsilon)$.
\end{thm}
Here the collapsing hash function \cite{Unruh16} is a quantum generalization of the collision-resistant hash function. It's plausible to assume common cryptographic hash functions, like the SHA-3 \cite{KLtextbook}, satisfy this property \cite{Unruh16}; collapsing hash function could also be constructed from the LWE assumption \cite{Unruh16a}. See Section \ref{sec:3.3.2} for a more detailed discussion.\par
Finally we note that in the theorem statement above we omit some desirable properties of our protocol. In particular, in our protocol the client-side computation is also succinct (that is, $\fpoly(n,1/\epsilon)$).

\subsubsection{Protocols for the secure two-party computations}
Building on our SFS protocol, we construct protocols for 2PC. 
As before, we allow for approximate security, and aim at communication complexity that only depends on the input length and the targeted security. We further consider two different settings: in one setting we allow the quantum communication, and in the other setting we only allow the classical communication and local quantum computation. Our results are as follows.
\begin{thm}[Informal]\label{thm:1.3}
Assuming the existence of collapsing hash functions, there exists a quantum cryptographic protocol (Protocol \ref{prtl:2pc}) for the 2PC problem that is $\epsilon$-secure and has communication complexity  (including quantum and classical communications) $\fpoly(n,1/\epsilon)$.
\end{thm}
\begin{thm}[Informal]\label{thm:1.4}
    Assuming the existence of NTCF and collapsing hash functions, there exists a quantum cryptographic protocol over classical channel (Protocol \ref{prtl:cc2pc}) for the 2PC problem that is $\epsilon$-secure and has communication complexity $\fpoly(n,1/\epsilon)$.
\end{thm}
Here NTCF stands for noisy trapdoor claw-free functions \cite{BCMVV,zhang24,AMR22,BKVV}; in particular we do not need the adaptive hardcore bit property. The NTCF assumption could be instantiated by the LWE assumption \cite{BCMVV,BKVV}. (See Section \ref{sec:3.3.6} for details.)
\subsection{Brief Technical Overview}
We give a brief technical overview in this section; a more detailed technical overview is given in Section \ref{sec:techoverv}. We focus on two parts in this brief overview: (1) how do we construct the SFS protocol (which corresponds to Theorem \ref{thm:1.2})? (2) how do we go from the SFS protocol to the general secure two-party computations protocol (which corresponds to Theorem \ref{thm:1.3})?
\subsubsection{Construction of the SFS protocol}
We first give a toy protocol to illustrate the idea behind our protocol.\par
Very roughly speaking, we would like to design a protocol that works as follows:
\begin{enumerate}\item The client first uses some way to send the function input $x$ to the server, which allows the server to evaluate $f$ and get $f(x)$;\item Then the client deletes all the information other than the function output $f(x)$.\end{enumerate}
We first consider the following toy protocol.
\begin{enumerate}\item The client randomly samples $x_0,x_1\in \{0,1\}^n$ randomly and sends their superposition $$\frac{1}{\sqrt{2}}(\ket{x_0}+\ket{x_1})$$ to the server.
    \item The server evaluates $f$ and gets the state $$\frac{1}{\sqrt{2}}(\ket{x_0}\ket{f(x_0)}+\ket{x_1}\ket{f(x_1)})$$
    The server sends back the input register (holding $\ket{x_0},\ket{x_1}$) to the client.
    \item The client measures this register to collapse the state. If the server is honest the final outputs achieve SFS for $f$.
\end{enumerate}
We note that this protocol hasn't achieved what we want in several senses; especially we do not have security against malicious servers: the server could simply store the function inputs in step 2 and break the security.\par
Starting from this toy protocol, we will do a series of revisions and finally construct an SFS protocol. For a brief technical overview, we discuss three ideas in our constructions:
\begin{itemize}
    \item We will make use of the Hadamard test to test the state and certify the server's behavior. Note that in the toy protocol above, instead of asking the server to send back the input register, the client could ask the server to measure the input register on Hadamard basis; suppose the server's measurement result is $d^{(in)}\in \{0,1\}^n$. Then if the client further asks the server to measure the output register on the Hadamard basis and get $d^{(out)}\in \{0,1\}^m$, the following relation should be satisfied:
    $$d^{(in)}\cdot(x_0+x_1)+d^{(out)}\cdot (f(x_0)+f(x_1))\equiv 0\mod 2$$
    This provides a powerful test for certifying the server's operations. But a significant drawback is, the honest behavior is affected: if we revise the protocol in this way, once the input register is measured on the Hadamard basis, the client could not know which input (either $x_0$ or $x_1$) it should choose corresponding to the server-side state. In later constructions we will combine these ideas and introduce new ideas to design protocols that could both control the malicious server's behavior and preserve the honest setting behavior.
    \item In the construction of the full protocol, we will need to first construct a forward-succinct protocol, which means, the client-to-server communication complexity only depends on $n$ and the targeted security; then we further compile the protocol to make the server-to-client communication succinct. We formalize this protocol for reducing the server-to-client communication as a primitive that we call ``succinct testing for two-party relations'', and its construction is a variant of the answer-reduction compiler in \cite{BKLMM22}.
    \item We will view the SFS problem as a quantum cryptographic primitive called \emph{remote state preparation with verifiability} (RSPV), and use the corresponding framework, techniques and results \cite{GVRSP,BGKPV23,zhang24} to work on this problem.
\end{itemize}
In Section \ref{sec:2.2} to \ref{sec:2.5} we discuss our ideas in more detail.
\subsubsection{Construction of the general two-party computations protocol}
After the construction of the SFS protocol, we would like to further construct the protocol for general secure two-party computations. We want the communication complexity to be succinct (that is, polynomial in the input size and the targeted security, even when the outputs are long); what's more, we only want to assume the existence of collapsing hash functions. Although we allow for approximate security and could use quantum communication, the construction still requires a significant amount of works.\par
For a brief technical overview, we show an important step of our construction: building on our SFS protocol, we construct a protocol called \emph{secure function value preparation} (SFVP).\par
The notion of SFVP is similar to SFS with the following difference: in SFVP the function input $x$ is provided by the client as its private input (for comparison, in SFS $x$ is sampled and the client does not control it). To construct SFVP from SFS, we are going to use a classical cryptographic primitive called \emph{Yao's garbled circuits}.\par
\paragraph{Background on Yao's garbled circuits} Given a circuit $C$ and an input $x$ of the circuit, Yao's garbled circuit allows us to encode the circuit and input as follows: 
\begin{enumerate}
    \item Sample a random string $r$; the length of $r$ scales with the length of $x$.
    \item Encode $C$ by $r$ and denote the output as $\hat{C}$. We call this the circuit encoding.
    \item Encode $x$ by $r$ and denote the output as $\hat{x}$. We call this the input encoding.
\end{enumerate}
Then if the server is given $\hat{C},\hat{x}$, it could recover $C(x)$; what's more, what the server could know from $\hat{C},\hat{x}$ is no more than $C(x)$. Finally the Yao's garbled circuits could be constructed assuming only one-way functions, which are implied by collapsing hash functions.
\paragraph{Construction of SFVP}Using Yao's garbled circuits and SFS, our SFVP protocol goes as follows:
\begin{enumerate}
    \item Using SFS to sample the circuit encoding; the client gets a random $r$ and the server gets the corresponding $\hat{C_f}$, where $C_f$ is the circuit description of $f$.
    \item The client uses $r$ to compute and send the input encoding $\hat{x}$ to the server.
\end{enumerate}
In other words, the first two steps in the garbling scheme above are combined and implemented using SFS; note that the communication is succinct since $r$ is succinct.\footnote{The size of $C_f$ does not matter as long as $f$ has a succinct description (for example, by a Turing machine): in SFS protocol the client does not need to explicitly compute the circuit description $C_f$.}\par
Finally we briefly discuss other main obstacles and our techniques for constructing a general two-party computations protocol in succinct communication.
\begin{itemize}
    \item A 2PC protocol needs to take private inputs from both parties; intuitively we could see the the second step in the SFVP protocol above (the client sends the input encoding to the server) does not work any more. To solve this problem, we make use of the oneway-function-based 2PC protocol constructed in \cite{BCKM20,GLSV20}; since this 2PC step is only applied on the input encoding part the communication complexity is still succinct regardless of which 2PC protocol we use here.
    \item In the SFVP protocol above there is no security guarantee against a malicious client; this means that a malicious client could send some arbitrary value to the server without being detected. We will make use of several cryptographic primitives including the statistically-binding commitment, succinct zero-knowledge argument-of-knowledge, and collapsing hash functions to achieve the soundness against a malicious client.
    \item In the construction of 2PC protocol we will also encounter a technical issue that seems related to the adaptive security of garbling \cite{asyao}. To bypass this issue our protocol actually works on the garbling of the garbling of the function.
    \item Finally we construct the classical-channel 2PC protocol by compiling our quantum-channel 2PC protocol with existing RSPV constructions \cite{BGKPV23,zhang24} (corresponding to Theorem \ref{thm:1.4}).
\end{itemize}
\subsection{Related Works}
\paragraph{Other works on reducing communications}
In classical cryptography there are a plenty of works for reducing communications in various secure computation problems. These include fully homomorphic encryption (FHE) \cite{GentryFHE,FHELWE}, succinct arguments \cite{Kilian92}, succinct garbling \cite{LinPass14}, laconic function evaluation \cite{QWW18}, etc. Many of these primitives have achieved succinctness in their settings and definitions. But for our problem, SFS and general two-party computations for functions with long outputs, none of these primitives could make the communication complexity independent to both the function evaluation time and the output length. (As discussed before, this is impossible to achieve in classical cryptography.)\par
For succinctness in quantum cryptographic protocols, there are a series of works on succinct arguments for QMA \cite{BKLMM22,MNZ24,GTNV24,CCT,Faisal23}. Note that the problem setting of succinct arguments for QMA is different from our work: in succinct arguments for QMA we consider quantum circuits and the inputs and outputs are all public. Another work that discusses succinctness is \cite{jiayu20}, where the succinctness refers to the size of quantum communication in the protocol (where the size of classical communication is not counted).
\paragraph{Deletion in quantum cryptographic protocols} Our construction of the forward-succinct protocols contains a sense of revocation or deletion, which is a quantum phenomenon: recall that the client first sends the input to the server to allow the server to do something, and then this part could be measured on the Hadamard basis and gets destroyed. A series of existing works also make use of some types of deletion or revocation to achieve nontrivial cryptographic tasks, for example, \cite{Unruh15,BI20,BK,cvqcinlt}. As far as we know, our usage of this quantum phenomenon is different from these existing works on both the targeted problems and the technical level.
\paragraph{Other quantum advantages} There are various types of quantum advantages. Examples include quantum algorithms like the Shor algorithm and what have been discussed in Section \ref{sec:1.1}. Other examples include quantum advantage in shallow circuits \cite{BGK17}, quantum advantage in communication complexity (without considering the security) \cite{ATY17}. As far as we know, the quantum advantage in the communication complexity of secure computations has not been studied before.
\paragraph{Previous versions} This version has been significantly improved compared to its previous versions
 in 2023 \cite{qafirstversion}. 
Many limitations of the protocols in the previous versions have been overcome: previous versions only handle constant error tolerance and this version could achieve inverse-polynomial error tolerance; this version contains a proof of classical impossibility\footnote{We thank Prabhanjan Ananth for pointing to paper \cite{HW15}}; this version also contains protocols for secure two-party computations (some of them rely on a result coming out after the previous versions \cite{zhang24}).
\subsection{Summary and Discussion}
The big message of this work is that we study and discover a new type of quantum advantage in cryptography: we show that for the problem of secure function sampling with approximate security, there exists a quantum cryptographic protocol that achieves this task with communication complexity independent of both the function evaluation time and the output length, while it's impossible to achieve this task solely by classical cryptography. Furthermore, we generalize our results from the secure function sampling problem to the general secure two-party computation problem.\par
Our results lead to a series of open problems. For example:
\begin{itemize}
    \item Could we further generalize the problem setting (for example, to quantum circuits) or further weaken the assumptions (for example, to post-quantum one-way functions)?
    \item We note that for some types of functions, when we do not care about the security, the communication complexity could be even smaller than the input length. This sugguests that there might be some further room for reducing communications in secure computations: could we design secure computation protocols with communication complexity $O(CC(f))$ (where $CC(f)$ denotes the communication complexity of $f$ without considering the security)?
    \item In our constructions we only achieve approximate security. Could we overcome this issue and make the security loss negligible?
    \item Could our protocols be made practical?
\end{itemize}

\section*{Acknowledgement}
We thank Prabhanjan Ananth, Zhengfeng Ji and anonymous reviewers for discussions and comments. This work is done in Zhongguancun Laboratory.

\section*{Author Contribution Statement}
The author (Jiayu Zhang) is the person who makes all the contributions in this work. The author does not use AI in this work.

\section{Technical Overview}\label{sec:techoverv}
In this section we give a detailed technical overview for the whole constructions.
\subsection{Overview and Organizations}
Below we first give a high-level overview, including organizations of this section and the whole paper.
\paragraph{High-level overview and organization of this section} The secure function sampling (SFS) problem is introduced in Section \ref{sec:1.2.1}. First, although the goal of the problem is purely classical, to design its quantum protocol, we will first view it as a quantum cryptographic problem called \emph{remote state preparation with verifiability} (RSPV). In Section \ref{sec:2.2} we elaborate the notions of SFS, RSPV and their relations. \par
In Section \ref{sec:2.3}, \ref{sec:2.4} and \ref{sec:2.5} we give an overview for our construction of the quantum protocol for SFS. We first present a forward-succinct protocol, which means, the client-to-server communication is succinct; then we show how to revise it to a fully succinct protocol by a variant of the compiler in \cite{BKLMM22}.\par
In Section \ref{sec:2.6} we discuss the classical impossibility. The technique is an adaptation of \cite{HW15}.\par
In Section \ref{sec:2.7} we show how to construct protocols for general two-party computations with succinct communication.\par
\paragraph{Organization of the whole paper} The main content is organized in a slightly different order from this technical overview.\par
In Section \ref{sec:3} we review the preliminaries of our works.\par
    In Section \ref{sec:4} we formalize the SFS problem, prove the classical impossibility and prepare for the constructions.\par
    In Section \ref{sec:5} and Section \ref{sec:6} we give our construction of the quantum protocol for SFS. We first formalize the protocol that reduces the server-to-client communication (which is a variant of the compiler in \cite{BKLMM22}), and then construct the full protocol.\par
    In Section \ref{sec:7} we construct our protocols for general two-party computations. 


\subsection{Secure Function Sampling and Remote State Preparation with Verifiability}\label{sec:2.2}
As said before, to construct the SFS protocol, we will first view it as an RSPV problem and use the related tools. Below we first explain some more details in the description of the SFS problem, and then review the notion of RSPV and connect SFS to this notion.
\paragraph{Details of the SFS problem} The SFS problem is informally introduced in Section \ref{sec:1.2.1}. One missing detail is that, in the problem modeling the protocol will take an additional parameter $\kappa$ called the security parameter, which describes the security level of the protocol. We further clarify some naming and notation: note that the client-side outputs have a special part called flag, whose value is in $\{\fpass,\ffail\}$, where $\fpass$ means the protocol completes successfully and both parties should get the desired outputs, and $\ffail$ means the server is caught cheating (or called aborted, or rejected). \par
 Furthermore, the security is described by a security definition paradigm called simulation-based security. Recall that in Section \ref{sec:1.2.1} we describe the security of SFS as follows: ``what it (the malicious server) gets from the protocol is no more than one of the following two cases: ...''. Here the ``one of the following two cases: ...'' is called the ideal functionality of the problem; ``what it gets from the protocol'' is the real execution of the protocol. Then the formal security statement goes as follows:
 \begin{equation*} \forall \text{ poly-time adversarial server } \fAdv, \exists \text{ poly-time simulator }\fSim\text{ working on the server side}\end{equation*}\begin{equation}\label{eq:1}\text{such that }\tReal^{\fAdv}\approx^{ind} \fSim\circ\tIdeal.\end{equation}
 Here $\tReal^{\fAdv}$ denotes the real execution against adversary (malicious server) $\fAdv$,  $\fSim\circ\tIdeal$ is the ideal functionality composed with the simulator (note that there might be interaction between $\fSim$ and $\tIdeal$), and $\approx^{ind}$ denotes the computational indistinguishability. Here the introduction of the simulator, together with the computational indistinguishability, formalizes what it means by saying ``... is no more than ...'' in the informal claim.\par
 Finally we discuss the approximate security. This is defined by taking the indistinguishability in \eqref{eq:1} to be approximate. In other words, no polynomial-time distinguisher could distinguish the two sides of \eqref{eq:1} with probability advantage more than $\epsilon$.\par
 We elaborate the details of the security definition paradigm in Section \ref{sec:3.2}.
\paragraph{A review of remote state preparation with verifiability (RSPV)} RSPV is an important primitive in quantum cryptography. As its basic form, consider a finite set of states $\{\ket{\varphi_1},\ket{\varphi_2},\cdots\ket{\varphi_D}\}$. The client would like to sample a random state $\ket{\varphi_i}$ from this set and send this state to the server, so that in the end the client knows the state description (which could be the index $i$ here), while the server only holds the state. In the malicious setting, what the malicious server could get should be no more than one of the following two cases: (1) the server gets the state as in the honest behavior; (2) the server is caught cheating and gets nothing. Existing works give RSPV constructions for several specific state families \cite{GVRSP,qfactory,cvqcinlt,GMP,zhang24}. We give a more detailed review for RSPV in Section \ref{sec:3.3.7}.\par
\cite{zhang24} provides a framework for working on RSPV; this work will build on this framework. 
\paragraph{SFS, RSPV and preRSPV} We first note that SFS could be viewed as a variant of RSPV. The problem setting is basically the same: in SFS the server would like to get a classical value $f(x)$; note that classical states could be viewed as a special case of quantum states. Then the state family (the $\{\ket{\varphi_1},\ket{\varphi_2},\cdots\ket{\varphi_D}\}$ above) corresponds to the values of $f$ over all the input $x$.\par
 Then as discussed in \cite{zhang24}, its construction could be reduced to a notion called preRSPV, as follows. The preRSPV is defined to be a pair of protocols $(\pi_{\ttest},\pi_{\tcomp})$; these two protocols run in the same set-up and stand for the test mode and the computation mode correspondingly. The soundness (or called security or verifiability) of preRSPV is roughly defined as follows: for any malicious server, if it passes in the test-mode protocol, then the computation-mode protocol achieves what we want (that is, prepare the target output states). On the one hand, a preRSPV could be amplified to an RSPV protocol by sequential repetition, as shown in \cite{zhang24}; on the other hand, preRSPV is usually easier to work on when we construct protocols from scratch.\par
Below we will construct the preRSPV for our problem. Then it could be amplified to an RSPV, which is exactly a protocol for SFS.
\subsection{Construction of the Forward-succinct Protocol for Certain Functions}\label{sec:2.3}
As said before, we will first construct forward-succinct protocols, which means, the size of the client-to-server communication is $\fpoly(n,\kappa,1/\epsilon)$. We will construct a preRSPV (as described in Section \ref{sec:2.2}) protocol that is forward-succinct.\par
In this section we will first consider a toy protocol, given below in Toy Protocol \ref{toyprtl:1}, \ref{toyprtl:2}. Toy Protocol \ref{toyprtl:1}, \ref{toyprtl:2} is a forward-succinct preRSPV protocol for functions with certain desirable properties. In Section \ref{sec:2.4} we will further revise the protocol and get a forward-succinct preRSPV protocol for general functions.\par
Below we first discuss the protocol in Section \ref{sec:2.3.1}, and then discuss the intuition for its soundness in Section \ref{sec:2.3.2}.
\subsubsection{Protocol design}\label{sec:2.3.1}
As said before, in preRSPV we consider a pair of protocols, corresponding to the computation mode and the test mode. We first present the computation-mode protocol.
\begin{toyprtl}[Computation-mode]\label{toyprtl:1}
\begin{enumerate}
    \item The client samples and stores two different strings $x_0,x_1\in \{0,1\}^n$; then the client prepares and sends \begin{equation}\label{eq:1r}\frac{1}{\sqrt{2}}(\ket{x_0}+\ket{x_1})\end{equation} to the server.
    \item The server evaluates $f$ on the received state and gets the state
    $$\frac{1}{\sqrt{2}}(\ket{x_0}\ket{f(x_0)}+\ket{x_1}\ket{f(x_1)}).$$
    The server sends back the input register (holding $\ket{x_0}, \ket{x_1}$) through a quantum channel back to the client.\par
    Now the joint state between the client and the server is as follows:
    \begin{equation}\label{eq:2}\frac{1}{\sqrt{2}}(\underbrace{\ket{x_0}}_{\text{client}}\underbrace{\ket{f(x_0)}}_{\text{server}}+\underbrace{\ket{x_1}}_{\text{client}}\underbrace{\ket{f(x_1)}}_{\text{server}}).\end{equation}
    \item The client measures its quantum state and stores the outcome (either $x_0$ or $x_1$) as $x$. The server holds the corresponding $f(x)$.
\end{enumerate}
\end{toyprtl}
We note that Toy Protocol \ref{toyprtl:1} alone already achieves something: for a server that is honest during the protocol execution but becomes malicious after the protocol completes, the security already holds.\footnote{We thank Prabhanjan Ananth for raising Toy Protocol \ref{toyprtl:1} as a simplification of the original construction
 and giving this interesting observation in a personal discussion.} In the world of classical cryptography this is called the honest-but-curious security; although this type of security is actually trivial to achieve in the quantum world, this is a good beginning for explaining the intuition.\par
To get the security against a malicious server, we need the following test-mode protocol.
\begin{toyprtl}[Test-mode]\label{toyprtl:2}
    \begin{enumerate}
        \item[1, 2.] The step 1 and step 2 are the same as Toy Protocol \ref{toyprtl:1}.
        \item[3.] The client randomly chooses to execute one of the following two tests with the server:
        \begin{itemize}
            \item (Computational basis test) The client measures its state on the computational basis and asks the server to measure its state on the computational basis and send back the result. Suppose the client's outcome is $x$ and the server's response is $y$. The client checks $f(x)=y$ and rejects if the test fails.
            \item (Hadamard basis test) The client measures its state on the Hadamard basis and asks the server to measure its state on the Hadamard basis and send back the result. Suppose the client's outcome is $d^{(in)}$ and the server's response is $d^{(out)}$. The client checks (below $\cdot$ denotes vector inner product)
            \begin{equation}\label{eq:3r}d^{(in)}\neq 0^n\quad\land\quad d^{(in)}\cdot (x_0+x_1)+d^{(out)}\cdot(f(x_0)+f(x_1))\equiv 0 \mod 2\end{equation}
             and rejects if the test fails.\footnote{As a convention, below we may omit ``the client rejects if the test fails'' when we say the client checks something.}
        \end{itemize}
    \end{enumerate}
\end{toyprtl}
In other words, starting from state \eqref{eq:2}, the client will choose to test the server in either computational basis or in the Hadamard basis. We note that some forms of computational basis test and Hadamard basis test have also been used in existing works like \cite{BCMVV,MahadevVerification}, but our problem settings, ideas and constructions are different from these existing works.\par
How could a malicious server attack this protocol? The server might try to store something more than $f(x)$ (for example, $x$) before sending back the input register (in the second step of the protocol); here the test-mode part is designed to catch it cheating. 
 Recall the preRSPV soundness: if the adversary could pass the test-mode protocol with high probability, then in the computation-mode protocol what it could get is no more than the honest execution output $f(x)$. Below we explain the intuition of why these protocols satisfy this soundness, under the condition that $f$ satisfies certain desirable property.
\subsubsection{Intuition for the soundness}\label{sec:2.3.2}
We first note that if the server could pass the tests in the test-mode protocol, then starting from the output state of the second step, it should pass both the computational basis test and Hadamard basis test.\par
Denote the joint state of the client and the server after the second step of the protocol as $\ket{\phi}$. Let's first see what we could say on it under the condition that the server could pass the computational basis test. This implies that the server could recover $f(x)$ on its own side; in other words, we could say $\ket{\phi}$ is in the following form up to a server-side operation:
\begin{equation}\label{eq:3}\ket{\phi}\approx\underbrace{\ket{x_0}}_{\text{client}}\underbrace{\ket{f(x_0)}\ket{\varphi_0}}_{\text{server}}+\underbrace{\ket{x_1}}_{\text{client}}\underbrace{\ket{f(x_1)}\ket{\varphi_1}}_{\text{server}}\end{equation}
We could compare \eqref{eq:3} with the state from the honest execution, \eqref{eq:2}: roughly speaking, what remains to be done is to show that for each $b\in \{0,1\}$, $\ket{\varphi_b}$ could be simulated from solely $f(x_b)$. We will make use of the condition that the server could pass the Hadamard test starting from \eqref{eq:3}, under the extra condition that $f$ satisfies certain desirable property (elaborated below).\par
Denote the two terms on the right hand side of \eqref{eq:3} as $\ket{\phi_0},\ket{\phi_1}$, and let's call them the $0$-branch and the $1$-branch; thus $\ket{\phi}\approx\ket{\phi_0}+\ket{\phi_1}$. First it's intuitive to see that these two branches correspond to the two branches in state \eqref{eq:1r} correspondingly; if the state \eqref{eq:1r} is replaced by a single branch the state \eqref{eq:3} will also become the corresponding branch. Then let's assume the function $f$ has the desirable property that when the protocol starts with a normalized single-branch state (that is, \eqref{eq:1r} is replaced by a normalized single branch, which is either $\ket{x_0}$ or $\ket{x_1}$), the server could only pass the Hadamard test with $\frac{1}{2}$ probability.\footnote{For example when $f$ is a random function it satisfies this condition: without loss of generality consider the $0$-branch, then note that the $0$-branch contains basically no information on the value of $x_1$ or $f(x_1)$, so it's hard for the server to come up with a $d=(d^{(in)},d^{(out)})$ that passes the check in \eqref{eq:3r} with probability better than random guessing.}\par
So on the one hand each branch of $\ket{\phi}$ has norm only $\frac{1}{\sqrt{2}}$ and each branch (after normalization) could only pass the Hadamard test with probability $\frac{1}{2}$; on the other hand $\ket{\phi}$ passes with high probability (that is, close to $1$). By linear algebra calculations we get that the passing space and failing space of these states in the Hadamard test should interfere in the right way:
\begin{equation}\label{eq:6}\Pi_{\fpass}O_{HT}\ket{\phi_0}\approx \Pi_{\fpass}O_{HT}\ket{\phi_1}\end{equation}
\begin{equation}\label{eq:7}\Pi_{\ffail}O_{HT}\ket{\phi_0}\approx -\Pi_{\ffail}O_{HT}\ket{\phi_1}\end{equation}
where $\Pi_{\fpass},\Pi_{\ffail}$ is the projection onto the passing and failing space, $O_{HT}$ is the server's operation in the Hadamard test.\par
\eqref{eq:6}\eqref{eq:7} implies that $\ket{\phi_0}$ could be mapped to $\ket{\phi_1}$ (and vice versa) by the following operation:
\begin{enumerate}
    \item Apply $O_{HT}$;
    \item Add a $(-1)$ phase on $d=(d^{(in)},d^{(out)})$ that satisfies $d^{(in)}\cdot (x_0+x_1)+d^{(out)}\cdot(f(x_0)+f(x_1))\equiv 1\mod 2$;
    \item Apply $O^{-1}_{HT}$.
\end{enumerate}
Recall that we would like to show that for each $b\in \{0,1\}$, $\ket{\varphi_b}$ could be simulated from solely $f(x_b)$; without loss of generality consider how to simulate $\ket{\varphi_1}$ from solely $f(x_1)$. Let's roughly show why this is true.\par
First by the discussion above $\ket{\varphi_1}$ could be prepared from $\ket{\phi_0}$ by applying the three-step operation above and disgarding other reigsters. But the step 2 of the three-step operation above makes use of $x_0$ which is not what we want. To proceed note that the step 2 above could be written as the composition of the following two operations:
\begin{itemize}\item Add a $(-1)$ phase on $d^{(in)}$ that safisfies $d^{(in)}\cdot (x_0+x_1)\equiv 1\mod 2$; \item Add a $(-1)$ phase on $d^{(out)}$ that satisfies $d^{(out)}\cdot(f(x_0)+f(x_1))\equiv 1\mod 2$.\end{itemize} Note that we only care about $\ket{\varphi_1}$, which lies on the server side of $\ket{\phi_1}$, so we could skip the $d^{(in)}$ part in the operations above.\par
Finally note that the state $\ket{\phi_0}$ contains basically no information about $x_1$ and could be simulated by sampling a random $x_0$. In summary, the operation that simulates $\ket{\varphi_1}$ from solely $f(x_1)$ is as follows:
\begin{enumerate}
    \item Sample a freshly random $x_0$ and simulate the corresponding $\ket{\phi_0}$;
    \item Apply $O_{HT}$;
    \item Add a $(-1)$ phase on $d^{(out)}$ that satisfies $d^{(out)}\cdot(f(x_0)+f(x_1))\equiv 1\mod 2$;
    \item Apply $O_{HT}^{-1}$.
    \item Keep only the registers corresponding to $\ket{\varphi_1}$.
\end{enumerate}
\subsection{Construction of the Forward-succinct Protocol for General Functions}\label{sec:2.4}
In the previous section we present a toy protocol for forward-succinct preRSPV. But the protocol in the previous section has several limitations:
\begin{itemize}
    \item In Section \ref{sec:2.3.2} we assume an extra condition on the function $f$; what we want is a protocol for general functions.
    \item Some properties of Toy Protocol \ref{toyprtl:1}, \ref{toyprtl:2} are obstacles for further improving the protocol. These include:
    \begin{itemize}
        \item The protocol contains backward quantum communication. 
        \item The initial state \eqref{eq:1r} does not have an existing RSPV protocol from standard assumptions.
    \end{itemize}
\end{itemize}
In this section we show how to construct a forward-succinct preRSPV that solves the problem above.\par
Below we first discuss our protocol in Section \ref{sec:2.4.1}, and then discuss its soundness proof in Section \ref{sec:2.4.2}.
\subsubsection{Protocol design}\label{sec:2.4.1}
To solve the problem above, we revise the Toy Protocol \ref{toyprtl:1}, \ref{toyprtl:2} as follows: 
\begin{itemize}
    \item To support general functions, we introduce an input padding $(r^{(in)}_0,r^{(in)}_1)$ where each of them is sampled from $\{0,1\}^\kappa$. This part is padded into \eqref{eq:1r} and the Hadamard test is also updated.
    \item  To remove the backward quantum communication, we let the server measure the input register (holding $x_0$, $x_1$ in \eqref{eq:2}) together with the input padding on the Hadamard basis in step 2, and send back the results. The side effect is, in the computation mode (step 3 of Protocol \ref{toyprtl:1}) the client could not know whether it should keep $x_0$ or $x_1$. To solve the problem, an output padding $(r^{(out)}_0,r^{(out)}_1)$ is introduced and padded into the state \eqref{eq:1r}; this part will be used to test and identify the state on the computational basis.
    \item Finally, we start with the state $\frac{1}{\sqrt{2}}(\ket{0}\ket{x_0}\ket{r^{(in)}_0}\ket{r^{(out)}_0}+\ket{1}\ket{x_1}\ket{r^{(in)}_1}\ket{r^{(out)}_1})$ for state \eqref{eq:1r} (that is, adding a bit holding branch subscripts $0,1$). By \cite{zhang24} there exists an RSPV protocol for this state from standard assumptions.
\end{itemize}
In summary, the new protocol is roughly as follows. 
\begin{prtl}[Test-mode]\label{prtl:1}
    \begin{enumerate}
        \item The client samples and stores the following strings: function inputs $x_0,x_1\in \{0,1\}^n$; input padding $r_0^{(in)},r_1^{(in)}\in \{0,1\}^\kappa$; output padding $r_0^{(out)},r_1^{(out)}\in \{0,1\}^\kappa$. Then the client prepares and sends \begin{equation}\label{eq:1r2}\frac{1}{\sqrt{2}}(\ket{0}\ket{x_0}\ket{r_0^{(in)}}\ket{r_0^{(out)}}+\ket{1}\ket{x_1}\ket{r_1^{(in)}}\ket{r_1^{(out)}})\end{equation} to the server.
        \item The server evaluates $f$ on the received state and gets the state
        $$\frac{1}{\sqrt{2}}(\ket{0}\ket{x_0}\ket{r_0^{(in)}}\ket{r_0^{(out)}}\ket{f(x_0)}+\ket{1}\ket{x_1}\ket{r_1^{(in)}}\ket{r_1^{(out)}}\ket{f(x_1)}).$$
        The server measures the registers holding function inputs and input paddings on the Hadamard basis; denote the measurement results as $d^{(in)}\in \{0,1\}^n,d^{(inpad)}\in \{0,1\}^\kappa$. The server sends back these results to the client.\par
        The client checks $d^{(inpad)}\neq 0^\kappa$. Then it stores $\theta_0=d^{(in)}\cdot x_0+d^{(inpad)}\cdot r_0^{(in)}\mod 2$ and $\theta_1=d^{(in)}\cdot x_1+d^{(inpad)}\cdot r_1^{(in)}\mod 2$.\par
        Now the remaining state on the server side is :
        \begin{equation}\label{eq:2r2}\frac{1}{\sqrt{2}}((-1)^{\theta_0}\ket{0}\ket{r_0^{(out)}}\ket{f(x_0)}+(-1)^{\theta_1}\ket{1}\ket{r_1^{(out)}}\ket{f(x_1)}).\end{equation}
       \item The client randomly chooses to execute one of the following two tests with the server:
       \begin{itemize}\item (Computational basis test) The client asks the server to measure its remaining state on the computational basis and send back the result. The client checks the server's response is in $\{0||r_0^{(out)}||f(x_0),1||r_1^{(out)}||f(x_1)\}.$
        \item (Hadamard basis test) The client asks the server to measure its remaining state on the Hadamard basis and send back the result. Suppose the server's response is $d^{(sub)}\in \{0,1\},d^{(outpad)}\in \{0,1\}^\kappa,d^{(out)}\in \{0,1\}^m$. The client checks
            \begin{equation}\label{eq:3r2}d^{(sub)}+d^{(outpad)}\cdot (r_0^{(out)}+r_1^{(out)})+d^{(out)}\cdot(f(x_0)+f(x_1))\equiv \theta_1-\theta_0 \mod 2.\end{equation}
       \end{itemize}
    \end{enumerate}
\end{prtl}
\begin{prtl}[Computation-mode]\label{prtl:2}
    \begin{enumerate}
        \item[1, 2.] The step 1 and step 2 are the same as Protocol \ref{prtl:1}.
        \item[3.] The client asks the server to measure the first two registers (registers holding $\ket{0}\ket{r_0^{(out)}}$ and $\ket{1}\ket{r_1^{(out)}}$) on the computational basis and send back the result.\par
         The client checks the server's response is either $0||r_0^{(out)}$ or $1||r_1^{(out)}$. The client chooses $x_0$ as its final outputs in the former case and chooses $x_1$ as its final outputs in the later case. The server-side remaining state is the corresponding function output.\par
    \end{enumerate}
\end{prtl}
\subsubsection{Discussion on the soundness}\label{sec:2.4.2}
Let's first intuitively see how these changes overcome the limitation in Toy Protocol \ref{toyprtl:1}, \ref{toyprtl:2}. First, with the input padding, regardless of what functions we are working on, an adversary running on a single branch could only pass with $\frac{1}{2}$ probability: note that $d^{(inpad)}\cdot r_0^{(in)}$ (or $d^{(inpad)}\cdot r_1^{(in)}$) has completely randomized the right hand side of \eqref{eq:3r2}. Then for the effect of the output padding, intuitively if the server outputs $r_b^{(out)}$ correctly, it means that the state has been collapsed to the $b$-branch.\par
The formal proofs are more complicated than this intuitive discussion. To prove the soundness we make use of a security proof guideline in a previous work on RSPV \cite{cvqcinlt}. As in \cite{cvqcinlt}, we define different forms of states including \emph{basis-honest form}, \emph{basis-phase correspondance form} and \emph{phase-honest form}, which define stricter and stricter restrictions on the form of states; we then analyze different components of the protocol and show that these different components restrict the form of states step by step. We give a more detailed overview of the proof in Section \ref{sec:6.2.1}.
\subsection{Construction of the Fully Succinct Protocol}\label{sec:2.5}
In this section we discuss how to further revise the forward-succinct protocol in the previous section to get a fully succinct protocol, which proves Theorem \ref{thm:1.2}. Note that the reason that Protocol \ref{prtl:1}, \ref{prtl:2} is not fully succinct is that the server-to-client communication is long. Below we present a protocol for reducing the server-to-client communication; the construction is a variant of the compiler in \cite{BKLMM22}. In Section \ref{sec:5} we formalize our protocol and prove its (post-quantum) soundness.\par
\paragraph{Problem setting}Suppose the client holds $x$ and the server claims to the client that it holds a witness $w$ such that $R(x,w)=0$, where $R$ is a public efficient classical function. The client would like to use a protocol to verify the server's claim in succinct communication; we call this primitive succinct testing for two-party relations. Note that the setting may seem similar to the succinct arguments problem (see Section \ref{sec:3.3.4}), but here $x$ is a client-side input (instead of a public one), so succinct arguments do not apply in a direct way. Simply letting the client reveal $x$ and applying succinct arguments do not work either: in this way it is only verified that the server knows such a $w$ after it knows $x$, but we would like to verify that the server holds $w$ in the very beginning.\par
Note that such a primitive is sufficient to compile the Protocol \ref{prtl:1}, \ref{prtl:2} to make the server-to-client communication succinct: all the server-to-client communications in Protocol \ref{prtl:1}, \ref{prtl:2} are used for client-side checking of an NP relation (for example, \eqref{eq:3r2}), so these communications together with the checkings could be replaced by callings to the succinct testing protocol. An additional note is that we do not care about the secrecy of $x$ during the protocol; we call it ``testing'' to suggest the slightly unusual setting that the input $x$ is private in the beginning but is allowed to be revealed during the protocol.
\paragraph{Protocol construction} The protocol is roughly as follows.
\begin{enumerate}
    \item (Commit) The server commits to $w$ using a collapsing hash function $h$: it computes $h(w)=c$ and sends $c$ to the client.
    \item (Prove) The client and the server execute a succinct arguments of knowledge for the statement ``$\exists w \text{ such that } h(w)=c$''.
    \item (Reveal) The client reveals $x$ to the server.
    \item (Prove) The client and the server execute a succinct arguments of knowledge for the statement ``$\exists w \text{ such that } f(x,w)=0\land h(w)=c$''.
\end{enumerate}
The soundness is intuitive: by committing to $w$, an arguments of knowledge for the statement about $w$ in step 4 should mean that the server knows this $w$ when it does the commitment.\par
\paragraph{Comparison to existing works} We note that the overall structure of our protocol is not new: \cite{BKLMM22} gives a compiler for reducing the server-to-client communication, and the protocol has the following similar structure: the server does a commitment, both parties execute a succinct argument of knowledge, the client reveals something, and both parties execute a succinct argument of knowledge for both the statement and the commitment. But our problem setting and protocol are technically different from \cite{BKLMM22}: in \cite{BKLMM22} the client-side messages are all instance-independent (they are freshly new random bits) and in our problem $x$ is the protocol input (but is allowed to be revealed). To handle these technical differences we re-formalize the protocol.
\subsection{Classical Impossibility}\label{sec:2.6}
In this section we discuss the proof of the classical impossibility. The formal proof is given in Section \ref{sec:4.2}.\par
The proof is an adaptation of the proof in \cite{HW15}. \cite{HW15} proves that for the ``secure function evaluation'' problem it's impossible to construct a classical cryptographic protocol with succinct communication. Our secure function sampling problem is different from the secure function evaluation problem in \cite{HW15} in the following ways: (1) SFS considers functions with a single input, while in \cite{HW15} the functions receive private inputs from both parties; (2) in our setting the function input $x$ is sampled during the protocol as a client-side output while in \cite{HW15} the function inputs are provided by the parties; (3) we allow approximate security while \cite{HW15} considers normal (negligible soundness error) security. We note that these differences actually make our result stronger. Technically we adapt the ideas and techniques in \cite{HW15} to our setting and handle technical differences.\par
Below we give an overview of the idea of the proof. We will prove that when the function is a pseudorandom generator (PRG) this task is classically impossible; note that PRG expands a short seed to a long pseudorandom output string. 
The proof of impossibility is by a compression argument. Note that for classical protocols, the server-side outputs could be deterministically predicted given the following information: (1) the server's operations, the initial state and the random coins; (2) the client-to-server communication. For constructing the compressor, the first part could be hard-coded and we could focus on the second part. If such an SFS protocol with succinct communication exists, it means that the PRG outputs could be efficiently compressed to something succinct (this is by considering an adversary that outputs the transcripts and calling the simulation-based security) and could be efficiently decompressed. By the properties of PRG this is impossible.
\subsection{Protocols for Secure Two-party Computations}\label{sec:2.7}
In this section we give an overview on how to construct protocols for secure two-party computations based on the SFS protocol.\par
The roadmap is as follows:
\begin{enumerate}
    \item We first generalize the function family from classical deterministic functions to classical randomized functions.
    \item We then construct a protocol for the setting where the function input is provided as the client-side private input (instead of sampled during the protocol). We call it \emph{secure function value preparation} (SFVP).
    \item We then construct a secure two-party computation protocol with succinct communication. An important step for achieving this is how to make the protocol sound against the malicious client.
    \item We then show how to compile the protocol with quantum communication to a protocol with fully classical communication.
\end{enumerate}
Note that each step builds upon or is a revision of previous ones. Below we elaborate each step.
\subsubsection{From deterministic functions to randomized functions}An efficient randomized classical function could be equivalently modeled by a deterministic function on the function input and the random coins; let's denote it as $f(x;r):\{0,1\}^n\times \{0,1\}^{L}\rightarrow \{0,1\}^m$. Note that the number of random coins $L$ could be much bigger than $n$, and we want a protocol whose communication complexity is also independent to the number of random coins.\par
The solution is to take a PRG to simulate the random coins: we consider the function $g(x,s) :=f(x;\tPRG(s))$. Note that the inputs of $g$ are succinct and running our SFS protocol for $g$ gives a protocol with succinct communication. By the property of PRG the outputs of this protocol are indistinguishable\footnote{Note that here the definition of the completeness is also based on indistinguishability.} to the desired outputs of SFS of $f$.
\subsubsection{Secure function value preparation: from sampled input to client-chosen input}\label{sec:2.7.2} In the original setting of SFS the client could not control the function input $x$. We then consider the setting where $x$ is a client-side input and the client would like to transmit $f(x)$ to the server; we call it secure function value preparation (SFVP) and design a protocol $\fSFVP$ for this problem. Our idea is as follows.\par
 We make use of a cryptographic primitive called Yao's garbled circuit \cite{YaoGCOrigin}: For a classical circuit $C$ and input $x$, Yao's garbled circuit transforms them into $\fGarbleC(C,r)$ and $\fGarbleI(x,r)$ where $r$ is the shared random coins; let's call them the circuit encoding and input encoding correspondingly. (We note that $\fGarbleC$ could itself be a randomized function for fixed $(C,r)$; so a deterministic function description of $\fGarbleC$ could be written as $\fGarbleC(C,r;r_{\text{GCin}})$.) 
This primitive has the following properties:
\begin{itemize}
    \item The server could recover $C(x)$ given $\fGarbleC(C,r)$ and $\fGarbleI(x,r)$.
    \item $\fGarbleC(C,r)$ and $\fGarbleI(x,r)$ could be simulated from solely $C(x)$; intuitively this means that when the server is given $\fGarbleC(C,r)$ and $\fGarbleI(x,r)$, what it knows is no more than $C(x)$.
    \item As shown in the notations, the computation of $\fGarbleC(C,r)$ does not use $x$ and the computation of $\fGarbleI(x,r)$ does not use $C$. What these two algorithms share are some random coins $r$. This property is called ``decomposable'' in garbled circuit primitive. We additionally note that the shared coins $r$ is also succinct: its size only depends on the length of $x$ and the security parameter $\kappa$.
\end{itemize}
Now we could construct the $\fSFVP$ protocol for $f$ and the client-side input $x$ as follows. 
\begin{enumerate}
\item The client and the server execute an SFS protocol for function $g(r):=\fGarbleC(C_f,r)$, where $C_f$ is the circuit description of $f$. The client should get a random $r$ and the server gets the corresponding $\fGarbleC(C_f,r)$.
\item The client computes and sends $\fGarbleI(x,r)$ to the server.
\end{enumerate}
Note that the Yao's garbled circuits could be constructed from one-way functions, which is implied by collapsing hash functions.
\subsubsection{Generalization to two-party computations}\label{sec:2.7.3} Now we would like to further generalize our protocols to the general two-party computations, where both parties could have private inputs and private outputs. To overcome the obstacles we need a combination of several ideas, which we discuss below.\par
\paragraph{Use two-party computations to transmit input encoding} For introducing the first idea we first consider a single function $f(x_A,x_B)$ as the outputs, which takes input $x_A$ from Alice and $x_B$ from Bob. If we simply use the $\fSFVP$ protocol, the circuit encoding step could be executed normally while the input encoding step needs to take inputs from both parties. How could Alice and Bob do this securely? The first idea is as follows: the transmission of input encoding, which is $\fGarbleI((x_A,x_B),r)$, is performed by a two-party computation. Note that this step is succinct since it's only about $x_A,x_B,r$, which are all succinct; the computation of $\fGarbleI$ is also relatively easy. Then this two-party computation step could be achieved quantumly assuming only one-way functions by existing works \cite{BCKM20,GLSV20}.\par
\paragraph{Obstacles in the malicious setting} Recall that we are working in the malicious setting where the parties could deviate from the honest behavior maliciously. Examining the security of the constructions so far, there are still two unsolved obstacles for achieving general secure two-party computations:
\begin{itemize}
    \item In the study of SFS we only care about the security against a malicious server; in two-party computations either Alice or Bob could be malicious. So we need some type of security guarantee against a malicious client in the $\fSFVP$ protocol.
    \item A more subtle obstacle is that in the proof of simulation-based security of garbling-based protocol, if the adversary first receives the circuit encoding, the simulator needs to simulate the circuit encoding first; then the simulator does not have an easy way to simulate the input encoding so that the simulated circuit encoding and simulated input encoding still decode to $C(x)$.\footnote{We suspect this obstacle is related to the notion of adaptive security of garbled circuits.}
\end{itemize}
Our ideas are as follows:
\begin{itemize}
    \item We design a new secure function value preparation protocol that also achieves a type of security against a malicious client. The construction makes use of the combination of statistically-binding computationally-hiding string commitments, collapsing hash functions and post-quantum succinct zero-knowledge arguments of knowledge (succinct-ZK-AoK) to control the client's behavior. This protocol is denoted as $\fSFVPTwo$.
    \item In the construction of the two-party computation protocol both parties will first transmit the input encoding and then transmit the circuit encoding (with $\fSFVPTwo$). Note that since $\fSFVPTwo$ itself is based on $\fSFVP$ so we are actually performing the garbling technique twice: in $\fSFVP$ we first transmit the circuit encoding and then transmit the input encoding, while in the two-party computation protocol we first transmit the input encoding and then the circuit encoding.
\end{itemize}
Below we elaborate our ideas and constructions.
\paragraph{Secure function value preparation with security against malicious clients}  The $\fSFVPTwo$ protocol relies on Naor's statistically-binding string commitment scheme \cite{Naorcomm}, denoted by $\fCommit$ (see Section \ref{sec:3.3.comm} for details): to commit to a string $x$, the receiver first samples public randomness $pp$, and then the sender samples $r$ and calculates $\fCommit(x,r,pp)$. To open the commitment the sender simply reveals $r$.\par
The setting of $\fSFVPTwo$ is as follows, which is slightly unusual. Initially the client chooses $x$ as the protocol input; the client also holds a random string $r$ as the commitment randomness. The server holds the commitment $com=\fCommit(x,r,pp)$. $pp$ is randomly sampled and is public. In summary we assume the client has already committed to its function input before the protocol. The security against malicious clients is defined using the simulation-based paradigm, under the condition that the initial setup is done honestly. In more detail, roughly speaking, the soundness against malicious client says that what the client sends to the server should be exactly $f(x)$ where $x$ is the committed string; otherwise the client will be caught cheating.\par
Now the construction of $\fSFVPTwo$ is roughly as follows.
\begin{enumerate}
    \item Both parties execute the $\fSFVP$ protocol to send $f(x)$ to the server.
    \item The server uses a collapsing hash function $h$ to calculate $h(f(x))=c$. The server sends $h,c$ to the client.
    \item Using a succinct-ZK-AoK, the client proves the following statement to the server: ``$\exists (x,r)$ such that $h(f(x))=c\land \fCommit(x,r,pp)=com$''.
\end{enumerate}
Note that all the basic primitives used above could be constructed from collapsing hash functions \cite{Naorcomm,CMSZ21,Kilian92}. The security against malicious clients is intuitive: by the commitment and the proof-of-knowledge the client holds $x$; note that the statistically-binding property of the commitment guarantees that $x$ is unique. Then by the $h(f(x))=c$ and the collision-resistance of $h$ we know what the server gets should be $f(x)$. The ZK property is used to preserve the security against malicious servers.\par
Finally we note that the intuitive discussion of the client-side soundness above skipped an important obstacle: note that the discussion above implicitly assumes that the client passes the PoK with probability close to $1$; in general either aborting or non-aborting could happen with some probability and constructing a simulator that dealing with both cases is not easy. Here our solution is to further revise the protocol above as follows: the server will repeat the third step above for $\fpoly(1/\epsilon)$ rounds (where $\epsilon$ is the error tolerance on the approximate soundness), and requires that the client should pass the checking in each round. Then the simulator could do a cut-and-choose: if the simulator randomly picks a round in them, with high probability the client should pass the PoK verification with high probability conditioned on it passes in all the previous rounds.
\paragraph{Construction of the secure two-party computations protocol} Now we could construct our secure two-party computations protocol as follows. We consider the general setting where there are both Alice-side output function $f_A(x_A,x_B)$ and Bob-side output function $f_B(x_A,x_B)$.
\begin{enumerate}
    \item Both parties execute a two-party computation for the following two tasks:
    \begin{itemize}
        \item Transmit the suitable input encoding;
        \item Commit to the randomness used in the input encoding.
    \end{itemize}
    \item Alice uses $\fSFVPTwo$ to transmit the circuit encoding of $f_B$ to Bob. An honest Bob could decode and get $f_B(x_A,x_B)$.
    \item Bob uses $\fSFVPTwo$ to transmit the circuit encoding of $f_A$ to Alice. An honest Alice could decode and get $f_A(x_A,x_B)$.
\end{enumerate}
Finally we note that to prove the simulation-based soundness we need a stronger form of soundness for the garbled circuit scheme: instead of letting one simulator simulate both the input encoding and circuit encoding, the input encoding is now simulated first by an input encoding simulator and the circuit encoding is simulated later by a circuit encoding simulator (given the simulated input encoding and the circuit outputs). Luckily Yao's garbled circuit satisfies this slightly stronger soundness requirement.
\subsubsection{De-quantizing the communication} Finally we ask: is it possible to construct protocols for secure two-party computations using classical channel and local quantum computations? Note that in our protocol the quantum communication comes from the following two steps:
 \begin{itemize}\item In $\fSFS$ protocol the client prepares and sends states in the form of $\frac{1}{\sqrt{2}}(\ket{0}\ket{x_0}\ket{r_0^{(in)}}\ket{r_0^{(out)}}+\ket{1}\ket{x_1}\ket{r_1^{(in)}}\ket{r_1^{(out)}})$ to the server. \item We make use of secure two-party computations from one-way functions \cite{BCKM20,GLSV20}, in which one party prepares and sends BB84 states to the other during the protocol.
 \end{itemize}
We make use of existing RSPV protocols for these states to replace these prepare-and-send steps. By \cite{zhang24} there exist RSPV protocols for these types of states assuming NTCF. Suppose in the original protocol Alice prepares and sends some states to Bob; after this step is replaced by the corresponding RSPV protocols, the soundness of the new protocol is preserved by the following reason:
\begin{itemize}\item When Bob is malicious, by the soundness of RSPV what the malicious Bob could do is no more than the ideal functionality (caught cheating or honest behavior) up to an approximation error.
\item When Alice is malicious, note that the effect of the RSPV protocols is no more than the following operations: Alice simulates and prepares the joint output states of RSPV and sends the simulated Bob's part to Bob. (Note that RSPV protocols do not take private inputs from Bob.) If the original protocol is sound against malicious Alice the new protocol should also be sound.
\end{itemize}

\section{Preliminaries}\label{sec:3}
In this section we review and clarify basic notions and notations, modeling and basic properties of cryptographic protocols, and important cryptographic primitives.\par
Note that many choices of notions and notations, including their definitions, would be similar to \cite{zhang24}.
\subsection{Notations and Basics}
We first clarify basic mathematical notations.
\begin{nota}
We use $\bN$ to denote the natural numbers (positive integers). We use $1^\kappa$ to denote the string $\underbrace{111\cdots 1}_{\text{for }\kappa\text{ times}}$. We use $[n]$ to denote the set $\{1,2,\cdots n\}$. We use $\leftarrow_r$ to denote random sampling. We use $||$ to denote string concatenation. We use $|\cdot |$ to denote the Euclidean norm of a vector. We also use $|C|$ to denote the size of a circuit $C$. We use $a\equiv b\mod m$ to denote the congruence modulo $m$. We use $\cdot$ to denote the inner product between vectors. We use $\tDist(S)$ to denote the set of probability distribution on $S$. We write $a\approx_{\epsilon}b$ when $|a-b|\leq \epsilon$ where $a,b$ are real numbers.
\end{nota}
We refer to \cite{NielsenChuangs,watroustqi} for basics of quantum computation. Below we review or clarify notations.
\begin{nota}
We use Dirac notation $\ket{\varphi}$ to denote pure states. $\ket{\varphi}\in \cH$ means that $\ket{\varphi}$ is in the Hilbert space $\cH$. We call $\ket{\varphi}$ normalized if $|\ket{\varphi}|=1$.\par
We use $\tD(\cH)$ to denote the set of density operators on a Hilbert space $\cH$. This could be seen as the quantum analog of probability distribution in the classical world.\par
Quantum operations on pure states are described by unitaries. The output state of applying a unitary quantum operation $U$ on quantum state $\ket{\varphi}$ is $U\ket{\varphi}$.\par
In general, quantum operations are described by completely-positive trace-preserving (CPTP) maps on density operators. The output state of applying a quantum operation $\cE$ on quantum state $\rho$ is denoted as $\cE(\rho)$.
\end{nota}
We work on registers, which is very convenient.
\begin{nota}
    A quantum register is a labeled Hilbert space. We use bold font (like $\bS$) to denote registers; the corresponding Hilbert space is $\cH_{\bS}$.\par
    Then we use $\otimes$ to denote the tensor product and $\tr_{\bR}$ to denote the partial trace that traces out register $\bR$.\par
    The joint state on a classical register and a quantum register is call a cq-state, which could also be described by density operators.\par
    We use $\Pi^{\breg}_{v}$ to denote the projection onto the space that the register $\breg$ is in value $v$, and use $\Pi^{\breg}_{S}$ to denote the projection onto the space that the value of register $\breg$ is in set $S$.
\end{nota}
\begin{nota}
We write $\ket{\varphi}\approx_{\epsilon}\ket{\phi}$ if $|\ket{\varphi}-\ket{\phi}|\leq \epsilon$.\par
We write $\rho\approx_{\epsilon}\sigma$ if $\tTD(\rho,\sigma)\leq \epsilon$ where $\tTD$ denotes the trace distance between $\rho,\sigma\in \tD(\cH)$.
\end{nota}
\begin{fact}[Purification]
For each $\rho\in \tD(\cH_{\bS})$, there exists $\ket{\phi}\in \cH_{\bS}\otimes\cH_{\bR}$ such that $\tr_{\bR}(\ket{\phi}\bra{\phi})=\rho$. We say $\ket{\phi}$ is a purification of $\rho$.\par
If $\ket{\phi},\ket{\phi^\prime}\in \cH_{\bS}\otimes\cH_{\bR}$ are two purification of $\rho$, then there exists a unitary $U$ operating on $\cH_{\bR}$ such that $U\ket{\phi}=\ket{\phi^\prime}$.\par
If $\ket{\phi},\ket{\phi^\prime}\in \cH_{\bS}\otimes\cH_{\bR}$ are purifications of $\rho$, $\rho^\prime$ correspondingly, $\ket{\phi}\approx_{\epsilon}\ket{\phi^\prime}$, then $\rho\approx_{\epsilon}\rho^\prime$.
\end{fact}
We refer to \cite{AroraBarak} for basics of complexity theory. Below we review or clarify notions and notations.
\begin{nota}
We use $\fpoly$ to denote a polynomial function (that is, $O(n^{c})$ and $\Omega(n^{-c})$ where $c\in \bN$), and use $\fneg$ to denote a negligible function (that is, converges to $0$ faster than any inverse polynomial). 
\end{nota}
\begin{nota}\label{nota:3.5}
    We say a family of classical functions $f=(f_n)_{n\in \bN}$, $f_n:\{0,1\}^n\rightarrow \{0,1\}$ is classical polynomial-time computable if it could be computed by a classical Turing machine $TM$ in polynomial time.\par
    This notion also applies to function families with both input length parameters and output length parameters: we say $f=(f_{n,m})_{n\in \bN,m\in \bN}$, $f_{n,m}:\{0,1\}^n\rightarrow \{0,1\}^m$ is classical polynomial-time computable if the same definition holds.
\end{nota}
\begin{nota}
    A classical circuit $C$ could be described by $((g_i)_{i\in [|C|]},\text{topo})$ where $(g_i)_{i\in [|C|]}$ is a sequence of elementary classical gates and topo is the topology of the circuit (which includes the circuit graph and the identification of input wires and the output wires).\par
    A quantum circuit is defined in the same way with $(g_i)_{i\in [|C|]}$ being a sequence of elementary quantum gates.\par
    We use $C(x)$ to denote the outputs of evaluating a circuit $C$ on $x$, and use $C(\rho)$ to denote the outputs of evaluating a quantum  circuit $C$ on initial state $\rho$.
\end{nota}
\begin{nota}
    We use $\Pr[C(\rho)\rightarrow o]$ to denote the probability that running a quantum circuit $C$ on state $\rho$ gives output $o$.
    \end{nota}
\begin{nota}\label{nota:3.8}
We say a family of circuit $C=(C_{\kappa})_{\kappa\in \bN}$ is polynomial-size if $|C_\kappa|=\fpoly(\kappa)$ and say it's polynomial-time uniform if the circuit description is classical polynomial-time computable.
\end{nota}


Then we review important complexity classes.
\begin{nota}
    We say a language $L$ is in complexity class P if it could be decided by a classical polynomial-time Turing machine.\par
    We say a language $L$ is in complexity class NP if there exists a classical polynomial-time Turing machine $TM$:
    \begin{itemize}
        \item (Completeness) For each $x\in L$, there exists a witness $w$ such that $TM(x,w)=1$.
        \item (Soundness) For each $x\not\in L$, for any $w$ there is $TM(x,w)=0$.
    \end{itemize}\par
    We say $(L_{\text{yes}},L_{\text{no}})$ is in complexity class QMA if there exists a polynomial-time uniform family of quantum circuits $V$:
    \begin{itemize}
        \item (Completeness) For each $H\in L_{\text{yes}}$, there exists a quantum state $w$ such that $\Pr[V(H,w)\rightarrow 1]\geq \frac{2}{3}$.
        \item (Soundness)  For each $H\in L_{\text{no}}$, for any quantum state $w$ there is $\Pr[V(H,w)\rightarrow 0]\geq \frac{2}{3}$.
    \end{itemize}\par
    Finally we note that there are non-uniform variants of these complexity classes that take advice strings (or advice states). For example:\par
    We say a $(L_{\text{yes}},L_{\text{no}})$ is in complexity class BQP/qpoly if there exists a polynomial-time uniform family of quantum circuits $V$ and a family of polynomial-size advice state $(v_{\kappa})_{\kappa\in \bN}$:
    \begin{itemize}
        \item (Completeness) For each $C\in L_{\text{yes}}$, $\Pr[V(C,v_{|C|})\rightarrow 1]\geq \frac{2}{3}$.
        \item (Soundness)  For each $C\in L_{\text{no}}$, $\Pr[V(C,v_{|C|})\rightarrow 0]\geq \frac{2}{3}$.
    \end{itemize}
\end{nota}
\begin{conv}\label{conv:1}
    Notation \ref{nota:3.5}, \ref{nota:3.8} formalize what we mean when we say a function is efficient. Note that for classical functions its efficiency could be described by either Turing machines (Notation \ref{nota:3.5}) or circuits (Notation \ref{nota:3.8}), while for quantum circuits we usually do not work on the quantum analog of Turing machines but work on quantum circuits instead.\par
    As a convention, in this paper we consider non-uniform version of these notions by default. For example, when we consider efficient quantum operations, we are considering a polynomial-time uniform family of quantum circuits $V$ together with a family of polynomial-size advice states $(v_{\kappa})_{\kappa\in \bN}$.
\end{conv}
Finally the following fact is about the description of efficient classical randomized functions.
\begin{fact}\label{fact:2}
    An efficient classical randomized function $\{0,1\}^n\rightarrow \tDist(\{0,1\}^m)$ could be described by a deterministic function on the inputs and random coins: $f:\{0,1\}^n\times\{0,1\}^L\rightarrow \{0,1\}^m$ where $L$ is the number of gates for evaluating the function.
\end{fact}
Note that when we say efficient functions we consider deterministic functions by default (but in this paper we will also encounter lots of randomized functions). Another convention is that when we use two different representations of efficient randomized functions (as in Fact \ref{fact:2}) we use semicolon to seprate the random coins: for example, the randomized function $f(x),x\in \{0,1\}^n$ as in Fact \ref{fact:2} could be represented by a deterministic function $f(x;r),x\in \{0,1\}^n,r\in \{0,1\}^L$.
\subsection{Basics of Cryptography and Quantum Cryptography}\label{sec:3.2}
Below we review the hierarchy of notions for modeling the cryptographic protocols in classical and quantum cryptography.
\subsubsection{Modeling}
\paragraph{Value and registers}As said before, it's convenient to work on registers for describing quantum states. For classical data, we could either describe them by values or say they're in some classical registers.
\paragraph{Interactive protocols} In this work we only need to consider interactive protocols between two parties; we call them either client-server or Alice-Bob. For modeling the security we also need to consider an environment, which describes everything else outside the protocol execution.\par
The protocol execution is synchronous, that is, the execution of protocols could be divided into cycles of four phases: the client (or Alice) does some local operations; the client sends messages to the server (or Bob); the server does some local operations; the server sends messages back to the client.\par
The operation of a protocol against an adversary (for example, a malicious server) could be described as
\begin{equation}\label{eq:real}\text{ProtocolName}^{\fAdv}(\tvals,\bregs)\end{equation}
Suppose the initial state (on the registers considered) is $\rho_0$, then the output state could be described as
$$\text{ProtocolName}^{\fAdv}(\tvals,\bregs)(\rho_0)$$
We say a protocol $\pi$ is a classical protocol if the operations in the honest setting are fully classical.
\paragraph{Inputs and outputs}Besides the inputs and outputs related to the tasks, the following inputs and outputs are commonly used. The protocol could take an additional input $1^\kappa$, which is called the security parameter, for formalizing the security (see below). For many problems the output registers will contain a special register $\bflag$, whose value is in $\{\fpass,\ffail\}$, where $\fpass$ means the protocol completes successfully (we also say the client accepts the results in the client-server setting) and $\ffail$ means the malicious party is caught cheating (we also say the protocol aborts or the client rejects in the client-server setting).\par
We use $\Pi_{\fpass}$ to denote the projection on register $\bflag$ to the passing space.
\paragraph{Modeling of cryptographic primitives}One typical template of definitions of cryptographic primitives is as follows. We would like the protocol to satisfy at least the following requirements:
\begin{itemize}
    \item (Correctness, or called completeness) When both parties are honest, the protocol completes successfully and the task is achieved.
    \item (Security, or called soundness) The protocol holds some security properties against a malicious attacker. For example: for any malicious attacker, if the task is not achieved correctly, the attacker is caught cheating.
    \item (Efficiency) Both parties run in polynomial time.
\end{itemize}
For the security, there are two settings:
\begin{itemize}
    \item Security in the classical world: the adversary is completely classical.
    \item Security in the (post-)quantum\footnote{``Post-quantum'' means the protocol is classical while the adversary is quantum.} world: the adversary could be quantum.
\end{itemize}
In this paper we consider security in the (post-)quantum world by default.\par
Formalizing the definition of the security is nontrivial. A paradigm for formalizing the security is based on the simulation, which we discuss below.
\subsubsection{Simulation-based security definition paradigm}
We first review the notions and notations of state family and indistinguishability.
\begin{conv}\label{conv:2}
We note that when we study the security, we consider a family of states parameterized by the security parameter $\kappa$, for example, $(\rho_{\kappa})_{\kappa\in \bN}$. For notation simplicity we simply write $\rho$ and make the security parameter and the underlying state family implicit. We note that when we talk about efficient operations there is also an implicit security parameter $\kappa$; so when we talk about $D(\rho)$ where $D$ is an efficient quantum operation, we are actually talking about $(D(1^\kappa)(\rho_\kappa))_{\kappa\in \bN}$. We omit these parameters for simplicity.
\end{conv}
\begin{defn}
    We write $\rho\approx^{\tind}_\epsilon\sigma$ if for any efficient quantum operation $D$, there exists a negligible function $\fneg$ such that
    $$|\Pr[D(\rho)\rightarrow 0]-\Pr[D(\sigma)\rightarrow 0]|\leq \epsilon+\fneg(\kappa).$$
    We say $\rho$ is $\epsilon$-indistinguishable to $\sigma$ and $D$ is called the distinguisher.
\end{defn}
As reviewed in Section \ref{sec:2.2}, the simulation-based security definition paradigm considers the indistinguishability between the ``real world'' (real execution) and the ``ideal world'' (ideal functionality). To formalize the security, we first need to define the ideal functionality $\tIdeal$. One typical template of definition for the ideal functionality is as follows:
\begin{exmp}\label{exmp:3.1}
Consider the client-server setting in this example. $\tIdeal$ receives a bit $b$ from the server and it works as follows:
\begin{itemize}
    \item If $b=0$, $\tIdeal$ sets $\bflag$ to be $\fpass$ and sets the outputs to be the honest behavior outputs.
    \item If $b=1$, $\tIdeal$ sets $\bflag$ to be $\ffail$.
\end{itemize}
\end{exmp}
Then as discussed in Section \ref{sec:2.2}, the simulation-based security is defined to be the indistinguishability of the real execution (that is, \eqref{eq:real}) and the ideal functionality composed with the simulation (that is, $\fSim\circ \tIdeal$); see also equation \eqref{eq:1} in Section \ref{sec:2.2}. If we expand the expression, it goes as follows:
\begin{exmp}
    Consider the client-server setting as in Example \ref{exmp:3.1}. We say a protocol $\pi$ is $\epsilon$-sound if the following holds.\par
    For any efficient quantum adversary $\fAdv$ there exist efficient quantum operations $\fSim=(\fSim_0,\fSim_1)$ such that for any initial state $\rho_0\in \tD(\cH_{\bS}\otimes\cH_{\bbE})$ there is
\begin{equation}\label{eq:12}\pi^{\fAdv}(\rho_0)\approx^{\tind}_{\epsilon}\underbrace{\fSim_1}_{\text{server side}}(\tIdeal(\underbrace{\fSim_0}_{\text{server side and outputs $b$}}(\rho_0)))\end{equation}
Here $\bS$ is the register corresponding to the server, $\bbE$ is the environment. Note that the distinguisher in \eqref{eq:12} has access to all the inputs, the client-side and server-side outputs and the environment.
\end{exmp}
We note that in this paper we need to consider approximate security, as in the definitions above. 
And we note that in \eqref{eq:12} the operations implicitly take the public parameters as their parameters, for example, the security parameter $\kappa$; the error tolerance $\epsilon$ could also be part of the public parameters.
\subsection{Important Cryptographic Primitives}\label{sec:3.3}
In this section we review basic cryptographic primitives and related results that will be used in our work.
\subsubsection{One-way functions and pseudorandom generators}\label{sec:3.3.1}
One-way functions and pseudorandom generators are both very basic primitives in cryptography. \cite{KLtextbook} One-way functions are functions that are easy to compute but hard to invert. Formally:
\begin{defn}[One-way function]
    We say an efficient function family $(f_\kappa)_{\kappa\in \bN},f_\kappa:\{0,1\}^\kappa\rightarrow \{0,1\}^*$ is a family of post-quantum one-way function if the following holds: for any efficient quantum operation $\fAdv$, there is a negligible function $\fneg(\kappa)$ such that
    $$\Pr_{x\leftarrow \{0,1\}^\kappa}[f(\fAdv(f(x)))=f(x)]\leq \fneg(\kappa)$$
\end{defn}
Pseudorandom genrators (PRG) are functions whose outputs are indistinguishable to random outputs. Formally:
\begin{defn}[PRG]\label{defn:PRG}
    We say an efficient function family $\tPRG(1^\kappa,1^L):\{0,1\}^\kappa\rightarrow \{0,1\}^L$ is a family of post-quantum pseudorandom generator if the following holds: for any efficient quantum operation $\fAdv$, for any $L=\fpoly(\kappa)$, there is a negligible function $\fneg(\kappa)$ such that
    \begin{equation}\label{eq:prg}\Pr\begin{bmatrix}s\leftarrow_r \{0,1\}^\kappa\\\fAdv(\tPRG(1^\kappa,1^L)(s))\rightarrow 0\end{bmatrix}\approx_{\fneg(\kappa)}\Pr\begin{bmatrix}u\leftarrow_r\{0,1\}^L\\\fAdv(u)\rightarrow 0\end{bmatrix}\end{equation}
\end{defn}
\begin{nota}Note that the notation in \eqref{eq:prg} denotes the probability of the last row when the variables are sampled as above.\par
    Another notation to express it is
    \begin{equation}\label{eq:prg2}\tDist\begin{bmatrix}s\leftarrow_r \{0,1\}^\kappa\\\tPRG(1^\kappa,1^L)(s)\end{bmatrix}\approx^{\tind}\tDist\begin{bmatrix}u\leftarrow_r\{0,1\}^L\\ u\end{bmatrix}\end{equation}
    which denotes the indistinguishability between the distributions of the last rows.
\end{nota}
A well-known fact is \cite{Zhanprf}:
\begin{thm}
    The existence of post-quantum one-way functions is equivalent to the existence of post-quantum PRGs.
\end{thm}
\begin{conv}\label{conv:3} In this paper we often omit ``post-quantum'' and simply say, for example, one-way functions or PRGs.
\end{conv}
\subsubsection{Commitments}\label{sec:3.3.comm}
Commitment is a very basic primitive in cryptography. There are different ways to formalize this primitive and we formalize it following the Naor's construction \cite{Naorcomm}.
\begin{defn}
    We say an efficient function family $\fCommit(1^n,1^\kappa):\{0,1\}^n\times\{0,1\}^{n\kappa}\times\{0,1\}^{3n\kappa}$ is a statistically-binding computationally-hiding commitment scheme if it satisfies:
    \begin{itemize}
        \item (Computationally hiding) For any efficient quantum adversary $\fAdv$, for any $n=\fpoly(\kappa)$, there is a negligible function $\fneg(\kappa)$ such that
        \begin{align*}
            &\Pr\begin{bmatrix}
                \fAdv\rightarrow (m_0,m_1,pp)\in \{0,1\}^{n}\times \{0,1\}^n\times \{0,1\}^{3n\kappa}\\
                r\leftarrow_r\{0,1\}^{n\kappa}\\
                \fAdv(\fCommit(m_0,r,pp))=0
            \end{bmatrix}\\\approx_{\fneg(\kappa)}&\Pr\begin{bmatrix}
                \fAdv\rightarrow (m_0,m_1,pp)\in \{0,1\}^{n}\times \{0,1\}^n\times \{0,1\}^{3n\kappa}\\
                r\leftarrow_r\{0,1\}^{n\kappa}\\
                \fAdv(\fCommit(m_1,r,pp))=0
            \end{bmatrix}
        \end{align*}
        \item (Statistically binding) For any $n=\fpoly(\kappa)$ there is a negligible function $\fneg(\kappa)$ such that
        \begin{equation}
            \Pr\begin{bmatrix}
                pp\leftarrow_r\{0,1\}^{3n\kappa}\\
                \exists (m_0,r_0,m_1,r_1):\fCommit(m_0,r_0,pp)=\fCommit(m_1,r_1,pp)\land m_0\neq m_1
            \end{bmatrix}\leq \fneg(\kappa)
        \end{equation}
    \end{itemize}
\end{defn}
Here we call $pp$ the public randomness and $r$ the sender-side randomness.
\begin{thm}[By \cite{Naorcomm}]
    There exists a statistically-binding computationally-hiding commitment scheme assuming the existence of one-way functions.
\end{thm}
A typical commitment protocol based on $\fCommit$ is as follows.
\begin{enumerate}
\item Committing phase:
\begin{enumerate}
    \item The receiver samples $pp\leftarrow_r \{0,1\}^{3n\kappa}$ and sends it to the sender.
    \item The sender samples and stores $r\leftarrow_r\{0,1\}^{n\kappa}$. The sender computes and sends $\fCommit(x,r,pp)$ to the receiver as the commitment.
\end{enumerate}
\item Opening phase:
\begin{enumerate}
    \item The sender sends $r$ to the receiver.
    \item The receiver checks $\fCommit(x,r,pp)=com$.
\end{enumerate}
\end{enumerate}
\subsubsection{Collapsing hash functions}\label{sec:3.3.2}
Collapsing hash functions are the quantum generalization of the collision-resistant hash functions. For convenience we first review the collision-resistant hash functions.
\begin{defn}[Collision-resistant hash functions]\label{defn:3.4}
    Suppose $\fHash(1^n,1^\kappa)$ is an efficient procedure that samples functions from $\{0,1\}^n$ to $\{0,1\}^\kappa$; we additionally require that the description length of $h$ sampled from $\fHash(1^n,1^\kappa)$ is $\fpoly(\kappa)$ (note that the description could be based on Turing machines). We say $\fHash$ is a family of collision-resistant hash functions if the following holds: for any efficient quantum adversary $\fAdv$, for any $n=\fpoly(\kappa)$, there exists a negligible function $\fneg(\kappa)$ such that
    $$\Pr\begin{bmatrix}
    h\leftarrow_r \fHash(1^n,1^\kappa)\\
    \fAdv(h)\rightarrow (x_0,x_1)\\
    h(x_0)=h(x_1)
    \end{bmatrix}\leq \fneg(\kappa)$$
\end{defn}
Note that although Definition \ref{defn:3.4} is already stated in the post-quantum setting, this definition is usually not sufficient when we need the quantum analog of collision-resistance. What we need is the following notion, the collapsing hash functions \cite{Unruh16}:
\begin{defn}
    Consider $\fHash$ whose basic set-up is as in Definition \ref{defn:3.4}. Suppose $\bX$ is a quantum register holding $n$ qubits. Define the following two operations:
    \begin{itemize}\item Define $\fEv(h,\bX)$ as follows, where $h$ is a function from $\{0,1\}^n$ to $\{0,1\}^\kappa$: $\fEv(h,\bX)$ evaluates $h$ on $\bX$ and measures and disgards the results (that is, only $\bX$ is kept in the end).
        \item Define $\fCollapse(\bX)$ as the operation that collapses $\bX$ on the standard basis.
    \end{itemize}
        We say $\fHash$ is a family of collapsing hash functions if the following holds: for any efficient quantum adversary $\fAdv$, for any $n=\fpoly(\kappa)$, there exists a negligible function $\fneg(\kappa)$ such that
        $$\Pr\begin{bmatrix}
            h\leftarrow_r \fHash(1^n,1^\kappa)\\
            \fAdv(h)\rightarrow \rho\in \tD(\cH_{\bX}\otimes \cH_{\bS})\\
            \fEv(h,\bX)\\
            \fAdv(h,\bX,\bS)\rightarrow 0
        \end{bmatrix}\approx_{\fneg(\kappa)}\Pr\begin{bmatrix}
            h\leftarrow_r \fHash(1^n,1^\kappa)\\
            \fAdv(h)\rightarrow \rho\in \tD(\cH_{\bX}\otimes \cH_{\bS})\\
            \fCollapse(\bX)\\
            \fAdv(h,\bX,\bS)\rightarrow 0
        \end{bmatrix}$$
\end{defn}
The following results are from existing works \cite{Unruh16a,Unruh16}.
\begin{thm}
    If $\fHash$ is collapsing, then it's collision-resistant.
\end{thm}
\begin{thm}
    Existence of collapsing hash functions implies the existence of one-way functions.
\end{thm}
\begin{thm}
    There exists a family of collapsing hash functions assuming the LWE assumption.
\end{thm}
The LWE \cite{regevLWE} assumption is a widely-accepted assumption with security against quantum computations.\par
\begin{thm}
    There exists a family of collapsing hash functions in the quantum random oracle model.
\end{thm}
The quantum random oracle model (QROM) \cite{QRO} is a widely-used idealized model for modeling hash functions. Although the constructions in QROM are not always instantiable \cite{rorevisited,ES20}, it provides a good evidence that hash functions used in reality, like SHA-3, are indeed collapsing.
\subsubsection{Secure two-party computations}\label{sec:3.3.3}
Below we review the notion of secure two-party computations \cite{GB01}. We first formalize its set-up. Here we consider its general form, where both parties have private inputs and both parties receive private outputs.
\begin{setup}\label{setup:twopc}The set-up for general secure two-party computations is as follows. We consider a protocol between Alice and Bob.\par
    The inputs of the protocol are: \begin{itemize}\item Public parameters: problem sizes $1^{n_A},1^{n_B},1^{m_A},1^{m_B}$, where $A$ means Alice, $B$ means Bob, $n$ means the input size and $m$ means the output size; security parameter $\kappa$.
        \item Public efficient classical functions $f_{A}:\{0,1\}^{n_A}\times\{0,1\}^{n_B}\rightarrow \{0,1\}^{m_A}$, $f_{B}:\{0,1\}^{n_A}\times\{0,1\}^{n_B}\rightarrow \{0,1\}^{m_B}$.
        \item Alice-side input $x_{A}\in \{0,1\}^{n_A}$,  Bob-side input $x_{B}\in \{0,1\}^{n_B}$.
    \end{itemize}
    The outputs of the protocol are stored in the following registers: Alice-side $\bflag^{(A)}$ holding values in $\{\fpass,\ffail\}$, $\boldy^{(A)}$ holding values in $\{0,1\}^{m_A}$, Bob-side $\bflag^{(B)}$ holding values in $\{\fpass,\ffail\}$, $\boldy^{(B)}$ holding values in $\{0,1\}^{m_B}$.\par
    Either Alice or Bob could be malicious. For modeling the initial states in the malicious setting, when Alice is malicious, suppose the Alice-side registers are denoted by $\bS^{(A)}$; when Bob is malicious, suppose the Bob-side registers are denoted by $\bS^{(B)}$. The environment is denoted by $\bbE$.\par
    The completeness and soundness are in Definition \ref{defn:3.6}, \ref{defn:3.7} below. The efficiency is defined in the common way.
\end{setup}
\begin{defn}\label{defn:3.6}
    We say a protocol under Set-up \ref{setup:twopc} is complete if when both parties are honest, the output state of the protocol is negligibly close to the following state: $\bflag^{(A)},\bflag^{(B)}$ holds the value $\fpass$, $\boldy^{(A)}$ holds $f_A(x^{(A)},x^{(B)})$, $\boldy^{(B)}$ holds $f_B(x^{(A)},x^{(B)})$.
\end{defn}
\begin{defn}\label{defn:3.7}
    To define the soundness under Set-up \ref{setup:twopc}, we first define $\tTwoPCIdeal^{(A)},\tTwoPCIdeal^{(B)}$ as follows.\par
    $\tTwoPCIdeal^{(B)}$ works as follows. It interacts with the honest Alice and malicious Bob. Note that the Alice's private inputs are prepared in advanced but Bob's private inputs could be adversarily generated.
    \begin{enumerate}
\item $\tTwoPCIdeal^{(B)}$ receives a bit $b\in \{0,1\}$ from Bob, then:
\begin{itemize}
    \item (Early abort option) If $b=1$, $\tTwoPCIdeal^{(B)}$ sets $\bflag^{(A)}$ to be $\ffail$.
    \item If $b=0$, $\tTwoPCIdeal^{(B)}$ takes $x^{(A)}$ from Alice's inputs and take $x^{(B)}$ from Bob. Then $\tTwoPCIdeal^{(B)}$ stores $f_B(x^{(A)},x^{(B)})$ in $\boldy^{(B)}$. Then:
    \begin{enumerate}
        \item (Early abort option) $\tTwoPCIdeal^{(B)}$ receives a bit $b^\prime\in \{0,1\}$ from Bob, then:
        \begin{itemize}\item If $b=1$, $\tTwoPCIdeal^{(B)}$ sets $\bflag^{(A)}$ to be $\ffail$.\item If $b=0$, $\tTwoPCIdeal^{(B)}$ sets $\bflag^{(A)}$ to be $\fpass$, and stores $f_A(x^{(A)},x^{(B)})$ in $\boldy^{(A)}$.\end{itemize}
    \end{enumerate}
\end{itemize}
\end{enumerate}
$\tTwoPCIdeal^{(B)}$ is defined similarly by swapping the role of Alice and Bob in the definition above.\par
    We say a protocol $\pi$ under Set-up \ref{setup:twopc} is sound if:
    \begin{itemize}
        \item (Soundness against malicious Bob) For any efficient quantum adversary $\fAdv$ playing the role of Bob, there exist efficient quantum operations $\fSim$ such that for any polynomial-size inputs $1^{n_A},1^{n_B},1^{m_A},1^{m_B},f_A,f_B,x^{(A)}$ and any initial state $\rho_0\in \tD(\cH_{\bS^{(B)}}\otimes \cH_{\bbE})$:
        \begin{equation}\label{eq:2pcsec}\pi^{\fAdv}(\rho_0)\approx^{\tind}(\underbrace{\fSim}_{\text{working on the server side and interacting with $\tTwoPCIdeal^{(B)}$}}\circ\tTwoPCIdeal^{(B)})(\rho_0)\end{equation}
        We note that the distinguisher in \eqref{eq:2pcsec} has access to all the inputs above, client-side states, server-side states and the environment.
        \item (Security against malicious Alice) Similar to the definition above by swapping Alice and Bob in the definition.
    \end{itemize}
\end{defn}
\begin{conv}\label{conv:4}
    We make the following convention for working on protocol design and analysis.\par
    First, in this paper, we use ``Set-up'' to organize the set-up of protocols and notions, as in this section. Then when we design protocols, simply refering to a set-up gives lots of information including the inputs, output registers and security definitions. For readability we still review the inputs when we describe protocols.\par
    For simplicity of expressions, we might omit the inputs that are public and given by values. Note that in \eqref{eq:2pcsec} above we already make the inputs $1^{n_A},f_A$, etc implicit for simplicity.
\end{conv}
We need the following theorem in our work, which is proved in \cite{BCKM20,GLSV20}.
\begin{thm}\label{thm:2pcowf}
    Assuming the existence of one-way functions, there exists a quantum protocol for general secure two-party computations.
\end{thm}
\subsubsection{Succinct arguments, succinct-AoK and succinct-ZK-AoK}\label{sec:3.3.4}
In this section we review the notion of succinct arguments, succinct arguments of knowledge (succinct-AoK) and succinct zero-knowledge arguments of knowledge (succinct-ZK-AoK) for NP. \cite{Kilian92,CMSZ21}
\begin{setup}\label{setup:sa}
    The set-up for succinct arguments / succinct-AoK / succinct-ZK-AoK for NP is as follows. We consider a classical protocol between a client and a server.\par
    The protocol takes the following public parameters and inputs: problem sizes $1^n,1^m$, security parameter $1^\kappa$, an efficient classical relation $R:\{0,1\}^n\times\{0,1\}^m\rightarrow \{0,1\}$, an input $x\in \{0,1\}^n$.\par
    The server holds a classical register $\bS_w$ whose value is in $\{0,1\}^m$.\par
    The client-side output register is $\bflag$ whose value is in $\{\fpass,\ffail\}$.\par
    For succinct arguments and succinct-AoK the server is the malicious party. For succinct-ZK-AoK both parties could be malicious. For modeling the initial states in the malicious server setting, suppose the overall server-side registers are denoted by $\bS$. For modeling the initial states in the malicious client setting, the client-side registers are denoted by $\bbC$ and the environment is denoted by $\bbE$.\par
    The completeness is in Definition \ref{defn:3.8} below. The soundness and zero-knowledge property is as follows.\begin{itemize}\item If we consider the succinct arguments primitive, the soundness is in Definition \ref{defn:3.9}; \item if we consider the succinct-AoK, the soundness is in Definition \ref{defn:3.10} below; if we consider the succinct-ZK-AoK, the soundness is in Definition \ref{defn:3.10} below and the zero-knowledge property is in Definition \ref{defn:3.zk} below.\end{itemize} The efficiency is defined in the common way. We additionally want the property that the size of the client-side computation is $\fpoly(n,\kappa)$.
\end{setup}
\begin{defn}\label{defn:3.8}
    We say a protocol under Set-up \ref{setup:sa} is complete if when the server is honest and initially $\bS_w$ holds a value $w$ that satisfies $R(x,w)=1$, the client accepts with probability $\geq 1-\fneg(\kappa)$.
\end{defn}
\begin{defn}\label{defn:3.9}
    We say a protocol under Set-up \ref{setup:sa} (for succinct arguments primitive) is sound if the following holds. For any efficient quantum adversary $\fAdv$ as the server, any polynomial-size inputs $1^n,1^m,R,x$, any initial state $\rho_0\in \tD(\cH_{\bS})$, if there is no $w$ that satisfies $R(x,w)=1$, the client rejects with probability $\geq 1-\fneg(\kappa)$.
\end{defn}
\begin{defn}\label{defn:3.10}
    We say a protocol under Set-up \ref{setup:sa} (for succinct-AoK primitive) is $(\delta,\epsilon)$-sound if the following holds. For any efficient quantum adversary $\fAdv$ as the server there exists an efficient quantum operation $\fExt$ such that for any polynomial-size inputs $1^n,1^m,R,x$, for any initial state $\rho_0\in \tD(\cH_{\bS})$, if the client accepts with probability $\geq 1-\delta$, there is
    $$\tr(\Pi_{w:R(x,w)=1}^{\bS_{w}}\fExt(\rho_0))\geq 1-\epsilon-\fneg(\kappa)$$ 
\end{defn}
\begin{defn}\label{defn:3.zk}
    We say a protocol $\pi$ under Set-up \ref{setup:sa} (for succinct-ZK-AoK primitive) is zero-knowledge if the following holds. For any efficient quantum adversary $\fAdv$ as the client there exists an efficient quantum operation $\fSim$ such that for any polynomial-size inputs $1^n,1^m,R,x$, for any initial state $\rho_0\in \tD(\cH_{\bbC}\otimes\cH_{\bbE})$ and any $w$ in register $\bS_w$ such that $R(x,w)=1$:
    \begin{equation}\label{eq:17}
        \pi^{\fAdv}(\rho_0)\approx^{\tind}\underbrace{\fSim}_{\text{client side}}(\rho_0)\otimes\underbrace{\ket{w}\bra{w}}_{\bS_w}
    \end{equation}
    Note that the distinguisher in \eqref{eq:17} also has access to all the inputs, which are omitted in the expression\footnote{Note that since the statement is quantified with ``for all inputs'', whether or not we explicitly give these inputs to the distinguisher actually do not matter: the distinguisher could always hardcode them. In the later content we will not care about this detail.}.
\end{defn}
By Kilian's work \cite{Kilian92} and the proof of its post-quantum security \cite{CMSZ21}, there is:
\begin{thm}\label{thm:existzk}
    Assuming the existence of collapsing hash functions, there exists a succinct-ZK-AoK protocol for NP that is $(\delta,O(\delta))$-sound for all $\delta>0$.
\end{thm}
We note that the $\delta,\epsilon$ parameters in our set-up are different from the $\epsilon$ parameter in \cite{CMSZ21}: we define the parameters to be the failing errors, while \cite{CMSZ21} defines $\epsilon$ to be the success probability.
\subsubsection{Yao's garbled circuits}\label{sec:3.3.5}
Yao's garbled circuits is a basic and famous construction in secure computations. \cite{YaoGCOrigin} It achieves a notion called decomposable randomized encoding. \cite{AppleBaum2017} In this section we review these notions.
\begin{defn}\label{defn:3.gc}
    A decomposable randomized encoding scheme is defined to be efficient classical algorithms $(\fGarbleC,\fGarbleI,\fDc)$ where:
    \begin{itemize}
        \item $\fGarbleI$ is a deterministic algorithm that takes $1^n,1^\kappa$ as parameters, takes $x\in \{0,1\}^n,r\in \{0,1\}^{n\times\kappa}$ and calculates the encoding of $x$ using $r$. This part is called the input encoding and $r$ is the shared randomness used in both $\fGarbleI$ and $\fGarbleC$.
        \item Use $\cC_n$ to denote the set of circuits with input length $n$. $\fGarbleC$ takes $1^n,1^\kappa$ as parameters, takes $C\in \cC_n,r\in \{0,1\}^{n\times\kappa}$ and calculates the encoding of $C$ using $r$. This part is called the circuit encoding. Note that $\fGarbleC$ itself could be a randomized algorithm; its deterministic function description that makes its internal random coins explicit could be written as $\fGarbleC(C,r;r_{\text{GCin}})$ (see Fact \ref{fact:2}).
        \item $\fDc$ is a deterministic algorithm that takes the outputs of $\fGarbleC(C,r)$ and $\fGarbleI(x,r)$ as inputs and outputs $C(x)$.
    \end{itemize}
    Here ``decomposable'' means that the encoding of $C(x)$ could be decomposed into the input encoding and the circuit encoding separately. What $\fGarbleC$ and $\fGarbleI$ share is only the shared random coins $r$.\par
        
    And they satisfy the following properties:
    \begin{itemize}
        \item (Completeness) For any polynomial-size $1^n$, any $x,C$:
        $$\Pr_{r\leftarrow\{0,1\}^{n\kappa}}[\fDc(\fGarbleC(C,r),\fGarbleI(x,r))=C(x)]\geq 1-\fneg(\kappa)$$
        \item (Soundness) There exists an efficient classical algorithm $\fSim$ such that for any polynomial-size $1^n$, any $x,C$:
        \begin{equation}\label{eq:14}\fSim(C(x))\approx^{\tind}\tDist\begin{bmatrix}r\leftarrow_r\{0,1\}^{n\kappa}\\ (\fGarbleC(C,r),\fGarbleI(x,r))\end{bmatrix}\end{equation}
    \end{itemize}
\end{defn}

By Yao's famous construction \cite{YaoGCOrigin} there is:
\begin{thm}
    Assuming the existence of one-way functions, Yao's Garbled Circuits a decomposable randomized encoding scheme.
\end{thm}
As discussed in Section \ref{sec:2.7.3}, in this work we need a relatively stronger form of soundness for Yao's scheme, which is formalized below.
\begin{defn}\label{defn:3.14gc}
    We say a garbling scheme under Definition \ref{defn:3.gc} is sound with input-first simulation if there exist efficient classical algorithms $\fSim_{\text{ie}},\fSim_{\text{ce}}$ such that for any polynomial-size $1^n$, any $x,C$:
    \begin{equation}
        \tDist\begin{bmatrix}
            \tilde{x}\leftarrow\fSim_{\text{ie}}()\\ (\tilde{x},\fSim_{\text{ce}}(\tilde{x},C(x)))
        \end{bmatrix}\approx^{\tind}\tDist\begin{bmatrix}r\leftarrow_r\{0,1\}^{n\kappa}\\
            \hat{x}:=\fGarbleI(x,r)\\ (\hat{x},\fGarbleC(C,r))\end{bmatrix}
    \end{equation}
\end{defn}
Note that we use $:=$ to define new variables in the expressions.
\begin{thm}\label{thm:yaogc}
    Assuming the existence of one-way functions, Yao's Garbled Circuits a decomposable randomized encoding scheme with soundness under Definition \ref{defn:3.14gc}.
\end{thm}
\subsubsection{Trapdoor claw-free functions}\label{sec:3.3.6}
(Noisy) trapdoor claw-free functions (NTCF/TCF) is an important technique in quantum cryptography. It is raised in \cite{BCMVV,MahadevVerification} for the test-of-quantumness and CVQC problem. There are some variants of definitions raised for different problems, for example, NTCF/TCF with the adaptive hardcore bit property \cite{BCMVV,GVRSP}, or extended NTCF \cite{MahadevVerification}. Below we review the definition of NTCF without the adaptive hardcore bit property \cite{BKVV,zhang24}.
\begin{defn}[NTCF, without the adaptive hardcore bit property]\label{defn:3.12}
    The NTCF is defined to be a class of polynomial time algorithms as below, which take security parameter $1^\kappa$ as the parameter. $\fKg$ is a sampling algorithm. $\fDc$, $\fCHK$ are deterministic algorithms. $\fEv$ is allowed to be a sampling algorithm. $ \fpoly^\prime$ is a polynomial that determines the the range size. $$\fKg(1^\kappa)\rightarrow (\sk,\pk),$$ $$\fEv_\pk: \{0,1\}\times \{0,1\}^{\kappa}\rightarrow \{0,1\}^{\fpoly^\prime(\kappa)},$$ $$\fDc_\sk: \{0,1\}\times \{0,1\}^{\fpoly^\prime(\kappa)}\rightarrow \{0,1\}^{\kappa}\cup \{\bot\},$$ $$\fCHK_{\pk}: \{0,1\}\times \{0,1\}^{\kappa}\times \{0,1\}^{\fpoly^\prime(\kappa)}\rightarrow \{\ftrue ,\ffalse\}$$ And they satisfy the following properties:\par
	\begin{itemize}
	\item (Correctness) 
	\begin{itemize}
	\item (Noisy 2-to-1) For all possible $(\sk,\pk)$ in the range of $\fKg(1^\kappa)$ there exists a sub-normalized probability distribution $(p_y)_{y\in \{0,1\}^{\fpoly^\prime(\kappa)}}$ that satisfies: for any $y$ such that $p_y\neq 0$, $\forall b\in \{0,1\}$, there is $\fDc_\sk(b,y)\neq \bot$, and
	\begin{equation}\label{eq:64co}\fEv_\pk(\ket{+}^{\otimes (1+\kappa)})\approx_{\fneg(\kappa)}\sum_{y:p_y\neq 0}\frac{1}{\sqrt{2}}(\ket{\fDc_\sk(0,y)}+\ket{\fDc_\sk(1,y)})\otimes \sqrt{p_y}\ket{y}\end{equation}
	\item (Correctness of $\fCHK$) For all possible $(\sk,\pk)$ in the range of $\fKg(1^\kappa)$, $\forall x\in \{0,1\}^{\kappa}, \forall b\in \{0,1\}$:
	$$\fCHK_\pk(b,x,y)=\ftrue\Leftrightarrow\fDc_{\sk}(b,y)=x$$
	\end{itemize}
	\item (Claw-free) For any BQP adversary $\fAdv$,
	\begin{equation}\Pr\left[\begin{aligned}&(\sk,\pk)\leftarrow \fKg(1^\kappa),\\&\fAdv(\pk)\rightarrow (x_0,x_1,y):\quad x_0\neq \bot,x_1\neq \bot, x_0\neq x_1\\&\fDc_\sk(0,y)=x_0,\fDc_\sk(1,y)=x_1\end{aligned}\right]\leq \fneg(\kappa)\end{equation}
		\end{itemize}
\end{defn}
By \cite{BCMVV,BKVV} we have the following:
\begin{thm}
    NTCF could be instantiated from the LWE assumption.
\end{thm}
\subsubsection{Remote state preparation with verifiability}\label{sec:3.3.7}
Remote state preparation with verifiability (RSPV) \cite{zhang24} is an important primitive in quantum cryptography. Famous works on RSPV include \cite{GVRSP}, where a classical channel RSPV protocol for $\ket{+_\theta}$ states is constructed. \cite{zhang24} gives a framework for working on RSPV and gives more protocols. Below we review the notions and results that are related to our work.\par
We first review a common set-up for RSPV. Below we repeat the definitions used in \cite{zhang24}.
\begin{setup}\label{setup:rspv}
    A typical set-up for RSPV is as follows. We consider a protocol between a client and a server. We consider a tuple of states $(\ket{\varphi_1},\ket{\varphi_2},\cdots\ket{\varphi_D})$.\par
    The protocol takes the following parameters: security parameter $1^\kappa$, approximation error parameter $1^{1/\epsilon}$.\par
    The outputs of the protocol are stored in the following registers: the client-side $\bflag$ holding value in $\{\fpass,\ffail\}$, the client-side classical register $\bD$ holding value in $[D]$, the server-side quantum register $\bQ$ whose dimension is suitable for holding the states.\par
    The server is the malicious party. For modeling the initial states in the malicious setting, suppose the server-side registers are denoted by $\bS$ and the environment is denoted by $\bbE$.\par
    The completeness and soundness are in Definition \ref{defn:3.13}, \ref{defn:3.14} below. The efficiency is defined in the common way.
\end{setup}
\begin{defn}[Completeness]\label{defn:3.13}
    We say a protocol under Set-up \ref{setup:rspv} is complete if when the server is honest the output state of the protocol is negligibly close to the following state: $\bflag$ holds the value $\fpass$, $\bD$ holds a random $i\in [D]$, $\bQ$ holds the corresponding $\ket{\varphi_i}$.\par
    Below we use $\rho_{tar}$ to denote this output state on $\bD,\bQ$ in the honest setting.
\end{defn}
\begin{defn}[Soundness]\label{defn:3.14}
    To define the soundness under Set-up \ref{setup:rspv}, we first define $\tRSPVIdeal$ as follows.\par
    $\tRSPVIdeal$ receives a bit $b$ from the server, then:
    \begin{itemize}
        \item If $b=0$, it sets $\bflag$ to be $\fpass$, samples $i\in [D]$ and stores it in $\bD$, and stores $\ket{\varphi_i}$ in $\bQ$.
        \item If $b=1$, it sets $\bflag$ to be $\ffail$.
    \end{itemize}
    We say a protocol $\pi$ under Set-up \ref{setup:rspv} is approximately sound if:\par
    For any efficient quantum adversary $\fAdv$, there exist efficient quantum operations $\fSim=(\fSim_0,\fSim_1)$ such that for any state $\rho_0\in \tD(\cH_{\bS}\otimes\cH_{\bbE})$:
    $$\pi^{\fAdv}(\rho_0)\approx_{\epsilon}^{\tind}\underbrace{\fSim_1}_{\text{on }\bS,\bQ}(\tRSPVIdeal(\underbrace{\fSim_0}_{\text{on }\bS\text{ and outputs }b}(\rho_0)))$$
    We note that (1) there are some variants of soundness definitions for RSPV. The definition we use here corresponds to Definition 3.3 in \cite{zhang24}, which follows the standard form of the simulation-based paradigm. (2) Since the error tolerance parameter is already part of the set-up so we could simply say ``approximately sound'' in the definition without ambiguity; we could also say it's $\epsilon$-sound for convenience.    
\end{defn}
In the framework of \cite{zhang24}, there is a useful notion called preRSPV. Below we review this notion.
\begin{setup}\label{setup:prerspvnt}
    A typical set-up of preRSPV is as follows. The basic set-up is similar to Set-up \ref{setup:rspv} with the following differences: we are considering a pair of protocols $(\pi_{\ttest},\pi_{\tcomp})$ instead of a single protocol; and we do not model $1^{1/\epsilon}$ as part of the protocol inputs\footnote{Instead, $\epsilon$ directly appears in the soundness definition.}.\par
    The completeness and soundness are in Definition \ref{defn:3.15}, \ref{defn:3.16} below.
\end{setup}
\begin{defn}\label{defn:3.15}
    We say $(\pi_{\ttest},\pi_{\tcomp})$ under Set-up \ref{setup:prerspvnt} is complete if:
    \begin{itemize}
        \item In $\pi_{\ttest}$, when the server is honest, the passing probability is negligibly close to $1$.
        \item In $\pi_{\tcomp}$, when the server is honest, the output state of the protocol is as required in Definition \ref{defn:3.13}.
    \end{itemize}
\end{defn}
\begin{defn}\label{defn:3.16}
    We say $(\pi_{\ttest},\pi_{\tcomp})$ under Set-up \ref{setup:prerspvnt} is $(\delta,\epsilon)$-sound if:\par
    For any efficient quantum adversary $\fAdv$, there exists an efficient quantum operation $\fSim$ such that for any initial state $\rho_0\in \tD(\cH_{\bS}\otimes\cH_{\bbE})$:\par
    If
    $$\tr(\Pi_{\fpass}(\pi_{\ttest}^{\fAdv}(\rho_0)))\geq 1-\delta$$
    then
    $$\Pi_{\fpass}(\pi_{\tcomp}^{\fAdv}(\rho_0))\approx_\epsilon^{\tind}\Pi_{\fpass}(\underbrace{\fSim}_{\bS,\bQ,\bflag}(\underbrace{\rho_{tar}}_{\bD,\bQ}\otimes\rho_0))$$
\end{defn}
Note that here we only consider indistinguishability on the passing space, which is chosen for some convenience reason in \cite{zhang24}. As shown in \cite{zhang24}, there is a protocol template that amplifies a preRSPV to an RSPV by a suitable sequential repetition:
\begin{mdframed}[backgroundcolor=black!10]
    Below we assume $(\pi_{\ttest},\pi_{\tcomp})$ is a preRSPV under Set-up \ref{setup:prerspvnt} that is $(\delta,\epsilon_0)$-sound. 
    \begin{prtl}\label{prtl:amptempltemp}
        A temporary protocol $\pi_{\ttemp}$ is as follows. The protocol inputs are: public parameters $1^{1/\epsilon}, 1^\kappa$. It is required that $\epsilon>\epsilon_0$.
        \begin{enumerate}
            \item Define $L=\frac{512}{\delta(\epsilon-\epsilon_0)^3},p=\frac{\epsilon-\epsilon_0}{8}.$\par
            For each $i\in [L]$:
            \begin{enumerate}
                \item The client randomly chooses $\tmode^{(i)}=\ttest$ with probability $p$ or $\tmode^{(i)}=\tcomp$ with probability $1-p$.
                \item The client executes $\pi_{\tmode^{(i)}}$ with the server. Store the outputs in $\bflag^{(i)},\bD^{(i)},\bQ^{(i)}$.
            \end{enumerate}
            \item The client sets $\bflag$ to be $\ffail$ if any of $\bflag^{(i)}$ is $\ffail$. Otherwise it sets $\bflag$ to be $\fpass$; and it randomly chooses $i\in [L]$ such that $\tmode^{(i)}=\tcomp$ and sends $i$ to the server; both parties use the states in $\bD^{(i)},\bQ^{(i)}$ as the outputs.
        \end{enumerate}
    \end{prtl}
    \begin{prtl}\label{prtl:amptempl}
        An RSPV protocol under Set-up \ref{setup:rspv} is as follows. The protocol inputs are: public parameters $1^{1/\epsilon}, 1^\kappa$. It is required that $\epsilon>\epsilon_0$.
        \begin{enumerate}
            \item Define $L=\frac{1728}{(\epsilon-\epsilon_0)^3}$.\par
            For each $i\in [L]$:
            \begin{enumerate}
                \item The client and the server execute $\pi_{\ttemp}$ with approximation error tolerance $\frac{\epsilon+\epsilon_0}{2}$. Store the outputs in register $\bflag^{(i)},\bD^{(i)},\bQ^{(i)}$.
            \end{enumerate}
            \item The client sets $\bflag$ to be $\ffail$ if any of $\bflag^{(i)}$ is $\ffail$. Otherwise it sets $\bflag$ to be $\fpass$; and it randomly chooses $i\in [L]$ and sends $i$ to the server; both parties use the states in $\bD^{(i)},\bQ^{(i)}$ as the outputs.
        \end{enumerate}
    \end{prtl}
\end{mdframed}
The completeness and efficiency are from the protocol description.
\begin{thm}\label{thm:3.10}
    Protocol \ref{prtl:amptempl} is $\epsilon$-sound.
\end{thm}
\begin{proof}
    The amplification is the amplification procedure used in \cite{zhang24}. Note that the amplification from preRSPV to Protocol \ref{prtl:amptempltemp} is what \cite{zhang24} calls ``preRSPV to RSPV amplification''. The amplification to Protocol \ref{prtl:amptempl} is to amplify from the passing-space simulation to all-space simulation, which is proved in Section 3.1.1 of \cite{zhang24}.
\end{proof}
Finally we review the following theorems about the concrete constructions of RSPV protocols, which are also needed in our work.
\begin{thm}[By \cite{zhang24}]\label{thm:kpsrspv}
    Assuming the existence of NTCF, there exists a classical-channel RSPV protocol for state family $\{\frac{1}{\sqrt{2}}(\ket{0}\ket{x_0}+\ket{1}\ket{x_1}):x_0,x_1\in \{0,1\}^n\}$.
\end{thm}
\begin{thm}[By \cite{BGKPV23}]\label{thm:bb84rspv}
    Assuming the existence of NTCF, there exists a classical-channel RSPV protocol for BB84 states.
\end{thm}
Note that the NTCF above does not assume the adaptive hardcore bit property, as formalized in Definition \ref{defn:3.12}.
\section{Secure Function Sampling with Long Outputs}\label{sec:4}
In this section we formalize the notion of secure function sampling (SFS), prove its classical impossibility, and prepare for the construction of its quantum protocols.
\subsection{Definition of Secure Function Sampling}\label{sec:4.1}
In this section we formalize the notion of SFS. Below we first formalize its set-up.
\begin{setup}\label{setup:1}
    The set-up for SFS is as follows. We consider a protocol between a client and a server.\par
    The inputs of the protocol are all public, which are: problem sizes $1^n,1^m$, security parameter $1^\kappa$, error tolerance parameter $1^{1/\epsilon}$, an efficient classical function $f:\{0,1\}^n\rightarrow \{0,1\}^m$.\par
    The outputs of the protocol are as follows. The client-side output registers are $\bflag$, which holds a value in $\{\fpass,\ffail\}$, and $\bx^{(out)}$, which holds a string in $\{0,1\}^n$; the server-side output register is $\bQ^{(out)}$, which holds a string in $\{0,1\}^m$.\par
    The server is the malicious party. For modeling the initial states in the malicious setting, suppose the server-side registers are denoted by $\bS$ and the environment is denoted by $\bbE$.\par
    The completeness and soundness are in Definition \ref{defn:4.1}, \ref{defn:4.2} below. The efficiency is defined in the common way.
\end{setup}
As discussed in the introduction, in this paper we would like to construct an SFS protocol in succinct communication.\par
The completeness of SFS is as follows.
\begin{defn}[Completeness of SFS]\label{defn:4.1}
We say a protocol under Set-up \ref{setup:1} is complete if when the server is honest, the output state of the protocol is negligibly close to the following state: the register $\bflag$ holds the value $\fpass$, the register $\bx^{(out)}$ holds a random classical string $x\leftarrow_r\{0,1\}^n$, the register $\bQ^{(out)}$ holds the corresponding classical string $f(x)$.\par
We use $\rho_{tar}$ to denote the output state on $\bx^{(out)},\bQ^{(out)}$ in the honest setting.
\end{defn}
The soundness of SFS is as follows.
\begin{defn}[Soundness of SFS]\label{defn:4.2}
    To define the soundness, we first define $\tSFSIdeal$ as follows.\par
    $\tSFSIdeal$ receives a bit $b\in\{0,1\}$ from the server, then:
        \begin{itemize}\item If $b=0$, $\tSFSIdeal$ sets $\bflag$ to be $\fpass$, and it randomly samples $x\leftarrow_r \{0,1\}^n$, stores $x$ in $\bx^{(out)}$ and stores $f(x)$ in $\bQ^{(out)}$.
            \item If $b=1$, $\tSFSIdeal$ sets $\bflag$ to be $\ffail$.
        \end{itemize}
    Then we define the soundness as follows. We say a protocol $\pi$ under Set-up \ref{setup:1} is $\epsilon$-sound if:\par
For any efficient quantum adversary $\fAdv$, there exist efficient quantum operations $\fSim=(\fSim_0,\fSim_1)$ such that for any polynomial-size inputs $1^n,1^m,1^{1/\epsilon},f$ and any initial state $\rho_0\in \tD(\cH_{\bS}\otimes\cH_{\bbE})$:
\begin{equation}\label{eq:13}\pi^{\fAdv}(\rho_0)\approx^{\tind}_{\epsilon}\underbrace{\fSim_1}_{\text{on }\bS,\bQ^{(out)}}(\tSFSIdeal(\underbrace{\fSim_0}_{\text{on }\bS\text{ and outputs }b}(\rho_0)))\end{equation}
\end{defn}
\subsection{Classical Impossibility}\label{sec:4.2}
In this section we prove the SFS primitive is not achievable by a classical protocol. 
As described in Section \ref{sec:2.6}, the impossibility result considers a pseudorandom generator (PRG) as the function $f$. Below we state the theorem.
\begin{thm}\label{thm:4.1r}
    Suppose a classical protocol $\pi$ in Set-up \ref{setup:1} is complete, efficient and has communication complexity $\fpoly(n,\kappa)$. Then the following statement could \emph{not} be true:\par
    For any efficient quantum adversary $\fAdv$, there exist efficient quantum operations $\fSim=(\fSim_0,\fSim_1)$ such that for any polynomial-size $1^n,1^m$, take $\epsilon<1-O(1)$ and $f=\tPRG(1^n,1^m)$, for any initial state $\rho_0\in \tD(\cH_{\bS}\otimes\cH_{\bbE})$:
\begin{equation}\label{eq:13r}\pi^{\fAdv}(f)(\rho_0)\approx^{\tind}_{\epsilon}\underbrace{\fSim_1}_{\text{on }\bS,\bQ^{(out)}}(\tSFSIdeal(f)(\underbrace{\fSim_0}_{\text{on }\bS\text{ and outputs }b}(\rho_0)))\end{equation}
\end{thm}
The following corollary reduces the existence of PRG to one-way functions and encapsulates the theorem using Set-up \ref{setup:1}:
\begin{cor}
Assuming the existence of post-quantum one-way functions, there does not exist a classical SFS protocol as described in Set-up \ref{setup:1} with communication complexity $\fpoly(n,\kappa,1/\epsilon)$.
\end{cor}
The overview of the proof is given in Section \ref{sec:2.6}. The formal proof is as follows.
\begin{proof}[Proof of Theorem \ref{thm:4.1r}]
    Suppose there exists a classical protocol $\pi$ that is complete, efficient, has communication complexity $\fpoly(n,\kappa)$ and satisfies \eqref{eq:13r}. We prove that's impossible.\par
    Consider an adversary $\fAdv$ as follows: $\fAdv$ follows the honest behavior of the protocol, but keeps the transcript of each step. In the end it outputs both the honest setting outputs and the transcripts. Suppose the output register holding the transcripts part is $\bbt$, which holds a classical string in $\{0,1\}^{\fpoly(n,\kappa)}$. By the completeness, for any $n,m,\rho_0$, the output states of $\pi^{\fAdv}(f)(\rho_0)$ should have the following form: the client holds a random $x\leftarrow_r \{0,1\}^n$, the server holds the corresponding $\tPRG(x)\in\{0,1\}^m$ in register $\bQ^{(out)}$, and $\bbt$ holds the transcripts; $\rho_0\in \tD(\cH_{\bS}\otimes\cH_{\bbE})$ is not affected. Finally we use $r\in \{0,1\}^L$ to denote the random coins used by $\fAdv$ and use $\fAdv(r)$ to denote the adversary running with random coin $r$ (recall that $\pi$ is a classical protocol so this could be done); then the discussion above holds with probability $1-\fneg(\kappa)$ for adversary $\fAdv(r)$ with a random $r\leftarrow_r\{0,1\}^L$.\par
    Below we will construct $\fCompressor:\{0,1\}^m\times\{0,1\}^L\rightarrow \{0,1\}^{\fpoly(n,\kappa)}$ and $\fDecompressor:\{0,1\}^{\fpoly(n,\kappa)}\times\{0,1\}^L\rightarrow\{0,1\}^m$. Intuitively they could be understood as a compressor that compresses the PRG outputs to the transcripts in $\bbt$ and a decompressor that decompresses $\bbt$ to PRG outputs.
    \begin{itemize}
    \item By \eqref{eq:13r} against adversary $\fAdv(r),r\in \{0,1\}^L$, there exists efficient quantum operations $(\fSim_0(r),\fSim_1(r))$ such that \eqref{eq:13r} holds.\par
    Define $\fCompressor:\{0,1\}^m\times\{0,1\}^L\rightarrow \{0,1\}^{\fpoly(n,\kappa)}$ as follows. We use $y\in\{0,1\}^m$ to denote its first input and use $r\in \{0,1\}^L$ to denote its second input. $\fCompressor$ first prepares the state $\fSim_0(r)(\rho_0)$, and measures to collapse the register holding $b$; then it stores $y$ in register $\bQ^{(out)}$; then it applies $\fSim_1(r)$; in the end it keeps the states in $\bbt$ and disgards all the other parts.
    By \eqref{eq:13r} we have
    \begin{equation}\label{eq:20}\tDist\begin{bmatrix}x\leftarrow_r\{0,1\}^n,r\leftarrow_r\{0,1\}^L\\ (x,\fCompressor(\tPRG(x),r))\end{bmatrix}\approx_{\epsilon}^{\tind}\tDist\begin{bmatrix}\text{exec }\pi^{\fAdv}(f)(\rho_0)\\ \text{outputs }\bx^{(out)},\bbt\end{bmatrix}\end{equation}
    \item Note that $\pi$ is a classical protocol so there is an efficient algorithm that recovers the outputs given the adversary's code together with its random coins, the initial state, and the messages from the client to the server.\par
     Define $\fDecompressor:\{0,1\}^{\fpoly(n,\kappa)}\times\{0,1\}^L\rightarrow\{0,1\}^m$ as follows. $\fDecompressor$ interprets its first input as the messages from the client to the server and interprets its second input as the server's random coins, simulates the protocol execution and outputs the outputs of $\pi$ on $\bQ^{(out)}$. Note that the initial state $\rho_0$ is not used in the operation of $\fAdv$ and the code of $\fAdv$ could be hard-coded. Then we have:
    \begin{equation}\label{eq:21}
        \tDist\begin{bmatrix}r\leftarrow_r\{0,1\}^L\\\text{exec }\pi^{\fAdv(r)}(f)(\rho_0)\\(x,t):=\text{values in }\bx^{(out)},\bbt\\ (x,\fDecompressor(t,r))\end{bmatrix} \approx_{\fneg(\kappa)}\tDist\begin{bmatrix}x\leftarrow_r\{0,1\}^n\\ (x,\tPRG(x))\end{bmatrix}
    \end{equation}
\end{itemize}
Combining \eqref{eq:20}\eqref{eq:21} we have
$$\tDist\begin{bmatrix}x\leftarrow_r\{0,1\}^n,r\leftarrow_r\{0,1\}^L\\ (x,\fDecompressor(\fCompressor(\tPRG(x),r),r))\end{bmatrix}\approx^{\tind}_{\epsilon}\tDist\begin{bmatrix}x\leftarrow_r\{0,1\}^n\\ (x,\tPRG(x))\end{bmatrix}$$
Note that $\tPRG$ is a determinisitic function. This directly implies that
$$\Pr\begin{bmatrix}
    x\leftarrow_r\{0,1\}^n,r\leftarrow_r\{0,1\}^L\\
    \fDecompressor(\fCompressor(\tPRG(x),r),r)=\tPRG(x)
\end{bmatrix}\geq 1-\epsilon-\fneg(\kappa)$$
Apply the property of PRG we get
$$\Pr\begin{bmatrix}
    u\leftarrow_r\{0,1\}^m,r\leftarrow_r\{0,1\}^L\\
    \fDecompressor(\fCompressor(u,r),r)=u
\end{bmatrix}\geq 1-\epsilon-\fneg(\kappa)$$
which is impossible information-theoretically when $m>\kappa+ \fpoly(n,\kappa)$ and $1-\epsilon\geq O(1)$. This completes the proof.
\end{proof}
\paragraph{Note} We note that Definition \ref{defn:4.2} is stated in a setting where the adversary could be quantum; thus Theorem \ref{thm:4.1r} is about the impossibility of the post-quantum security (classical protocol against quantum adversaries). We could also state the impossibility result in a setting where the adversary is assumed to be classical; our impossibility proof also works in this setting.
\subsection{Secure Function Sampling via RSPV}\label{sec:4.3}
We will construct the quantum protocol for SFS in Section \ref{sec:5}, \ref{sec:6}. This section is a preparation for the construction.\par
As discussed in Section \ref{sec:2.2}, we will view SFS as an RSPV problem, and make use of an intermediate notion called preRSPV. Below we formalize the set-up of the corresponding preRSPV.
\begin{setup}\label{setup:preRSPV}
    The basic set-up is similar to Set-up \ref{setup:1} with the following differences: we are considering a pair of protocols $(\fSFSTest,\fSFSComp)$ instead of a single protocol; and we do not model $1^{1/\epsilon}$ as part of the protocol inputs.\par
    The completeness and soundness are in Definition \ref{defn:4.3}, \ref{defn:4.4} below.
\end{setup}
\begin{defn}\label{defn:4.3}
    We say $(\fSFSTest,\fSFSComp)$ under Set-up \ref{setup:preRSPV} is complete if:
    \begin{itemize}
        \item In $\fSFSTest$, when the server is honest, the passing probability is negligibly close to $1$.
        \item In $\fSFSComp$, when the server is honest, the otuput state of the protocol is as required in Definition \ref{defn:4.1}.
    \end{itemize}
\end{defn}
\begin{defn}\label{defn:4.4}
    We say $(\fSFSTest,\fSFSComp)$ under Set-up \ref{setup:preRSPV} is $(\delta,\epsilon)$-sound if:\par
    For any efficient quantum adversary $\fAdv$, there exists an efficient quantum operation $\fSim$ such that for any initial state $\rho_0\in \tD(\cH_{\bS}\otimes\cH_{\bbE})$:\par
    If
    $$\tr(\Pi_{\fpass}(\fSFSTest^{\fAdv}(\rho_0)))>1-\delta$$
    then 
    $$\Pi_{\fpass}(\fSFSComp^{\fAdv}(\rho_0))\approx_{\epsilon}^{\tind}\Pi_{\fpass}(\underbrace{\fSim}_{\bS,\bQ^{(out)},\bflag}(\underbrace{\rho_{tar}}_{\bx^{(out)},\bQ^{(out)}}\otimes\rho_0))$$
    where $\rho_{tar}$ is as defined in Definition \ref{defn:4.1}.
\end{defn}
As discussed in Section \ref{sec:3.3.7}, the preRSPV could be amplified to an RSPV, which is the SFS that we formalized in Set-up \ref{setup:1}. (Note that although the set-ups are slightly different in the sense that Set-up \ref{setup:prerspvnt} considers a finite set of states while Set-up \ref{setup:preRSPV} works on a function $f$, the amplification still works.)\par
Then we briefly discuss the organization of Section \ref{sec:5} and \ref{sec:6}, where we construct our SFS protocol. We will first formalize the succinct testing for two-party relations in Section \ref{sec:5}, and then formalize the full SFS protocol in Section \ref{sec:6}. In Section \ref{sec:6} we first give the preRSPV $(\fSFSTest,\fSFSComp)$ in Section \ref{sec:6.1.1}, and then amplify it to an SFS in Section \ref{sec:6.1.2}. The intuitions of the protocols are described in Section \ref{sec:techoverv}. The overview of the security proof of $(\fSFSTest,\fSFSComp)$ is given in Section \ref{sec:6.2.1}, which makes use of the notions in \cite{cvqcinlt}. The security proof is completed in Section \ref{sec:6.2.2} and \ref{sec:6.2.3}.\par
Finally we discuss some conventions that might be implicit in later constructions and proofs. First we could review Convention \ref{conv:1}, \ref{conv:2}, \ref{conv:3}, \ref{conv:4}. We additionally note that in the protocol description we usually omit the explicit statement in the form of ``set $\bflag$ to be $\ffail$ if any step fails'' for simplicity: for example, if we are constructing Protocol A and this Protocol A calls Protocol B, we omit ``set $\bflag$ to be $\ffail$ if Protocol B fails''.
\section{Protocol for Reducing Server-to-Client Communications}\label{sec:5}
In this section we formalize the primitive that we call succinct testing for two-party relations, and construct the protocol. As said before, this protocol is a variant of the compiler in \cite{BKLMM22}. In Section \ref{sec:5.1} we formalize the problem set-up; in Section \ref{sec:5.2} we give the protocol construction.
\subsection{Problem Set-up}\label{sec:5.1}
\begin{nota}
    Recall that a relation between $S_1$ and $S_2$ is a function $S_1\times S_2\rightarrow \{0,1\}$. The efficiency is defined in the same way as Definition \ref{nota:3.5}, \ref{nota:3.8}.
\end{nota}
In this section we formalize the notion of succinct testing for two-party relations. Below we first formalize its set-up.
\begin{setup}\label{setup:succ}
The set-up for succinct testing for two-party relations protocols is as follows. We consider a protocol between a client and a server.\par
The inputs, outputs and registers are similar to Set-up \ref{setup:sa} with the following difference: $x\in \{0,1\}^n$ is the client-side input instead of a public input.\par
The completeness and soundness are in Definition \ref{defn:5.1}, \ref{defn:5.2} below. The efficiency is defined in the common way. We additionally want the property that the communication complexity is $\fpoly(n,\kappa)$.
\end{setup}
\begin{defn}[Completeness]\label{defn:5.1}
    We say a protocol under Set-up \ref{setup:succ} is complete if when the server is honest and initially $\bS_w$ holds a value $w$ that satisfies $R(x,w)=1$, the client accepts with probability $\geq 1-\fneg(\kappa)$.
\end{defn}
\begin{defn}[Soundness]\label{defn:5.2}We say a protocol under Set-up \ref{setup:succ} is $(\delta, \epsilon)$-sound if the following holds. For any malicious server $\fAdv$ there exists an efficient quantum operation $\fExt$ such that for any polynomial-size inputs $1^n,1^m,R,x$, for any initial state $\rho_0\in \tD(\cH_{\bS})$, if the client accepts with probability $\geq 1-\delta$, there is
    $$\tr(\Pi_{w:R(x,w)=1}^{\bS_{w}}\fExt(\rho_0))\geq 1-\epsilon-\fneg(\kappa)$$
    
\end{defn}
\subsection{Protocol Design}\label{sec:5.2}
Our succinct testing protocol is as follows. The protocol uses a family of collapsing hash functions as discussed in Section \ref{sec:3.3.2}. Then recall that there exists a succinct argument of knowledge protocol for NP assuming collapsing hash functions (see Section \ref{sec:3.3.4}).
\begin{mdframed}[backgroundcolor=black!10]
    The succinct testing protocol $(\fsuccTest)$ is as follows. 
    \begin{prtl}[$\fsuccTest$]\label{prtl:succtest}This protocol works under Set-up \ref{setup:succ}. The protocol inputs are: public parameters $1^n,1^m,1^\kappa$, public relation $R:\{0,1\}^n\times\{0,1\}^m\rightarrow\{0,1\}$, client-side input $x\in \{0,1\}^n$.\par
        Suppose $\fHash$ is a family of collapsing hash functions. Suppose $\fSAoK$ is a succinct arguments of knowledge protocol for NP that is $(\delta,O(\delta))$-sound for all $\delta>0$.
        \begin{enumerate}
            \item The client samples $h\leftarrow \fHash(1^m,1^\kappa)$ and sends it to the server.
            \item The server computes $h(w)=c$ and sends $c$ to the client.
            \item Using $\fSAoK$, the server proves the following statement to the client:
            $$\exists w: h(w)=c$$
            \item The client reveals $x$ to the server.
            \item Using $\fSAoK$, the server proves the following statement to the client:
            $$\exists w: R(x,w)=1\land h(w)=c$$
        \end{enumerate}
    \end{prtl}
\end{mdframed}
The correctness and the efficiency are from the protocol description. The communication complexity is $\fpoly(n,\kappa)$. The assumption used in the protocol is the collapsing hash functions, used in both the commitment (step 2) and the succinct arguments of knowledge protocol.\par
Below we state the soundness.
\begin{thm}\label{thm:5.1}
   For any $\delta>0$, $\fsuccTest$ is $(\delta,O(\delta))$-sound.
\end{thm}
\begin{proof}
    Consider an efficient quantum adversary $\fAdv$. Suppose on some inputs and initial state $\rho_0$ the protocol passes with probability $\geq 1-\delta$. Then the adversary should pass with probability $\geq 1-\delta$ in both the $\fSAoK$ in step 3 and step 5.\par
    Use $\rho_2$ to denote the joint state between the server and the client. Apply the soundness of $\fSAoK$ to step 3 we know that there exists an efficient quantum operation $\fExt_1$ working on the server side such that
    \begin{equation}\label{eq:23}\tr(\Pi_{w:h(w)=c}^{\bS_{w}}\fExt_1(\rho_2))\geq 1-O(\delta)-\fneg(\kappa)\end{equation}
    without loss of generality we could assume $\fExt_1$ is a unitary operation.\par
Apply the soundness of $\fSAoK$ to step 5 we know that there exists an efficient quantum operation $\fExt_2$ which takes $x$ as input and works on the server side such that
\begin{equation}\label{eq:24}\tr(\Pi_{w:R(x,w)=1\land h(w)=c}^{\bS_{w}}\fExt_2(x)(\rho_2))\geq 1-O(\delta)-\fneg(\kappa)\end{equation}
Now we apply the collapsing property of the hash functions to the state in \eqref{eq:23}. Use $\fM$ to denote the operation that measure the register $\bS_{w}$ and stores the results in register $\bS_{\ttemp}$; here $\bS_{\ttemp}$ is a new register that is not touched by either the adversary or the extractor. Then by the collapsing property we know $\fExt_1(\rho_2)$ is $O(\delta)$-indistinguishable to $\fM(\fExt_1(\rho_2))$ against all the efficient distinguishers that do not touch $\bS_{\ttemp}$. Substituting this to \eqref{eq:24} we get
\begin{equation}\label{eq:25}
    \tr(\Pi_{w:R(x,w)=1\land h(w)=c}^{\bS_{w}}\fExt_2(x)(\fExt_1^{-1}(\fM(\fExt_1(\rho_2)))))\geq 1-O(\delta)-\fneg(\kappa)
\end{equation}
Then since $h$ is collision-resistant the values in $\bS_w$ and $\bS_{\ttemp}$ should be the same on the space that $h(w)=c$:
\begin{equation}\label{eq:252}
    \tr(\Pi^{\bS_w,\bS_{\ttemp}}_{(w_1,w_2):w_1=w_2}\Pi_{w:R(x,w)=1\land h(w)=c}^{\bS_{w}}\fExt_2(x)(\fExt_1^{-1}(\fM(\fExt_1(\rho_2)))))\geq 1-O(\delta)-\fneg(\kappa)
\end{equation}
This further implies that the value in $\bS_{\ttemp}$ should also satisfy $R(x,w)=1$ with high probability, which further implies:
\begin{equation}\label{eq:253}
    \tr(\Pi_{w:R(x,w)=1\land h(w)=c}^{\bS_{w}}(\fExt_1(\rho_2)))\geq 1-O(\delta)-\fneg(\kappa)
\end{equation}
Thus we could define the overall extractor $\fExt$ as follows: $\fExt$ simulates the first and second step of $\fsuccTest$ on its own and applies $\fExt_1$ to extract the witness to register $\bS_{w}$. Then \eqref{eq:253} implies $\tr(\Pi_{w:R(x,w)=1}^{\bS_{w}}(\fExt(\rho_0)))\geq 1-O(\delta)-\fneg(\kappa)$ which is what we want.
\end{proof}
    
\section{Construction of Our Secure Function Sampling Protocol}\label{sec:6}
In this section we construct our SFS protocol. In Section \ref{sec:6.1} we give the protocol construction. In Section \ref{sec:6.2} we give the security proof.
\subsection{Protocol Design}\label{sec:6.1}
As said before, we first construct a preRSPV, and then amplify it to an RSPV which is exactly an SFS.
\subsubsection{Construction of the preRSPV}\label{sec:6.1.1}
\begin{mdframed}[backgroundcolor=black!10]
    The preRSPV $(\fSFSTest, \fSFSComp)$ is as follows. The protocols works under Set-up \ref{setup:preRSPV}. The protocol inputs are: public parameters $1^n,1^m,1^\kappa$, public function $f:\{0,1\}^n\rightarrow \{0,1\}^m$.
    \begin{prtl}[$\fSFSTest$]\label{prtl:3}
    \begin{enumerate}
        \item The client samples and stores the following strings: function inputs $x_0,x_1\in \{0,1\}^n$; input padding $r_0^{(in)},r_1^{(in)}\in \{0,1\}^\kappa$; output padding $r_0^{(out)},r_1^{(out)}\in \{0,1\}^\kappa$. Then the client prepares and sends \begin{equation}\label{eq:1r2n}\frac{1}{\sqrt{2}}(\ket{0}\ket{x_0}\ket{r_0^{(in)}}\ket{r_0^{(out)}}+\ket{1}\ket{x_1}\ket{r_1^{(in)}}\ket{r_1^{(out)}})\end{equation} to the server.
        \item The server evaluates $f$ on the received state and gets the state
        $$\frac{1}{\sqrt{2}}(\underbrace{\ket{0}}_{\bQ^{(sub)}}\underbrace{\ket{x_0}}_{\bQ^{(in)}}\underbrace{\ket{r_0^{(in)}}}_{\bQ^{(inpad)}}\underbrace{\ket{r_0^{(out)}}}_{\bQ^{(outpad)}}\underbrace{\ket{f(x_0)}}_{\bQ^{(out)}}+\ket{1}\ket{x_1}\ket{r_1^{(in)}}\ket{r_1^{(out)}}\ket{f(x_1)}).$$
        The server measures $\bQ^{(in)}$ and $\bQ^{(inpad)}$ on the Hadamard basis; denote the measurement results as $d^{(in)}\in \{0,1\}^n,d^{(inpad)}\in \{0,1\}^\kappa$. The server sends back these results to the client.\par
        The client checks $d^{(inpad)}\neq 0^\kappa$. Then it calculates and stores 
        \begin{equation}\label{eq:32}\theta_0=d^{(in)}\cdot x_0+d^{(inpad)}\cdot r_0^{(in)}\mod 2\end{equation} \begin{equation}\label{eq:33}\theta_1=d^{(in)}\cdot x_1+d^{(inpad)}\cdot r_1^{(in)}\mod 2\end{equation}
        Now the remaining state on the server side is :
        \begin{equation}\label{eq:2r2n}\frac{1}{\sqrt{2}}((-1)^{\theta_0}\ket{0}\ket{r_0^{(out)}}\ket{f(x_0)}+(-1)^{\theta_1}\ket{1}\ket{r_1^{(out)}}\ket{f(x_1)}).\end{equation}
       \item The client randomly chooses to execute one of the following two tests with the server:
       \begin{itemize}\item (Computational basis test) The client asks the server to measure $\bQ^{(sub)}$, $\bQ^{(outpad)}$ and $\bQ^{(out)}$ on the computational basis; denote the measurement results as $c$. The client and the server execute $\fsuccTest$ for the following relation $$c\in\{0||r_0^{(out)}||f(x_0),1||r_1^{(out)}||f(x_1)\}.$$
        \item (Hadamard basis test) The client asks the server to measure $\bQ^{(sub)}$, $\bQ^{(outpad)}$ and $\bQ^{(out)}$ on the Hadamard basis; denote the measurement results as $d^{(sub)}\in \{0,1\},d^{(outpad)}\in \{0,1\}^\kappa,d^{(out)}\in \{0,1\}^m$. The client and the server execute $\fsuccTest$ for the following relation
            \begin{equation}\label{eq:3r2n}d^{(sub)}+d^{(outpad)}\cdot (r_0^{(out)}+r_1^{(out)})+d^{(out)}\cdot(f(x_0)+f(x_1))\equiv \theta_1-\theta_0 \mod 2.\end{equation}
       \end{itemize}
    \end{enumerate}
    \end{prtl}
    \begin{prtl}[$\fSFSComp$]\label{prtl:3n}
        \begin{enumerate}
            \item[1, 2.] The step 1 and step 2 are the same as Protocol \ref{prtl:3}.
            \item[3.] The client asks the server to measure $\bQ^{(sub)}$ and $\bQ^{(outpad)}$ on the computational basis; denote the measurement results as $c$. The server sends back $c$ to the client.\par
             The client checks $c\in \{0||r_0^{(out)},1||r_1^{(out)}\}$. The client chooses $x_0$ as the client-side outputs if $c=0||r_0^{(out)}$ and chooses $x_1$ as the client-side outputs if $c=1||r_1^{(out)}$. The server-side outputs are stored in $\bQ^{(out)}$.\par
        \end{enumerate}
    \end{prtl}
\end{mdframed}
The correctness and efficiency are from the protocol description; the communication complexity is $\fpoly(n,\kappa)$. The assumption used in the protocol is the collapsing hash functions, used in the calling to the succinct testing protocol.\par
The soundness is stated below.
\begin{thm}\label{thm:6.1}
    For any $\delta>0$, $(\fSFSTest,\fSFSComp)$ is $(\delta,\fpoly(\delta))$-sound.
\end{thm}
The proof of Theorem \ref{thm:6.1} is given in Section \ref{sec:6.2}.
\subsubsection{Amplification from preRSPV to SFS}\label{sec:6.1.2}
Below we further amplify the preRSPV to an RSPV which is exactly an SFS. (See Section \ref{sec:4.3} and Section \ref{sec:3.3.7} for related discussions).
\begin{mdframed}[backgroundcolor=black!10]
    The SFS (or RSPV) protocol $\fSFS$ is as follows.
    \begin{prtl}[$\fSFS$]\label{prtl:6}This protocol works under Set-up \ref{setup:preRSPV}. The protocol inputs are: public parameters $1^n,1^m,1^\kappa,1^{1/\epsilon}$, public function $f:\{0,1\}^n\rightarrow \{0,1\}^m$.\par
    The protocol is by applying the amplification procedure in Protocol \ref{prtl:amptempltemp}, \ref{prtl:amptempl} to $(\fSFSTest,\fSFSComp)$.
    \end{prtl}
\end{mdframed}
The correctness and the efficiency are from the protocol description. The communication complexity is $\fpoly(n,\kappa,1/\epsilon)$. The assumption used in the protocol is the collapsing hash functions.\par
Below we state and prove the security.
\begin{thm}\label{thm:6.2}
    Protocol \ref{prtl:6} is $\epsilon$-sound.
\end{thm}
\begin{proof}
    By combining Theorem \ref{thm:6.1} and Theorem \ref{thm:3.10}.
\end{proof}
What remains to be done is the proof of Theorem \ref{thm:6.1}, which is given in the rest of this section.
\subsection{Security Proofs}\label{sec:6.2}
In this section we prove Theorem \ref{thm:6.1}.
\subsubsection{An overview of the proof of Theorem \ref{thm:6.1}}\label{sec:6.2.1}
As discussed in Section \ref{sec:2.4.2}, to prove Theorem \ref{thm:6.1}, we take an approach based on the notions in \cite{cvqcinlt}. We will define a series of ``forms'' of states, and analyze different components of the protocol $(\fSFSTest,\fSFSComp)$. We will prove these different components of the protocol restrict the forms of states step by step. Below we elaborate our approach.\par
First, to analyze the security, we work on the purified joint state between the client and the server, that is, we without loss of generality assume all the states are purified by some reference registers; this makes the security proof easier to work on.\par
Then we focus on the states by the end of the step 2 of the protocol. Note that the honest form of the state is given in \eqref{eq:2r2n}, which is:
\begin{equation}\label{eq:2r2n2}\frac{1}{\sqrt{2}}((-1)^{\theta_0}\ket{0}\ket{r_0^{(out)}}\ket{f(x_0)}+(-1)^{\theta_1}\ket{1}\ket{r_1^{(out)}}\ket{f(x_1)}).\end{equation}
We would like to prove the server's state in the malicious setting should also be close to this state in certain sense. To achieve this we define \emph{basis-honest form} and \emph{basis-phase correspondence form}, which are both sets of states, roughly as follows:
\begin{itemize}
    \item The basis-honest form contains states in the form of 
    $$\ket{0}\ket{r^{(out)}_0}\ket{f(x_0)}\ket{\psi_0}+\ket{1}\ket{r^{(out)}_1}\ket{f(x_1)}\ket{\psi_1}$$
    That is, the server-side registers $\bQ^{(sub)}\bQ^{(outpad)}\bQ^{(out)}$ store the right values as in the honest behavior, but there is no control on the adversary's state on the other registers. Note that we omit the client-side registers and the phases haven't been considered.
    \item The basis-phase correspondence form contains states in the form of 
    $$\sum_{\theta_0,\theta_1\in \{0,1\}^2}\underbrace{\ket{\theta_0}\ket{\theta_1}}_{\text{client side}}(\ket{0}\ket{r^{(out)}_0}\ket{f(x_0)}\ket{\psi_{0,\theta_0}}+\ket{1}\ket{r^{(out)}_1}\ket{f(x_1)}\ket{\psi_{1,\theta_1}})$$
    Here we take the phase information into consideration and require that the $0$-branch does not depend on the value of $\theta_1$ and the $1$-branch does not depend on the value of $\theta_0$.
\end{itemize}
Based on these notions, our approach for analyzing the output state of the step 2 goes as follows:
\begin{enumerate}
    \item If the server could pass the computational basis test from a state $\ket{\psi}$, then $\ket{\psi}$ is roughly (in certain sense) in the basis-honest form up to an isometry.
    \item If $\ket{\psi}$ is prepared by the first two steps of $\fSFSTest$ and is in the basis-honest form, then $\ket{\psi}$ is roughly (in certain sense) in a basis-phase correspondence form.
    \item If $\ket{\psi}$ is in a basis-phase correspondence form and satisfies certain efficiently-preparable property, if it could pass the Hadamard basis test, then $\ket{\psi}$ is roughly (in certain sense) the state \eqref{eq:2r2n}.
\end{enumerate}
In summary different components of $\fSFSTest$ controls the forms of states step by step. (Note that some form of efficiently preparable condition is needed to go through.)\par
Once we have control on the form of output states by the step 2 of $(\fSFSTest,\fSFSComp)$, the analysis of the remaining steps is relatively easy. The third step of $\fSFSComp$ basically forces the server to collapse the state; from the values of the output padding the client could keep the right function input $x$, while the server only holds the corresponding $f(x)$.
\paragraph{Comparison to \cite{cvqcinlt}}We note that many notions are from \cite{cvqcinlt} and we compare both works as follows.
\begin{itemize}\item On the one hand the technical difficulties are much simpler than that of \cite{cvqcinlt}. \cite{cvqcinlt} needs to analyze several tests and to reduce the technical works \cite{cvqcinlt} sacrifices, for example, composable security and efficient simulation. Here we do not need to sacrifice these desirable properties.
    \item On the other hand since we want our protocol to be (sequentially) composable we need careful technical works to construct the simulator.
\end{itemize}
\subsubsection{Definitions and statements}\label{sec:6.2.2}
\begin{setup}\label{setup:sfspreproof} The set-up for analyzing the output state by the step 2 of $(\fSFSTest,\fSFSComp)$ is as follows.\par
    The register naming and set-up is as follows:
\begin{itemize}
    \item The initial state is $\ket{\varphi_{\text{init}}}\in \cH_{\bS}\otimes\cH_{\bbE}$.
    \item The client holds registers $\bx_0$, $\bx_1$, $\btheta_0,\btheta_1$, $\bbbr^{(out)}_0$, $\bbbr^{(out)}_1$. The values in these registers are $x_0,x_1,\theta_0$, etc, as used in Protocol \ref{prtl:3}, \ref{prtl:3n}; so the domain of values in these registers follow the description in the protocol. In the honest setting these values should be all uniformly random.\par
    Use $\bx$ to denote $(\bx_0,\bx_1)$ and $\btheta,\bbbr^{(in)},\bbbr^{(out)}$ are defined similarly. Use $\bd^{\text{(step 2)}}$ to denote the client side registers holding $(d^{(in)},d^{(inpad)})$.
    \item The server holds registers $\bQ^{(sub)}$, $\bQ^{(outpad)}$, $\bQ^{(out)}$. In the honest setting the server's state should be \eqref{eq:2r2n}.
    \item We could assume the following client-side registers have been disgarded: $\bbbr^{(in)},\bd^{\text{(step 2)}}$.
    \item We consider a purification system (or called reference system) $\bR$, which contains the purification system for each classical registers (for example, $\bx$ and $\bbbr^{(out)}$); we also assume the disgarded registers are part of $\bR$. Denote the purification corresponding to $\btheta$ as $\bR_{\btheta}$.
\end{itemize}
Then we define $\Pi_0$ to be the projection onto the space that the value of $\bQ^{(sub)}\bQ^{(outpad)}\bQ^{(out)}$ is equal to $0||\bbbr^{(out)}_0||f(\bx_0)$ (where we use the register symbol to denote the corresponding values in them). Similarly define $\Pi_1$ to be the projection onto the space that the value of $\bQ^{(sub)}\bQ^{(outpad)}\bQ^{(out)}$ is equal to $1||\bbbr^{(out)}_1||f(\bx_1)$. Thus equation \eqref{eq:2r2n} is invariant under the operation of $(\Pi_0+\Pi_1)$.\par
We name different components of the protocol as follows.  We use $\fFirstTwoStep$ to denote the first two steps of $(\fSFSTest,\fSFSComp)$; use $\fSecondStep$ to denote the second step of $(\fSFSTest,\fSFSComp)$; use $\fComputationalTest$ to denote the computational basis test; use $\fHadamardTest$ to denote the Hadamard basis test; use $\fThirdStepComp$ to denote the third step of $\fSFSComp$.\par
Finally we define some special purified joint states. Note that the first step of $(\fSFSTest,\fSFSComp)$ the client prepares and sends a state to the server; denote the purification of this state as
$$\ket{\varphi_{\text{step 1}}}=\sum_{x_0,x_1,r_0^{(in)},r_1^{(in)},r_0^{(out)},r_1^{(out)}\in \{0,1\}^{2n+4\kappa}}\frac{1}{\sqrt{2^{2n+4\kappa}}}\underbrace{\ket{x_0,x_1,r_0^{(in)},r_1^{(in)},r_0^{(out)},r_1^{(out)}}}_{\text{client-side registers }\bx,\bbbr^{(in)},\bbbr^{(out)}}\otimes$$ $$\frac{1}{\sqrt{2}}(\underbrace{\ket{0}\ket{x_0}\ket{r_0^{(in)}}\ket{r_0^{(out)}}}_{\text{server-side registers }\bQ^{(sub)}\bQ^{(in)}\bQ^{(inpad)}\bQ^{(outpad)}}+\ket{1}\ket{x_1}\ket{r_1^{(in)}}\ket{r_1^{(out)}})\underbrace{\ket{\cdots}}_{\text{purification system in }\bR}$$
Then define the target purified joint state after step 2 as
$$\ket{\varphi_{\text{tar step 2}}}=\sum_{x_0,x_1\theta_0,\theta_1,r_0^{(out)},r_1^{(out)}\in \{0,1\}^{2n+2\kappa+2}}\frac{1}{\sqrt{2^{2n+2\kappa+2}}}\underbrace{\ket{x_0,x_1,\theta_0,\theta_1,r_0^{(out)},r_1^{(out)}}}_{\text{client-side registers }\bx,\btheta,\bbbr^{(out)}}\otimes$$
$$\frac{1}{\sqrt{2}}((-1)^{\theta_0}\underbrace{\ket{0}\ket{r_0^{(out)}}\ket{f(x_0)}}_{\text{server-side registers }\bQ^{(sub)}\bQ^{(outpad)}\bQ^{(out)}}+(-1)^{\theta_1}\ket{1}\ket{r_1^{(out)}}\ket{f(x_1)})\underbrace{\ket{\cdots}}_{\text{purification system in }\bR}$$
\end{setup}
\begin{defn}Now we introduce the notions for different forms of states as follows.
\begin{itemize}
    \item We say a state $\ket{\psi}$ in Set-up \ref{setup:sfspreproof} is in basis-honest form if \begin{equation}\label{eq:formintr}(\Pi_0+\Pi_1)\ket{\psi}=\ket{\psi}.\end{equation}
    We call the two terms on the right hand side of \eqref{eq:formintr} as the $0$-branch and the $1$-branch.
    \item We say a state $\ket{\psi}$ in Set-up \ref{setup:sfspreproof} is in basis-phase correspondence form if $\ket{\psi}$ is in the basis-honest form, and for each $b\in \{0,1\}$, its $b$-branch, denoted as $\ket{\psi_b}$, could be written as
    \begin{equation}\label{eq:38r}\ket{\psi_b}=\sum_{\theta_0,\theta_1\in \{0,1\}^2}\underbrace{\ket{\theta_0}\ket{\theta_1}}_{\text{client-side registers }\btheta_0,\btheta_1}\ket{\psi_{b,\theta_b}}\underbrace{\ket{\theta_0}\ket{\theta_1}}_{\text{purification register }\bR_{\btheta}}\end{equation}
    Note that $\ket{\psi_{b,\theta_b}}$ is the same for different values of $\theta_{1-b}$. Thus we say the $0$-branch is independent to the value of $\btheta_1$ and the $1$-branch is independent to the value of $\btheta_0$.
\end{itemize}
\end{defn}
To state the lemmas we need an efficiently-preparable notion.
\begin{defn}
    We say a state family $\sigma$ could be efficiently preparable from $\rho$ if there exists a family of efficient quantum operations $O$ such that $\sigma=O(\rho)$.
\end{defn}
Finally we define the closeness between states up to an operation on register $\bR$:
\begin{defn}\label{defn:upto}
    We say $\ket{\psi_1}\approx_{\epsilon}^{\text{up to }\bR}\ket{\psi_2}$ if there exists an operation $U$ that solely works on $\bR$ such that $U\ket{\psi_1}\approx_{\epsilon}\ket{\psi_2}$.
\end{defn}
\begin{fact}
    $\ket{\psi_1}\approx_{\epsilon}^{\text{up to }\bR}\ket{\psi_2}$, then $\tr_{\bR}(\ket{\psi_1}\bra{\psi_1})\approx_{\epsilon}^{\tind}\tr_{\bR}(\ket{\psi_2}\bra{\psi_2})$.
\end{fact}
Below we state our lemmas for proving Theorem \ref{thm:6.1}.
\begin{lem}\label{lem:6.3}
    If an efficient quantum adversary $\fAdv$ working on the server-side of a state $\ket{\psi}$ could pass the computational-basis test with probability $\geq 1-\delta$, then there exists an efficient quantum operation $O$ working on the server side such that $|(\Pi_0+\Pi_1)O\ket{\psi}|^2\geq 1-\fpoly(\delta)$.
\end{lem}
\begin{lem}\label{lem:6.4}
    Suppose $\fAdv$ is an efficient quantum adversary playing the role of the server. Suppose $\ket{\psi}=\fSecondStep^{\fAdv}(\ket{\varphi_{\text{step 1}}}\otimes\ket{\varphi_{\text{init}}})$ satisfies $|\Pi_{\fpass}(\Pi_0+\Pi_1)\ket{\psi}|^2\geq 1-\fpoly(\delta)$. Then there exists a state $\ket{\tilde\psi}$ that satisfies:
    \begin{itemize}
        \item $\ket{\tilde{\psi}}$ is efficiently-preparable from $\ket{\varphi_{\text{init}}}$.
        \item $\ket{\psi}\approx^{\text{up to }\bR}_{\fpoly(\delta)+\fneg(\kappa)}\ket{\tilde{\psi}}$. (Note that this together with the condition implies $|\Pi_{\fpass}(\Pi_0+\Pi_1)\ket{\tilde{\psi}}|^2\geq 1-\fpoly(\delta)$.)
        \item $\Pi_{\fpass}(\Pi_0+\Pi_1)\ket{\tilde{\psi}}$ is in a basis-phase correspondence form.
    \end{itemize}
\end{lem}
\begin{lem}\label{lem:6.5}
   Suppose $\ket{\psi}$ is $\fpoly(\delta)$-close to a basis-phase correspondence form and is efficiently preparable from $\ket{\varphi_{\text{init}}}$. If an efficient quantum adversary working on the server-side of $\ket{\psi}$ could pass the Hadamard-basis test with probability $\geq 1-\delta$, then there exists an efficient quantum operation $\fSim$ that satisfies
   $$\ket{\psi}\approx_{\fpoly(\delta)+\fneg(\kappa)}(\fI\otimes\underbrace{\fSim}_{\text{server side}})(\ket{\varphi_{\text{tar step 2}}}\otimes\ket{\varphi_{\text{init}}})$$
\end{lem}
\begin{lem}\label{lem:6.6}
    For any efficient quantum operation $\fAdv$ playing the role of the server there exists an efficient quantum operation $\fSim$ playing the role of the server such that \begin{equation}\label{eq:6.6.1}\Pi_{\fpass}(\fThirdStepComp^{\fAdv}(\ket{\varphi_{\text{tar step 2}}}\bra{\varphi_{\text{tar step 2}}}\otimes\ket{\varphi_{\text{init}}}\bra{\varphi_{\text{init}}}))\approx^{\tind}_{\fpoly(\delta)+\fneg(\kappa)}\Pi_{\fpass}(\fSim(\rho_{tar}\otimes\ket{\varphi_{\text{init}}}\bra{\varphi_{\text{init}}}))\end{equation}
    (Note that $\rho_{tar}$ is defined in Definition \ref{defn:4.1}. See also Definition \ref{defn:4.3} for the related security definition.)
\end{lem}
\subsubsection{Formal Security Proofs}\label{sec:6.2.3}
\begin{proof}[Proof of Lemma \ref{lem:6.3}]
    By the soundness of $\fsuccTest$ we know there exists a server-side operation $O$ that extracts a string in $\{0||r_0^{(out)}||f(x_0),1||r_1^{(out)}||f(x_1)\}$ to the register $\bQ^{(sub)}\bQ^{(outpad)}\bQ^{(out)}$, which completes the proof. 
\end{proof}
\begin{proof}[Proof of Lemma \ref{lem:6.4}]
    Consider the state $\ket{\psi_{mid}}$ as follows:
    \begin{enumerate}
        \item Starting from the state $\ket{\varphi_{\text{step 1}}}\otimes\ket{\varphi_{\text{init}}}$, copy the values in $\bQ^{(in)}$ to an auxiliary register $\baux$.
        \item Execute $\fSecondStep^{\fAdv}$ on $\ket{\varphi_{\text{step 1}}}\otimes\ket{\varphi_{\text{init}}}$.
        \item Use the values in $\baux$ to uncompute $\bQ^{(in)}$; swap $\baux$ to $\bQ^{(in)}$.
    \end{enumerate}
    We claim that
    \begin{equation}\label{eq:36}\Pi_{\fpass}(\Pi_0+\Pi_1)\ket{\psi}\approx_{\fneg(\kappa)}\Pi_{\fpass}(\Pi_0+\Pi_1)\ket{\psi_{mid}}\end{equation}
    \eqref{eq:36} holds because if we expand both $\Pi_{\fpass}(\Pi_0+\Pi_1)\ket{\psi}$ and $\Pi_{\fpass}(\Pi_0+\Pi_1)\ket{\psi_{mid}}$ by branches of $\ket{\varphi_{\text{step 1}}}$, the norm of their difference is bounded by $|\Pi_1\fAdv\Pi_0(\ket{\varphi_{\text{step 1}}}\otimes\ket{\varphi_{\text{init}}})|+|\Pi_0\fAdv\Pi_1(\ket{\varphi_{\text{step 1}}}\otimes\ket{\varphi_{\text{init}}})|$, which is $\fneg(\kappa)$. Note that this together with the condition that $|\Pi_{\fpass}(\Pi_0+\Pi_1)\ket{\psi}|^2\geq 1-\fpoly(\delta)$ implies that $\ket{\psi}\approx_{\fpoly(\delta)+\fneg(\kappa)}\ket{\psi_{mid}}$.\par
    The construction of the $\ket{\psi_{mid}}$ is to enforce that for each $b\in \{0,1\}$, the $b$-branch of $\Pi_{\fpass}\ket{\psi_{mid}}$ only comes from the $b$-branch of $\ket{\varphi_{\text{step 1}}}$ thus there is no information about the other branch thus should be independent to the value of $\btheta_{1-b}$. But this intuition is not exactly true: the reason is, $\theta_{1-b}$ is computed from $d^{\text{(step 2)}}$ and $x_{1-b},r^{(in)}_{1-b}$; although the $x_{1-b},r^{(in)}_{1-b}$ registers (together with their purification registers) are initially in a product state from the other parts on the $b$-branch, the measurement of $d^{\text{(step 2)}}$ and calculation of $\theta_{1-b}$ could create some entanglement between these registers. But note that the registers holding $d^{\text{(step 2)}}$, $r^{(in)}$, and the purification register $\bR_{\btheta}$ are all in $\bR$; thus as long as we could find an operation $\fAlignR$ that operates only on these registers that de-entangles $\btheta_{1-b}$, we could still prove the lemma.\par
    Now we construct an operation $\fAlignR$ that operates on $\bbbr^{(in)}$, $\bd^{\text{(step 2)}}$ and $\bR_{\btheta}$ such that $\Pi_{\fpass}(\Pi_0+\Pi_1)\fAlignR\ket{\psi_{mid}}$ is a basis-phase correspondence form. Let's first without loss of generality consider the $0$-branch: note that initially $r^{(in)}_1$ is randomly sampled from $\{0,1\}^\kappa$; then note for each $d^{(inpad)}\neq 0$ (which corresponds to the $\Pi_{\fpass}$ space), for each $\theta_1\in \{0,1\}$, the number of $r^{(in)}_1\in \{0,1\}^\kappa$ such that the value of $\btheta_1$ equals $\theta_1$ is $2^{\kappa-1}$ (see equation \eqref{eq:33}). Thus once we uncompute $\bbbr_1^{(in)}$ with $(\bd^{(inpad)},\bR_{\btheta_1})$ the $0$-branch will be independent to the value of $\btheta_1$. The case for the $1$-branch is the same. In summary we could construct $\fAlignR$ as follows: 
    \begin{enumerate}
        \item Uncompute $\bbbr^{(in)}_1$ with with $\bd^{(inpad)},\bR_{\btheta_1}$.
        \item Uncompute $\bbbr^{(in)}_0$ with with $\bd^{(inpad)},\bR_{\btheta_0}$.
    \end{enumerate}
    Then $\Pi_{\fpass}(\Pi_0+\Pi_1)\fAlignR\ket{\psi_{mid}}$ is a basis-phase correspondence form; thus defining $\ket{\tilde\psi}=\fAlignR\ket{\psi_{mid}}$ proves the lemma.
\end{proof}
The proof of Lemma \ref{lem:6.5} is as follows. Note that the overall ideas of the proof details are related to the discussion in Section \ref{sec:2.3.2}. We add a discussion on the proof details after the proof.
\begin{proof}[Proof of Lemma \ref{lem:6.5}]
    Suppose $\ket{\psi}$ is $\fpoly(\delta)$-close to a state $\ket{\tilde{\psi}}$ in the basis-phase correspondence form. Then by the soundness of $\fsuccTest$ we know there exists a server-side unitary operation $\fExt$ that extracts strings to $\bS_{dstep3}$ such that
    \begin{equation}\label{eq:6.5.1}|\Pi^{\bS_{dstep3}}_{d^{(sub)},d^{(outpad)},d^{(out)}:\text{LHS of \eqref{eq:3r2n}}\equiv \theta_0-\theta_1\mod 2}\fExt\ket{\tilde{\psi}}|^2\geq 1-\fpoly(\delta)\end{equation}
    Recall that the LHS of \eqref{eq:3r2n} is $d^{(sub)}+d^{(outpad)}\cdot(r^{(out)}_0+r^{(out)}_1)+d^{(out)}\cdot(f(x_0)+f(x_1))$.\par
    Introduce the following notation to simplify the expression: 
    \begin{itemize}\item Use $\Pi_{\equiv}$ to denote the projection onto the space that the value in $\bS_{dstep3}$ satisfies $\text{LHS of \eqref{eq:3r2n}}\equiv \theta_0-\theta_1\mod 2$; use $\Pi_{\not\equiv}$ to denote the complement of $\Pi_{\equiv}$.
        \item Use $\ket{\tilde\psi_0}$ to denote the $0$-branch of $\ket{\tilde\psi}$ and use $\ket{\tilde\psi_1}$ to denote the $1$-branch of $\ket{\tilde\psi}$. Then recall equation \eqref{eq:38r}, and defines $\ket{\tilde\psi_{0,\theta_0}}$ for each $\theta_0\in \{0,1\}$ and $\ket{\tilde\psi_{1,\theta_1}}$ for each $\theta_1\in \{0,1\}$ as \eqref{eq:38r}.
    \end{itemize} 
    Since $\ket{\tilde{\psi}}$ is in the basis-phase correspondence form, for each branch of this state there is:
    \begin{equation}\label{eq:6.5.2}\forall b\in \{0,1\},|\Pi_{\equiv}\fExt\ket{\tilde{\psi}_b}|\leq \frac{1}{\sqrt{2}}|\ket{\tilde{\psi}_b}|\end{equation}
    Combining \eqref{eq:6.5.1}\eqref{eq:6.5.2} we have
\begin{equation}\label{eq:42op}\Pi_{\equiv}\fExt\ket{\tilde\psi_0}\approx_{\fpoly(\delta)}\Pi_{\equiv}\fExt\ket{\tilde\psi_1}\end{equation}
\begin{equation}\label{eq:43op}\Pi_{\not\equiv}\fExt\ket{\tilde\psi_0}\approx_{\fpoly(\delta)}-\Pi_{\not\equiv}\fExt\ket{\tilde\psi_1}\end{equation}
Expanding $\ket{\tilde\psi_b}$ to $\ket{\tilde\psi_{b,\theta_b}}$ and making use of the form of $\ket{\tilde\psi}$, \eqref{eq:42op}\eqref{eq:43op} further imply:
\begin{equation}\label{eq:44op}\ket{\tilde{\psi}_{0,0}}\approx_{\fpoly(\delta)}-\ket{\tilde{\psi}_{0,1}}\end{equation}
\begin{equation}\label{eq:45op}\ket{\tilde{\psi}_{1,0}}\approx_{\fpoly(\delta)}-\ket{\tilde{\psi}_{1,1}}\end{equation}
\begin{equation}\label{eq:46op}U_{f(\bx),\bbbr^{(out)}}\ket{\tilde{\psi}_{0,0}}\approx_{\fpoly(\delta)}\ket{\tilde{\psi}_{1,0}}\end{equation}
where $U_{f(\bx),\bbbr^{(out)}}$ is defined to be the following operation:
\begin{enumerate}
    \item Apply $\fExt$;
    \item Add a $(-1)$ phase on the space of $\Pi_{\not\equiv}$. Note that this step requires the information of $f(x_0),f(x_1),r^{(out)}_0,r^{(out)}_1$ so we name $U$ as above.
    \item Apply $\fExt^{-1}$ (note that we chose $\fExt$ to be a unitary above).
\end{enumerate}
Now we construct the $\fSim$ as required in the statement of this lemma. Below we give the construction and analyze each step during the construction.\par
$\fSim$ works as follows:
\begin{enumerate}
    \item First prepare the state $\ket{\psi}$ from $\ket{\varphi_{\text{init}}}$ on its own register; here all the client-side registers, including $\bx,\bbbr^{(out)},\btheta$, are simulated within the simulator's memory; the server-side parts like $\bQ^{(sub)},\bQ^{(outpad)},\bQ^{(out)}$ are also simulated in new corresponding registers. Denote these new registers used for simulation as $\bx^\prime$, $\bbbr^{(out)\prime}$, $\btheta^\prime$, $\bQ^{(sub)\prime}$, etc.\par
    Suppose the state created in this step is $\ket{\psi^\prime}$. The difference between $\ket{\psi}$ and $\ket{\psi^\prime}$ is simply the location of registers (in more detail, $\bx^\prime$, $\bbbr^{(out)\prime}$, $\btheta^\prime$, $\bQ^{(sub)^\prime}$, $\bQ^{(outpad)\prime}$, $\bQ^{(out)\prime}$). Then we could define $\ket{\tilde\psi^\prime}$, $\ket{\tilde\psi_{b,\theta}^\prime}$ etc similarly. By the previous proof $\ket{\psi^\prime}$ is also $\fpoly(\delta)$-close to $\ket{\tilde\psi^\prime}$ and \eqref{eq:44op}\eqref{eq:45op}\eqref{eq:46op} all hold after replacing all the $\tilde\psi_{\cdots}$ by the corresponding $\tilde\psi_{\cdots}^\prime$.
    \item Do an xor operation between $\bQ^{(sub)}$ and $\bQ^{(sub)\prime}$. Discuss by cases on the values of $\bQ^{(sub)}$ and the xor result:
    \begin{itemize}
        \item Suppose the xor is $1$:\begin{itemize}\item On the $0$-branch of $\ket{\varphi_{\text{tar step 2}}}$ (in other words, $\bQ^{(sub)}$ has value $0$): now the $0$-branch could be written as (omitting the client-side registers)
        \begin{equation}
            \underbrace{\ket{0}\ket{r_0^{(out)}}\ket{f(x_0)}}_{\bQ^{(sub)}\bQ^{(outpad)}\bQ^{(out)}}\ket{\tilde\psi_1^\prime}
        \end{equation}
        We would like to transform this state to $\ket{\tilde\psi_0}$. The differences are as follows: (1) the obvious difference is that $\ket{\tilde\psi_1^\prime}$ is the $1$-branch while we want a $0$-branch; (2) the less obvious difference is $\ket{\tilde\psi_1^\prime}$ corresponds to the ``simulated client-side registers'' like $\bx^\prime$, $\btheta^\prime$, $\bbbr^{(out)\prime}$ etc while in what we want ($\ket{\tilde\psi_0}$) the server-side state should be related to the real client-side regsiters $\bx$, $\btheta$, $\bbbr^{(out)}$.\par
        We notice that the following operations allow us to achieve this transformation. \begin{enumerate}\item First by \eqref{eq:44op}\eqref{eq:45op} we could remove the additional phase in $\ket{\tilde\psi_{1,1}^\prime}$ (that is, apply a $(-1)$ phase controlled by $\btheta^\prime_1$).\item Then by \eqref{eq:46op} apply $U_{f(\bx_0),f(\bx_1^\prime),\bbbr_0^{(out)},\bbbr_1^{(out)^\prime}}$ to map this branch of simulated state to $\ket{\tilde\psi_{0}}$ (omitting temporary registers like $\bx^\prime,\btheta^\prime,\bbbr^{(out)\prime}$, etc; they are now in a fixed product state so we could ignore them).\par In more detail, we make use of values in the following registers to implement $U$: (1) the values in registers $\bQ^{(outpad)},\bQ^{(out)}$ (on the server-side of $\ket{\varphi_{\text{tar step 2}}}$, corresponding to $f(\bx_0)$ and $\bbbr_0^{(out)}$ in the subscripts of $U$); (2) the values in the ``simulated client-side registers'' $\bx_{1}^\prime,\bbbr_1^{(out)\prime}$. 
            \item Now the values of $\bQ^{(sub)}\bQ^{(outpad)}\bQ^{(out)}$ is the same as $\bQ^{(sub)\prime}\bQ^{(outpad)\prime}\bQ^{(out)\prime}$ so $\bQ^{(sub)\prime}\bQ^{(outpad)\prime}\bQ^{(out)\prime}$ could be uncomputed.
        \end{enumerate}
        \item The $1$-branch is similar. Note that we could control on the $\bQ^{(sub)}$ to map the $0$-branch first and $1$-branch next.
        \end{itemize}
        \item Suppose the xor is $0$: now for example, the $0$-branch could be written as (omitting the client-side registers)
        \begin{equation}
            \ket{0}\ket{r_0^{(out)}}\ket{f(x_0)}\ket{\tilde\psi_0^\prime}
        \end{equation}
        We note that in the two differences discussed above, although there is only one difference remaining here, applying $U_{f(\bx_0),f(\bx_1^\prime),\bbbr_0^{(out)},\bbbr_1^{(out)^\prime}}$ as above does not solve the problem. Here our solution is, in the step (b) above, first apply $U_{f(\bx_0^\prime),f(\bx_1^\prime),\bbbr_0^{(out)^\prime},\bbbr_1^{(out)^\prime}}$ to map it to the simulated $1$-branch, and then apply $U_{f(\bx_0),f(\bx_1^\prime),\bbbr_0^{(out)},\bbbr_1^{(out)^\prime}}$ as above.\par
        The $1$-branch is similar.
    \end{itemize}
    \item Finally uncompute the intermediate registers (like $\bx^\prime,\btheta^\prime,\bbbr^{(out)\prime}$) completes the simulation.
\end{enumerate}
\end{proof}
\paragraph{Note on the intuition of the proof} We note that one way to understand the details of this proof is to return to the discussion in Section \ref{sec:2.3.2}. Note that in the discussion in \ref{sec:2.3.2} the client-side Hadamard measurement is assumed to be honest (which introduces phases), while in the condition of this lemma we assume the output state from step 2 of $\fSFSTest$ is in a basis-phase correspondence form. Although the conditions are not the same, the techniques (like the construction of $U$) share similar ideas. One additional difficulty is that in this proof we want the final result to have good (sequential) composability so we need to take the initial state $\ket{\varphi_{\text{init}}}$ into consideration. To solve this problem, we construct $\fSim$ as follows: we first solely do the simulation from $\ket{\varphi_{\text{init}}}$, then xor it with the target state $\ket{\varphi_{\text{tar step 2}}}$; finally we make use of $U$ (or in more detail, \eqref{eq:44op}\eqref{eq:45op}\eqref{eq:46op}) to transform the state to the state we want.
\begin{proof}[Proof of Lemma \ref{lem:6.6}]
    Notice that $r_{1-b}^{(out)}$ is unpredictable from the $b$-branch. This implies that the left hand side of \eqref{eq:6.6.1} is negligibly close to the following state on the passing space:
    \begin{enumerate}
        \item Starting from the state $\ket{\varphi_{\text{tar step 2}}}\otimes\ket{\varphi_{\text{init}}}$, do a projection on $\bQ^{(sub)}\bQ^{(outpad)}\bQ^{(out)}$ and store the corresponding $x$ on the client-side register $\bx^{(out)}$. Keep a copy of the value in $\bQ^{(sub)}\bQ^{(outpad)}\bQ^{(out)}$ in a temporary register $\baux$.
        \item Apply $\fAdv$ as given in the lemma statement.
        \item Project onto the space that the value of $\bQ^{(sub)}\bQ^{(outpad)}\bQ^{(out)}$ is the same as $\baux$; mark $\bflag$ as $\fpass$ if the projection succeeds.\par
        Disgard $\baux$.
    \end{enumerate}
    Notice that the initial state in $(\bx^{(out)},\bQ^{(out)})$ is exactly $\rho_{tar}$. Thus the operation above gives the $\fSim$ we want, which completes the proof.
\end{proof}
\begin{proof}[Proof of Theorem \ref{thm:6.1}]
    First for any initial state $\rho_0\in \tD(\cH_{\bS}\otimes\cH_{\bbE})$ we could always purify it to $\ket{\varphi_{\text{init}}}\in \cH_{\bS}\otimes\cH_{\bbE}\otimes\cH_{\bbE^\prime}$; thus we could only consider pure states as in Set-up \ref{setup:sfspreproof}. Then if an adversary $\fAdv$ could pass with probability $\geq 1-\delta$ in $\fSFSTest$ it should pass with probability $\geq 1-\fpoly(\delta)$ in each step of $\fSFSTest$. Chaining Lemma \ref{lem:6.3}, \ref{lem:6.4}, \ref{lem:6.5} and \ref{lem:6.6} completes the proof.
\end{proof}
\section{Protocols for General Two-party Computations}\label{sec:7}
In this section we give the generalizations and improvements.
\begin{itemize}
    \item In Section \ref{sec:7.1} we construct the SFS protocol for randomized functions.
    \item In Section \ref{sec:7.2} we construct the SFVP protocol, where the function input $x$ is provided by the client.
    \item In Section \ref{sec:7.3} we construct the protocol for secure two-party computations in succinct communication.
    \item In Section \ref{sec:7.4} we construct the protocol for classical-channel secure two-party computations in succinct communication.
\end{itemize}
We refer to Section \ref{sec:2.7} for an overview; in this section we formalize the protocols and security proofs.
\subsection{SFS for Randomized Functions}\label{sec:7.1}
The set-up of our protocol is as follows.
\begin{setup}\label{setup:3}
    The set-up is similar to Set-up \ref{setup:1} with the following differences. The function $f$ is an efficient classical randomized function $f:\{0,1\}^n\times\{0,1\}^L\rightarrow \{0,1\}^m$ (as described in Fact \ref{fact:2}; the second input is the random coins). $1^L$ is given as a public parameter.\par
    The completeness and soundness are in Definition \ref{defn:7.1}, \ref{defn:7.2} below.
\end{setup}
\begin{defn}\label{defn:7.1}
    We say a protocol $\pi$ under Set-up \ref{setup:3} is complete if when the server is honest, the protocol passes with probability $\geq 1-\fneg(\kappa)$ and furthermore, denote the output state as $\rho_{out}\in \tD(\cH_{\bx^{(out)}}\otimes\cH_{\bQ^{(out)}})$, there is
    $$\rho_{out}\approx^{\tind}\rho_{tar}$$
    where $\rho_{tar}$ is defined to be the following state: the register $\bx^{(out)}$ holds a random classical string $x\leftarrow_r\{0,1\}^n$, the register $\bQ^{(out)}$ holds a sampling of $f(x;r)$ where $r\leftarrow_r\{0,1\}^L$.
\end{defn}
Note that the completeness here is also defined using the indistinguishability notion.\par
The soundness is defined in the same way as Definition \ref{defn:4.2} with the difference that the function called by the ideal functionality is now randomized.
\begin{defn}\label{defn:7.2}
    To define the soundness, first define $\tSFSIdeal$ as follows.\par
    $\tSFSIdeal$ receives a bit $b\in \{0,1\}$ from the server, then:
    \begin{itemize}
        \item If $b=0$, $\tSFSIdeal$ sets $\bflag$ to be $\fpass$ and prepares $\rho_{tar}$ in $\bx^{(out)}$, $\bQ^{(out)}$.
        \item If $b=1$, $\tSFSIdeal$ sets $\bflag$ to be $\ffail$.
    \end{itemize}
    We say a protocol $\pi$ under Set-up \ref{setup:3} is $\epsilon$-sound if:\par
    For any efficient adversary $\fAdv$, there exist efficient quantum operations $\fSim=(\fSim_0,\fSim_1)$ such that for any polynomial-size inputs $1^n,1^m,1^L,1^{1/\epsilon}, f$ and any initial state $\rho_0\in \tD(\cH_{\bS}\otimes\cH_{\bbE})$:
    $$\pi^{\fAdv}(\rho_0)\approx^{\tind}_{\epsilon}\underbrace{\fSim_1}_{\text{on }\bS,\bQ^{(out)}}(\tSFSIdeal(\underbrace{\fSim_0}_{\text{on }\bS\text{ and outputs }b}(\rho_0)))$$
\end{defn}
Our SFS protocol for randomized functions is as follows. The protocol uses the pseudorandom generator (PRG) primitive as discussed in Section \ref{sec:3.3.1}.
\begin{mdframed}[backgroundcolor=black!10]
    \begin{prtl}[SFS for randomized functions]\label{prtl:7}
        This protocol works under Set-up \ref{setup:3}. The protocol inputs are: public parameters $1^n,1^m,1^L,1^\kappa,1^{1/\epsilon}$, public function $f:\{0,1\}^n\times\{0,1\}^L\rightarrow \{0,1\}^m$.\par The protocol makes use of a pseudorandom generator $\tPRG(1^\kappa,1^L):\{0,1\}^\kappa\rightarrow \{0,1\}^L$. Below we omit the parameters.
        \begin{enumerate}
            \item The client and the server execute protocol $\fSFS$ (Protocol \ref{prtl:6}) with error tolerance $\epsilon$ for function $g:\{0,1\}^{n+\kappa}\rightarrow \{0,1\}^m$ defined as follows:
            $$g(x,s):=f(x;\tPRG(s))\text{ where }(x,s)\in \{0,1\}^{n}\times\{0,1\}^{\kappa}$$
            When the server is honest the client should get a random $x,s\leftarrow_r \{0,1\}^{n}\times\{0,1\}^{\kappa}$ and the server gets the corresponding $g(x,s)$. The server stores $g(x,s)$ in $\bQ^{(out)}$ as the final outputs; the client disgards $s$ and stores $x$ in $\bx^{(out)}$ as the final outputs.
        \end{enumerate}
    \end{prtl}
\end{mdframed}
The communication is succinct by the property of $\fSFS$ and the fact that the size of $(x,s)$ is only $n+\kappa$. The assumption is the collapsing hash functions, used in the $\fSFS$ protocol; note that PRG is implied by the existence of collapsing hash functions.
\begin{thm}
    Protocol \ref{prtl:7} is complete.
\end{thm}
\begin{proof}
    By the completeness of $\fSFS$ (Protocol \ref{prtl:6}) and the property of PRG (Definition \ref{defn:PRG}).
\end{proof}
The soundness is below.
\begin{thm}\label{thm:7.2}
    Protocol \ref{prtl:7} is $\epsilon$-sound.
\end{thm}
\begin{proof}
    For any efficient quantum adversary $\fAdv$, by Theorem \ref{thm:6.2} there exist efficient quantum operations $\fSim$ such that $\fSFS^{\fAdv}$ is $\epsilon$-indistinguishable to $\fDisgard(\bbs)\circ\fSim\circ\tSFSIdeal(g)$, where $\fDisgard(\bbs)$ denotes the client-side operation that disgards the $s$ value (in the last step of Protocol \ref{prtl:7}). Then by the property of PRG, $\fDisgard(\bbs)\circ\tSFSIdeal(g)$ is indistinguishable to $\tSFSIdeal(f)$ (that is, we replace the pseudorandom coins by the real random coins). Combining them completes the proof.
\end{proof}
 \subsection{Secure Function Value Preparation}\label{sec:7.2}
In this section we construct the secure function value preparation (SFVP) protocol, where the function input is provided by the client as its private input. The set-up is formalized below.
\begin{setup}\label{setup:4}
    The set-up for SFVP is similar to Set-up \ref{setup:1} with the following differences. Besides the inputs in Set-up \ref{setup:1}, the protocol also takes a private input $x\in \{0,1\}^n$ from the client. Then the client-side output register is only $\bflag$.\par
    The completeness and soundness are in Definition \ref{defn:7.3}, \ref{defn:7.4} below.
\end{setup}
\begin{defn}\label{defn:7.3}
We say a protocol under Set-up \ref{setup:4} is complete if when the server is honest, for any client-side input $x$, the output state of the protocol is negligibly close to the following state: the register $\bflag$ holds the value $\fpass$, and the register $\bQ^{(out)}$ holds the classical string $f(x)$.
\end{defn}
\begin{defn}\label{defn:7.4}
    To define the soundness, we first define $\tSFVPIdeal$ as follows.\par
    $\tSFVPIdeal$ receives a bit $b\in\{0,1\}$ from the server, then:
        \begin{itemize}\item If $b=0$, $\tSFVPIdeal$ sets $\bflag$ to be $\fpass$ and stores $f(x)$ in $\bQ^{(out)}$.
            \item If $b=1$, $\tSFVPIdeal$ sets $\bflag$ to be $\ffail$.
        \end{itemize}
    Then we define the soundness as follows. We say a protocol $\pi$ under Set-up \ref{setup:4} is $\epsilon$-sound if:\par
For any efficient quantum adversary $\fAdv$, there exist efficient quantum operations $\fSim=(\fSim_0,\fSim_1)$ such that for any polynomial-size inputs $1^n,1^m,1^{1/\epsilon},f,x$ and any initial state $\rho_0\in \tD(\cH_{\bS}\otimes\cH_{\bbE})$:
\begin{equation}\label{eq:13nq}\pi^{\fAdv}(\rho_0)\approx^{\tind}_{\epsilon}\underbrace{\fSim_1}_{\text{on }\bS,\bQ^{(out)}}(\tSFVPIdeal(\underbrace{\fSim_0}_{\text{on }\bS\text{ and outputs }b}(\rho_0)))\end{equation}
\end{defn}
The protocol is as follows. The protocol uses a decomposable garbling scheme as discussed in Section \ref{sec:3.3.5}.
\begin{mdframed}[backgroundcolor=black!10]
    \begin{prtl}[$\fSFVP$]\label{prtl:10}
        This protocol works under Set-up \ref{setup:4}. The protocol inputs are: public parameters $1^n,1^m,1^\kappa,1^{1/\epsilon}$, public function $f$, client-side input $x\in \{0,1\}^n$.\par
        The protocol makes use of a decomposable randomized encoding scheme $(\fGarbleC,\fGarbleI,\fDc)$.
        \begin{enumerate}
            \item Use $C_f$ to denote the circuit representation of $f$.\par
            The client and the server execute the SFS protocol (Protocol \ref{prtl:7})\footnote{Note that $g$ is a randomized function since $\fGarbleC$ could have internal random coins; see Section \ref{sec:3.3.5} for details.} for function $g$: $\{0,1\}^{n\kappa}\rightarrow \{0,1\}^{\fpoly(|C_f|,\kappa)}$:
            $$g(r):=\fGarbleC(C_f,r)$$
            with error tolerance $\epsilon$.\par
            If the server is honest the client should get a random $r\leftarrow_r\{0,1\}^{n\kappa}$ and the server gets a sample of $\fGarbleC(C_f,r)$.
            \item The client computes and sends $\fGarbleI(x,r)$ to the server.\par
            The honest server could compute $\fDc(\fGarbleC(C_f,r),\fGarbleI(x,r))$ to get $f(x)$.
        \end{enumerate}
    \end{prtl}
\end{mdframed}
The communication is succinct by the property of Protocol \ref{prtl:7} and the fact that $(x,r)$ is succinct. The assumption is the collapsing hash functions by the constructions of Protocol \ref{prtl:7} and Yao's garbled circuits. The completeness is from the completeness of Protocol \ref{prtl:7} and completeness of Yao's garbled circuits; note that since the final output of the protocol is determinisitic, the indistinguishability in the completeness definition of Protocol \ref{prtl:7} (Definition \ref{defn:7.1}) does not appear here.
\begin{thm}\label{thm:sfvps}
    Protocol \ref{prtl:10} is $\epsilon$-sound.
\end{thm}
\begin{proof}
    For any efficient quantum adversary $\fAdv$, we would like to construct efficient quantum operations $\fSim$ that satisfies Definition \ref{defn:7.4}.\par
     First by the soundness of the $\fSFS$ protocol (Theorem \ref{thm:7.2}) we know that, there exist efficient quantum operations $\fSim_{\text{SFS}}=(\fSim_{\text{SFS,}0},\fSim_{\text{SFS,}1})$ working on the server side such that  for any polynomial-size inputs and any initial state $\rho_0$,
     \begin{equation}\label{eq:37}
        \fSFVP^{\fAdv}(\rho_0)\approx_{\epsilon}^{\tind}\fHybrid(\rho_0)
     \end{equation}
     where the operation $\fHybrid$, operated on $\rho_0$, is defined as follows:
    \begin{enumerate}
        \item Apply $\fSim_{\text{SFS,}0}$ to $\rho_0$, which generates $b\in\{0,1\}$ as the server-side input to $\tSFSIdeal$.
        \item Execute $\tSFSIdeal$ for function $g$. (Recall the definition of $\tSFSIdeal$ in Definition \ref{defn:7.2}.)
        \item Apply $\fSim_{\text{SFS,}1}$. Then apply $\fAdv_1$, which denotes the adversary's operation after the first step.
        \item The client computes and sends $\fGarbleI(x,r)$ to the server.
        \item Apply $\fAdv_2$, which denotes the adversary's operation after the second step.
    \end{enumerate}
    Then by the soundness of Yao's garbled circuits (Definition \ref{defn:3.gc}) there exists an efficient quantum operation $\fSim_{\text{GC}}$ such that \begin{equation}\label{eq:38}\fSim_{\text{GC}}(C_f(x))\approx^{\tind}\tDist\begin{bmatrix}r\leftarrow_r\{0,1\}^{n\kappa}\\ (\fGarbleC(C_f,r),\fGarbleI(x,r))\end{bmatrix}\end{equation}
    Now we are ready to construct the simulator $\fSim=(\fSim_0,\fSim_1)$ as follows:
    \begin{itemize}
    \item $\fSim_0$:
    \begin{enumerate}
        \item  Apply $\fSim_{\text{SFS,}0}$ to $\rho_0$, which generates $b\in\{0,1\}$ as the server-side input to $\tSFVPIdeal$.
    \end{enumerate}
    \item $\fSim_1$:
    \begin{itemize}
        \item If $b=0$:
        \begin{enumerate}
            \item Use the $C_f(x)$ that it receives from $\tSFVPIdeal$ to prepare $\fSim_{\text{GC}}(C_f(x))$; call its outputs ``simulated circuit encoding'' and ``simulated input encoding''.
            \item  Apply $\fAdv_1\circ\fSim_{\text{SFS,}1}$, but use the simulated circuit encoding to replace the $\fGarbleC(C_f,r)$ (outputted by the end of step 1 of $\fSFVP$ protocol or from $\tSFSIdeal$ in $\fHybrid$).
            \item Apply $\fAdv_2$, but use the simulated input encoding to replace the $\fGarbleI(x,r)$ (transmitted in step 2 of $\fSFVP$ protocol).
        \end{enumerate}
        \item If $b=1$: simply apply $\fAdv_2\circ\fAdv_1\circ\fSim_{\text{SFS,}1}$ on the server-side state.
    \end{itemize}
\end{itemize}
Comparing $\fHybrid$ and $\fSim_1\circ\tSFVPIdeal\circ\fSim_0$ and use \eqref{eq:38} we could see
\begin{equation}\label{eq:39}\fHybrid(\rho_0)\approx^{\tind}\fSim_1(\tSFVPIdeal(\fSim_0(\rho_0)))\end{equation}
Combining \eqref{eq:37}\eqref{eq:39} we have
$$\fSFVP^{\fAdv}(\rho_0)\approx^{\tind}_{\epsilon}\fSim_1(\tSFVPIdeal(\fSim_0(\rho_0)))$$
which completes the proof.
\end{proof}
\subsection{Secure Two-party Computations Protocol}\label{sec:7.3}
In this section we give our secure two-party computations (2PC) protocol in succinct communication. As discussed in Section \ref{sec:2.7.3}, this is achieved in two steps:\par
\begin{itemize}
    \item In Section \ref{sec:7.3.1} we construct the secure function value preparation (SFVP) protocol with certain security against malicious clients.
    \item In Section \ref{sec:7.3.2} we construct the protocol for secure two-party computations.
\end{itemize}
\subsubsection{Secure function value preparation with soundness against malicious clients}\label{sec:7.3.1}
In this section we construct the SFVP protocol with certain security against malicious clients. As discussed in Section \ref{sec:2.7.3}, we consider a special set-up where the client-side input $x$ has been committed to the server. The set-up is formalized below.
\begin{setup}\label{setup:11}
    This set-up is similar to Set-up \ref{setup:4} (the set-up for ordinary SFVP) with the following differences.\par
For the client-side inputs, besides the function input $x$, there is an additional input $r\in \{0,1\}^{n\kappa}$ used for the opening of the commitment; for the public inputs, there are additionally $pp\in \{0,1\}^{3n\kappa}$ and the commitment $com=\fCommit(x,r,pp)$ (see Section \ref{sec:3.3.comm} for definitions of $\fCommit$; note that $\fCommit$ is statistically binding); here $r,pp$ should be considered to be randomly sampled (in both the completeness and soundness definitions below).\par
For the output registers, the server also holds a flag register $\bflag^{(S)}$ (in addition to the output register $\bQ^{(out)}$). The client-side flag register is denoted as $\bflag^{(C)}$.\par
Either the client or the server could be malicious. For modeling the initial states in the malicious server setting, suppose the overall server-side registers are denoted by $\bS$ and the environment is denoted by $\bbE$; for modeling the initial states in the malicious client setting, suppose the overall client-side registers are denoted by $\bbC$ and the environment is denoted by $\bbE$.\par
The completeness, soundness against malicious server and soundness against malicious client are defined in Definition \ref{defn:7.5}, \ref{defn:7.6} and \ref{defn:7.7} below.\par
We still want the property that the communication complexity is $\fpoly(n,\kappa,1/\epsilon)$.
\end{setup}
\begin{defn}\label{defn:7.5}
    We say a protocol under Set-up \ref{setup:11} is complete if when both the client and the server are honest, the output state of the protocol is negligibly close to the following state: $\bflag^{(C)},\bflag^{(S)}$ both hold the value $\fpass$, $\bQ^{(out)}$ holds the classical string $f(x)$.
\end{defn}
The soundness against malicious servers is as follows.
\begin{defn}\label{defn:7.6}
    To define the soundness against malicious servers, we first define $\tSFVPTwoIdealS$ as follows.\par
    $\tSFVPTwoIdealS$ receives a bit $b\in\{0,1\}$ from the server, then:
    \begin{itemize}\item If $b=0$, $\tSFVPTwoIdealS$ sets $\bflag^{(C)}$ to be $\fpass$ and stores $f(x)$ in $\bQ^{(out)}$.
        \item If $b=1$, $\tSFVPTwoIdealS$ sets $\bflag^{(C)}$ to be $\ffail$.
    \end{itemize}
Then we define the soundness as follows. We say a protocol $\pi$ under Set-up \ref{setup:11} is $\epsilon$-sound against malicious servers if:\par
For any efficient quantum adversary $\fAdv$ playing the role of the server, there exist efficient quantum operations $\fSim=(\fSim_0,\fSim_1)$ such that for any polynomial-size inputs that satisfy the conditions in the set-up and any initial state $\rho_0\in \tD(\cH_{\bS}\otimes\cH_{\bbE})$:
\begin{equation}\label{eq:13nqn}\pi^{\fAdv}(\rho_0)\approx^{\tind}_{\epsilon}\underbrace{\fSim_1}_{\text{on }\bS,\bQ^{(out)}}(\tSFVPTwoIdealS(\underbrace{\fSim_0}_{\text{on }\bS\text{ and outputs }b}(\rho_0)))\end{equation}
\end{defn}
The soundness against malicious client is as follows.
\begin{defn}\label{defn:7.7}
    To define the soundness against malicious clients, we first define $\tSFVPTwoIdealC$ as follows.\par
    $\tSFVPTwoIdealC$ receives a bit $b\in\{0,1\}$ from the client, then:
    \begin{itemize}\item If $b=0$, $\tSFVPTwoIdealC$ sets $\bflag^{(S)}$ to be $\fpass$ and stores $f(x)$ in $\bQ^{(out)}$.
        \item If $b=1$, $\tSFVPTwoIdealC$ sets $\bflag^{(S)}$ to be $\ffail$.
    \end{itemize}
Then we define the soundness as follows. We say a protocol $\pi$ under Set-up \ref{setup:11} is $\epsilon$-sound against malicious clients if:\par
For any efficient quantum adversary $\fAdv$ playing the role of the client, there exist efficient quantum operations $\fSim=(\fSim_0,\fSim_1)$ such that for any polynomial-size inputs that satisfy the conditions in the set-up and any initial state $\rho_0\in \tD(\cH_{\bbC}\otimes\cH_{\bbE})$:
\begin{equation}\label{eq:13nqnp}\pi^{\fAdv}(\rho_0)\approx^{\tind}_{\epsilon}\underbrace{\fSim_1}_{\text{on }\bbC}(\tSFVPTwoIdealC(\underbrace{\fSim_0}_{\text{on }\bbC\text{ and outputs }b}(\rho_0)))\end{equation}
\end{defn}
The protocol is as follows. The protocol uses a family of collapsing hash functions and a succinct-ZK-AoK protocol (see Section \ref{sec:3.3.2}, \ref{sec:3.3.4}).
\begin{mdframed}[backgroundcolor=black!10]
    \begin{prtl}[$\fSFVPTwo$]\label{prtl:sfvp2}
        This protocol works under Set-up \ref{setup:11}. The protocol inputs are: public parameters $1^n,1^m,1^\kappa,1^{1/\epsilon}$, public inputs $f$, $pp$, $com$, client-side inputs $x$, $r$, and it should be satisfied that $com=\fCommit(x,r,pp)$.\par
        The protocol makes use of a family of collapsing hash functions $\fHash$ and a succinct-ZK-AoK protocol $\fsuccZKAoK$ that is\footnote{Here  $\fpoly_{\text{AoK}}(\delta)=O(\delta)$ as discussed in Theorem \ref{thm:existzk}. We use $\fpoly_{\text{AoK}}$ to allow us to refer to this polynomial explicitly.} $(\delta,\fpoly_{\text{AoK}}(\delta))$ for any $\delta>0$.
        \begin{enumerate}
            \item Both parties execute $\fSFVP$ (Protocol \ref{prtl:10}) for function $f$, client-side input $x$ and error tolerance $\epsilon$.\par
            Denote the server's outcome as $y$. If both parties are honest there should be $y=f(x)$.
            \item The server samples $h\leftarrow\fHash(1^m,1^\kappa)$, computes $h(y)=c$ and sends $(h,c)$ to the client.
            \item Define $L=16/(\epsilon^2\cdot\fpoly^{-1}_{\text{AoK}}(\frac{\epsilon}{4}))$.\par
            For each $i\in [L]$:
            \begin{enumerate}
                \item Using $\fsuccZKAoK$, the client proves the following statement to the server:
                \begin{equation}\label{eq:43}\exists (x,r):h(f(x))=c\land \fCommit(x,r,pp)=com\end{equation}
            \end{enumerate}
            If all the tests pass successfully the server stores $y$ as its output.
        \end{enumerate}
    \end{prtl}
\end{mdframed}
The completeness and the succinct communication property are from the protocol description. The primitives used here are all implied by collapsing hash functions.\par
Below we state the soundness against malicious servers and malicious clients.
\begin{thm}\label{thm:7.4}
    Protocol \ref{prtl:sfvp2} is $\epsilon$-sound against malcious servers.
\end{thm}
\begin{thm}\label{thm:7.5}
    Protocol \ref{prtl:sfvp2} is $\epsilon$-sound against malicious clients.
\end{thm}
The proof of soundness against malicious servers is relatively easier.
\begin{proof}[Proof of Theorem \ref{thm:7.4}]
For any efficient quantum adversary $\fAdv$ playing the role of the malicious server, apply the soundness of $\fSFVP$ protocol (Theorem \ref{thm:sfvps}) and ZK-property of the $\fsuccZKAoK$ protocol (Definition \ref{defn:3.zk})\footnote{Note that here the client proves statements to the server so the role of client and server are swapped compared to Definition \ref{defn:3.zk}.} we know there exist efficient quantum operations $\fSim_{\text{SFVP}}=(\fSim_{\text{SFVP,}0},\fSim_{\text{SFVP,}1})$, $(\fSim_{\text{ZK,}i})_{i\in [L]}$ working on the server side such that for any polynomial-size inputs that satisfy the set-up and any initial states $\rho_0\in \tD(\cH_{\bS}\otimes\cH_{\bbE})$, 
$$\fSFVPTwo^{\fAdv}(\rho_0)\approx_{\epsilon}^{\tind}\fHybrid(\rho_0)$$
where $\fHybrid$ operated on $\rho_0$ is defined as follows:
\begin{enumerate}
    \item Apply the server-side operation $\fSim_{\text{SFVP,}0}$ to $\rho_0$, which generates $b\in \{0,1\}$ as the server-side input to $\tSFVPIdeal$.
    \item Execute $\tSFVPIdeal$.
    \item Apply the server-side operation $\fSim_{\text{SFVP,}1}$. Then apply $\fAdv_1$, which denotes the adversary's operation after the first step.
    \item The adversary sends a message to the client.
    \item For each $i\in [L]$, apply the server-side operation $\fSim_{\text{ZK,}i}$.
\end{enumerate}
This could be translated to a simulator interacting with $\tSFVPTwoIdealS$ directly, which completes the proof.
\end{proof}
The proof of soundness against malicious clients requires some technical works. We will analyze the third step and construct a simulator that works in a cut-and-choose manner to find a round with high passing probability, which implies that the client indeed holds a $\tilde{x}$ that satisfies \eqref{eq:43}. Then this together with the statistically-binding of $\fCommit$ and collision-resistance of $h$ implies that the server either indeed holds $f(x)$ or catches the client cheating, which completes the proof.
\begin{proof}[Proof of Theorem \ref{thm:7.5}]
    For any efficient adversary $\fAdv$ playing the role of the client and any initial state, let's use $\sigma_i\in \tD(\cH_{\bbC}\otimes\cH_{\bbE}\otimes\cH_{\bS})$ to denote the joint states by the end of the $i$-th round of the third step of $\fSFVPTwo$. Then use $\Pi_{\fpass}^{(\leq i)}(\sigma_i)$ to denote the projection of $\sigma_i$ onto the space that all the rounds so far have passed. Use $\Pi_{\fpass}^{(i)}$ to denote the projection onto the space that the $i$-th round passes; $\Pi_{\ffail}^{(i)}$ is defined similarly.\par
    Use $S$ to denote the set of round index $i\in [L]$ that satisfies $\tr(\Pi_{\ffail}^{(i)}(\Pi_{\fpass}^{(\leq i-1)}(\sigma_{i-1})))\geq \frac{\epsilon}{4}\fpoly^{-1}(\frac{\epsilon}{4})$. Then we have $|S|/L\leq \frac{\epsilon}{4}$. For each $i\not\in S$, if $\tr(\Pi_{\fpass}^{(\leq i-1)}(\sigma_{i-1}))\geq \frac{\epsilon}{4}$, we could apply the $(\fpoly^{-1}(\frac{\epsilon}{4}),\frac{\epsilon}{4})$-soundness of the PoK; if $\tr(\Pi_{\fpass}^{(\leq i-1)}(\sigma_{i-1}))\leq \frac{\epsilon}{4}$ the trace of any state $\cE(\Pi_{\fpass}^{(\leq i-1)}(\sigma_{i-1}))$ is bounded directly. Combining both cases into one expression we have: there exists an efficient quantum operation $\fExt_i$ working on the client side that extracts to register $\bbC_{\ttemp}$ such that
    \begin{equation}\label{eq:44}\Pi^{\bbC_{\ttemp}}_{\tilde{x},\tilde{r}:h(f(\tilde{x}))=c\land\fCommit(\tilde{x},\tilde{r},pp)=com}\fExt_i(\Pi_{\fpass}^{(\leq i-1)}(\sigma_{i-1}))\approx_{\max\{\frac{\epsilon}{4},2\times\frac{\epsilon}{4}\}}\fExt_i(\Pi_{\fpass}^{(\leq i-1)}(\sigma_{i-1})).\end{equation}
(Note that we write $\max\{\frac{\epsilon}{4},2\times\frac{\epsilon}{4}\}$ to suggest the fact that we are combining two cases in one expression.)\par
Note that \eqref{eq:44} holds for $i\not\in S$. Recall that $|S|/L\leq \frac{\epsilon}{4}$. Thus we get the following result, for a randomly selected $i\in [L]$:
\begin{equation}\label{eq:442}\tDen\begin{bmatrix}i\leftarrow_r[L]\\\Pi^{\bbC_{\ttemp}}_{\tilde{x},\tilde{r}:h(f(\tilde{x}))=c\land\fCommit(\tilde{x},\tilde{r},pp)=com}(\fExt_i(\Pi_{\fpass}^{(\leq i-1)}(\sigma_{i-1})))\end{bmatrix}\approx_{\epsilon}\tDen\begin{bmatrix}i\leftarrow_r[L]\\\fExt_i(\Pi_{\fpass}^{(\leq i-1)}(\sigma_{i-1}))\end{bmatrix}\end{equation}
where we use $\tDen$ to denote the density operator of the last row when the variables are sampled according to the first row.\par
Now by the statistically binding of $\fCommit$ we know:
\begin{equation}\label{eq:443}\tDen\begin{bmatrix}i\leftarrow_r[L]\\\Pi^{\bbC_{\ttemp}}_{\tilde{x},\tilde{r}:h(f(\tilde{x}))=c\land\fCommit(\tilde{x},\tilde{r},pp)=com}\fExt_i(\Pi_{\fpass}^{(\leq i-1)}(\sigma_{i-1}))\end{bmatrix}\approx_{\fneg(\kappa)}\tDen\begin{bmatrix}i\leftarrow_r[L]\\\Pi^{\bbC_{\ttemp}}_{\tilde{x}:\tilde{x}=x\land h(f(\tilde{x}))=c}(\fExt_i(\Pi_{\fpass}^{(\leq i-1)}(\sigma_{i-1})))\end{bmatrix}\end{equation}
(Note that $\tilde{x}$ is the value of $\bbC_{\ttemp}$ and $x$ is the actual input of the protocol.)\par
Note that intuitively on the right hand side of \eqref{eq:443} the client already holds $x$ that satisfies $h(f(x))=c$; by the collision-resistance of $h$ the value $y$ in the server-side should satisfy $y=f(x)$. Formally, we have
\begin{equation}\label{eq:444}\tDen\begin{bmatrix}i\leftarrow_r[L]\\\Pi^{\bbC_{\ttemp}}_{\tilde{x}:\tilde{x}=x\land h(f(\tilde{x}))=c}(\fExt_i(\Pi_{\fpass}^{(\leq i-1)}(\sigma_{i-1})))\end{bmatrix}\approx_{\fneg(\kappa)}\tDen\begin{bmatrix}i\leftarrow_r[L]\\\Pi^{\bQ^{(out)}}_{y:y=f(x)}\Pi^{\bbC_{\ttemp}}_{\tilde{x}:\tilde{x}=x}(\fExt_i(\Pi_{\fpass}^{(\leq i-1)}(\sigma_{i-1})))\end{bmatrix}\end{equation}
Combining \eqref{eq:442}\eqref{eq:443}\eqref{eq:444} we get a good control on the server-side information in the passing space.\par
Thus the overall simulator $\fSim$ working on the client side could go as follows:
\begin{enumerate}
    \item Simulate the first two steps of the $\fSFVPTwo$ protocol on its own memory.
    \item Pick a random $i\leftarrow [L]$.
    \item For each $i^\prime<i$, simulate these rounds on its own memory and record the passing/failing outcome.
    \item For the $i$-th round, simulate but do not record the passing/failing outcome.
    \item For each $i^\prime>i$, simulate these rounds on its own memory and record the passing/failing outcome.
    \item Send the joint passing/failing outcome as the input $b$ in $\tSFVPTwoIdealC$.
\end{enumerate}
\end{proof}
\subsubsection{Construction of the secure two-party computations protocol}\label{sec:7.3.2}
In this section we construct the two-party computations (2PC) protocol in succinct communication, which proves Theorem \ref{thm:1.3}. The set-up is as follows.
\begin{setup}\label{setup:2pcsc}
    The set-up for 2PC with approximate security is as follows. The set-up is similar to  Set-up \ref{setup:twopc} with the following differences: the protocol takes an additional approximation error tolerance parameter $1^{1/\epsilon}$.\par
    The soundness is formalized as the approximate variant of Definition \ref{defn:3.7}.\par
    We want the communication complexity to be $\fpoly(n_A+n_B,\kappa,1/\epsilon)$.
\end{setup}
We would like to design a protocol for this problem in succinct communication.\par
The protocol is as follows. The protocol makes use of several primitives formalized in Section \ref{sec:3.3}.
\begin{mdframed}[backgroundcolor=black!10]
    \begin{prtl}\label{prtl:2pc}
        The protocol works under Set-up \ref{setup:2pcsc}. The protocol inputs are: public parameters $1^{n_A},1^{m_A},1^{n_B},1^{m_B},1^\kappa,1^{1/\epsilon}$, public function $f_A,f_B$, Alice-side input $x^{(A)}\in \{0,1\}^{n_A}$, Bob-side input $x^{(B)}\in \{0,1\}^{n_B}$.\par
        The protocol uses the following primitive:
        \begin{itemize}\item The oneway-based-2PC protocol discussed in Section \ref{sec:3.3.3}.\item The statistically-binding commitment scheme $\fCommit$ as formalized in Section \ref{sec:3.3.comm}.\item The decomposable randomized encoding scheme $(\fGarbleC,\fGarbleI,\fDc)$ as formalized in Section \ref{sec:3.3.5}. Recall that the inputs to a deterministic description of $\fGarbleC$ are the circuit, the shared random coins used by $\fGarbleC$ and $\fGarbleI$, and the internal random coins of $\fGarbleC$.\item A pseudorandom generator $\tPRG(1^\kappa,1^{\fpoly((n_A+n_B),\kappa)})$.\end{itemize} 
        \begin{enumerate}
            \item Alice and Bob perform 2PC for the following functions:
            \begin{enumerate}
                \item (Sampling public randomness for commitments):\par Sample $pp^{(A)},pp^{(B)}\leftarrow_r\{0,1\}^{3(n_A+n_B+1)\kappa^2}$ as the public outputs.
                \item (Sample sender-side randomness for commitments):\par Sample $r^{(A)}_{com}\leftarrow_r\{0,1\}^{(n_A+n_B+1)\kappa^2}$ as the Bob-side output; sample $r^{(B)}_{com}\leftarrow_r\{0,1\}^{(n_A+n_B+1)\kappa^2}$ as the Alice-side output.
                \item (Sample randomness for garbling scheme; $r$ denotes the shared randomness used in $\fGarbleC$ and $\fGarbleI$ and $s$ is used to generate the internal randomness of $\fGarbleC$):\par Sample $r^{(A)}\leftarrow_r\{0,1\}^{(n_A+n_B)\kappa}$, $s^{(A)}\leftarrow_r\{0,1\}^{\kappa}$ as the Bob-side outputs; sample  $r^{(A)}\leftarrow_r\{0,1\}^{(n_A+n_B)\kappa}$, $s^{(B)}\leftarrow_r\{0,1\}^{\kappa}$ as the Alice-side outputs.
                \item (Commitments):\par Compute $com^{(A)}=\fCommit((r^{(A)},s^{(A)}),r^{(A)}_{com},pp^{(A)})$ as the public output; compute $com^{(B)}=\fCommit((r^{(B)},s^{(B)}),r^{(B)}_{com},pp^{(B)})$ as the public output.
                \item (Input encoding):\par Compute $\fGarbleI((x^{(A)},x^{(B)}),r^{(A)})$ as the Alice-side output; compute $\fGarbleI((x^{(A)},x^{(B)}),r^{(B)})$ as the Bob-side output.
            \end{enumerate}
            Here the superscript $(A)$ is used to denote the intermediate variables used for evaluating $f_A$ and superscript $(B)$ is used to denote the intermediate variables used for evaluating $f_B$; who receives these variables could depend on the protocol design.
            \item Note that the commitments are as required in Set-up \ref{setup:11}. Bob uses the $\fSFVPTwo$ protocol to send $g_A(r^{(A)},s^{(A)})$ to Alice with soundness error tolerance $\epsilon/2$, where:
            $$g_A(r^{(A)},s^{(A)}):=\fGarbleC(C_{f_A},r^{(A)};\tPRG(s^{(A)}))$$
            where $C_{f_A}$ is the circuit description of $f_A$.\par
            Alice could compute
            $$\fDc(\fGarbleC(C_{f_A},r^{(A)};\tPRG(s^{(A)})),\fGarbleI((x^{(A)},x^{(B)}),r^{(A)}))$$
            to get $f_A(x^{(A)},x^{(B)})$.
            \item Alice uses the $\fSFVPTwo$ protocol to send $g_B(r^{(B)},s^{(B)})$ to Bob with soundness error tolerance $\epsilon/2$, where:
            $$g_B(r^{(B)},s^{(B)}):=\fGarbleC(C_{f_B},r^{(B)};\tPRG(s^{(B)}))$$
            where $C_{f_B}$ is the circuit description of $f_B$.\par
            Bob could compute
            $$\fDc(\fGarbleC(C_{f_B},r^{(B)};\tPRG(s^{(B)})),\fGarbleI((x^{(A)},x^{(B)}),r^{(B)}))$$
            to get $f_B(x^{(A)},x^{(B)})$.
        \end{enumerate}
    \end{prtl}
\end{mdframed}
The communication is succinct by the succinctness of $\fSFVPTwo$ protocol and the fact that the two-party computations in step 1 are all performed on inputs and functions with size $\fpoly(n_A+n_B,\kappa)$. The assumptions and primitives used in this protocol are all implied by the existence of the collapsing hash functions. The completeness is by the protocol description.\par 
The soundness is stated below.
\begin{thm}
    Protocol \ref{prtl:2pc} is $\epsilon$-sound.
\end{thm}
\begin{proof}
    We prove the soundness against malicious Bob and the proof of the soundness against malicious Alice is similar.\par
    For any efficient quantum adversary $\fAdv$ playing the role of Bob, we would like to construct efficient quantum operations $\fSim$ playing the role of Bob that interacts with $\tTwoPCIdeal^{(B)}$ and simulates the real protocol.\par
    First apply the soundness of 2PC (the first step of Protocol \ref{prtl:2pc}), $\fSFVPTwo$ against malicious clients (the second step of Protocol \ref{prtl:2pc}), and $\fSFVPTwo$ against malicious servers (the third step of Protocol \ref{prtl:2pc}) we know that, there exist simulators (which interacts with the corresponding ideal functionality) $\fSim_{\text{step 1}}$, $\fSim_{\text{step 2}}$, $\fSim_{\text{step 3}}$ working on the Bob side such that for any polynomial-size inputs and any initial state $\rho_0\in \tD(\cH_{\bS^{(B)}}\otimes\cH_{\bbE})$, the output state of the protocol is $\epsilon$-indistinguishable to $\fHybrid_1(\rho_0)$, where $\fHybrid_1$ is as follows:
    \begin{enumerate}
        \item Apply $(\fSim_{\text{step 1}}\circ \tTwoPCIdeal^{(B)}_{\text{step 1}})$. Here $\tTwoPCIdeal^{(B)}_{\text{step 1}}$ is the ideal functionality of the 2PC protocol used in the first step of Protocol \ref{prtl:2pc} (note that this is different from $\tTwoPCIdeal$, which denotes the ideal functionality of the whole protocol).
        \item Apply $(\fSim_{\text{step 2}}\circ \tSFVPTwoIdealC(g_A))$.
        \item Apply $(\fSim_{\text{step 3}}\circ \tSFVPTwoIdealS(g_B))$.
    \end{enumerate}
    Further apply the definition of $\tPRG$ and the hiding property of $\fCommit$ we know $\fHybrid_1(\rho_0)$ is indistinguishable to $\fHybrid_2(\rho_0)$ where $\fHybrid_2$ is as follows:
    \begin{enumerate}
        \item  Apply $(\fSim_{\text{step 1}}\circ \tTwoPCIdeal^{\prime(B)}_{\text{step 1}})$, where $\tTwoPCIdeal^{\prime(B)}_{\text{step 1}}$ is similar to $\tTwoPCIdeal^{(B)}_{\text{step 1}}$ with the following difference: the commitment to $r^{(B)}$ (that is, $com^{(B)}$) is replaced by a commitment to a random string.
        \item Apply $(\fSim_{\text{step 2}}\circ \tSFVPTwoIdealC(g_A^\prime))$, where $g^\prime_A$ is a randomized function as follows: $g^\prime_A(r^{(A)})=\fGarbleC(C_{f_A},r^{(A)})$. Note the definition of $\tSFVPTwoIdealC$ (Definition \ref{defn:7.7}) could be adapted to randomized function in the natural way.
        \item Apply $(\fSim_{\text{step 3}}\circ \tSFVPTwoIdealS(g_B^\prime))$, where $g^\prime_B$ is a randomized function as follows: $g^\prime_B(r^{(B)})=\fGarbleC(C_{f_B},r^{(B)})$. Note the definition of $\tSFVPTwoIdealS$ (Definition \ref{defn:7.6}) could be adapted to randomized function in the natural way.
    \end{enumerate}
    Then what remains to be analyzed is the garbled circuits. By the soundness of the garbled circuits scheme (Definition \ref{defn:3.14gc}, Theorem \ref{thm:yaogc}) there exist efficient simulator $\fSim_{\text{ie}},\fSim_{\text{ce}}$ such that
    \begin{equation}
        \tDist\begin{bmatrix}
            \tilde{x}\leftarrow\fSim_{\text{ie}}()\\ (\tilde{x},\fSim_{\text{ce}}(\tilde{x},C_{f_B}(x^{(A)},x^{(B)})))
        \end{bmatrix}\approx^{\tind}\tDist\begin{bmatrix}r^{(B)}\leftarrow_r\{0,1\}^{(n_A+n_B)\kappa}\\
            \hat{x}:=\fGarbleI((x^{(A)},x^{(B)}),r^{(B)})\\ (\hat{x},\fGarbleC(C_{f_B},r^{(B)}))\end{bmatrix}
    \end{equation}
    Now we are ready to construct the overall simulator $\fSim$ that interacts with $\tTwoPCIdeal^{(B)}$ as follows.
    \begin{enumerate}
        \item $\fSim$ first works as $\fSim_{\text{step 1}}$ as in the first step of $\fHybrid_2$; if $\fSim_{\text{step 1}}$ early-aborts then $\fSim$ early aborts. Otherwise $\fSim$ could extract the input $x^{(B)}$ from $\fSim_{\text{step 1}}$'s preparation of inputs to $\tTwoPCIdeal^{\prime(B)}_{\text{step 1}}$.\par
        For the simulation of the output string from $\tTwoPCIdeal^{\prime(B)}_{\text{step 1}}$, $\fSim$ uses $\tilde{x}\leftarrow\fSim_{\text{ie}}()$ to simulate the input encoding $\fGarbleI((x^{(A)},x^{(B)}),r^{(B)})$.
        \item $\fSim$ then works as $\fSim_{\text{step 2}}$ and stores $b_{\text{step 2}}\in \{0,1\}$ which determines whether the server should abort in $\tSFVPTwoIdealC$.
        \item $\fSim$ then works as $\fSim_{\text{step 3}}$ and aborts if $\fSim_{\text{step 3}}$ aborts; otherwise it queries $\tTwoPCIdeal^{(B)}$ with $x^{(B)}$ and gets $f_B(x^{(A)},x^{(B)})$. Then it computes $\fSim_{\text{ce}}(\tilde{x},f_B(x^{(A)},x^{(B)}))$ to simulate the outputs of $\tSFVPTwoIdealS$.
        \item Finally $\fSim$ determines whether Alice should receives its outputs $f_A(x^{(A)},x^{(B)})$ based on $b_{\text{step 2}}$ (that is, use $b_{\text{step 2}}$ as the input to $\tTwoPCIdeal^{(B)}$; recall in Definition \ref{defn:3.7} $\tTwoPCIdeal^{(B)}$ receives a bit from Bob that determines whether Alice should receive its output in the final step).
    \end{enumerate}
    Comparing each case we know $\fHybrid_2(\rho_0)$ is indistinguishable to $\fSim(\rho_0)$.\par
    Combining this chain of hybrids completes the proof.
\end{proof}
\subsection{Classical-channel Secure Two-party Computations Protocol}\label{sec:7.4}
In this section we construct our classical-channel 2PC protocol in succinct communication.
\begin{itemize}
    \item In Section \ref{sec:7.4.1} we give a result on classical-channel 2PC (without considering the succinctness of communication). We construct this protocol as a preparation for the final construction.
    \item In Section \ref{sec:7.4.2} we give our protocol for classical-channel 2PC in succinct communication.
\end{itemize}
\subsubsection{On classical-channel 2PC without the succinct communication requirement}\label{sec:7.4.1}
We first construct a classical-channel 2PC protocol with approximate security assuming NTCF. This is by applying the RSPV-for-BB84 states (see \cite{BGKPV23,zhang24} and Theorem \ref{thm:bb84rspv}) to the oneway-based-2PC in \cite{BCKM20}.
\begin{mdframed}[backgroundcolor=black!10]
    \begin{prtl}\label{prtl:cc2pcpre}
    This protocol works under Set-up \ref{setup:2pcsc} but we do not put restriction on the communication complexity.\par
    Suppose the 2PC protocol in \cite{BCKM20} contains $L$ rounds of quantum communication in the form of ``One party sends random $N$-qubits BB84 states to the other party''. Our protocol replaces all these quantum communication by RSPV protocols for $N$-qubits BB84 states with error tolerance $\epsilon/L$.
    \end{prtl}
\end{mdframed}
The assumption is the existence of NTCF; the completeness and efficiency are from the protocol description.
\begin{thm}\label{thm:7.7}
    Protocol \ref{prtl:cc2pcpre} is $\epsilon$-sound.
\end{thm}
\begin{proof}
    First the original protocol in \cite{BCKM20} is $\fneg(\kappa)$-sound.\par
    Then suppose we are working on a protocol that is $\epsilon_0$-sound against malicious party B and we are going to replace a quantum communication step in the form of ``One party sends random $N$-qubits BB84 states to the other party'' by an RSPV for this state family with error tolerance $\epsilon/L$. Then:
    \begin{itemize}
        \item If party B is the receiver: by the soundness of RSPV the new protocol is $\epsilon_0+\epsilon/L$-sound against malicious party B.
        \item If party B is the sender: note that RSPV protocol does not take private inputs; B's any malicious operation in the new protocol could always be simulated in the original protocol as follows: the party B prepares the output states of the RSPV protocol on its own and sends A's part to A. Since the original protocol is $\epsilon_0$-sound against malicious B the new protocol is also $\epsilon_0$-sound against malicious B.
    \end{itemize}
\end{proof}
\subsubsection{Classical-channel 2PC in succinct communication}\label{sec:7.4.2}
Finally we are ready to construct our classical-channel 2PC protocol in succinct communication. This is achieved by using Protocol \ref{prtl:cc2pcpre} to replace the first step of Protocol \ref{prtl:2pc} and use the RSPV for $\ket{0}\ket{\tilde{x}_0}+\ket{1}\ket{\tilde{x}_1}$ to replace the quantum communication in the second and third step of Protocol \ref{prtl:2pc}.
\begin{mdframed}[backgroundcolor=black!10]
    \begin{prtl}\label{prtl:cc2pc}
    This protocol works under Set-up \ref{setup:2pcsc}.\par
    \begin{enumerate}
        \item Use Protocol \ref{prtl:cc2pcpre} to perform the 2PC in the step 1 of Protocol \ref{prtl:2pc}, with error tolerance $\frac{\epsilon}{3}$.
        \item Consider the step 2 of Protocol \ref{prtl:2pc} executed with error tolerance $\frac{\epsilon}{6}$. Suppose in this step there are $L$ rounds of quantum communication in the form of ``Bob samples keys $\tilde{x}_0,\tilde{x}_1$ and sends the states $\frac{1}{\sqrt{2}}(\ket{0}\ket{\tilde{x}_0}+\ket{1}\ket{\tilde{x}_1})$'' to Alice. (Explicitly, $\tilde{x}_1,\tilde{x}_1$ are $x_0||r_0^{(in)}||r_0^{(out)}$ and $x_1||r_1^{(in)}||r_1^{(out)}$ in the first step of Protocol \ref{prtl:3}.) Use the RSPV for this type of states (Theorem \ref{thm:kpsrspv}) to replace this step; the error tolerance for each such step is chosen to be $\frac{\epsilon}{6L}$.
        \item Same as the step 2 above with the following differences: we compile the step 3 of Protocol \ref{prtl:2pc} and the roles of Alice and Bob are swapped.
    \end{enumerate}
    \end{prtl}
\end{mdframed}
The assumptions are the existence of NTCF and collapsing hash functions; the completeness, efficiency and succinct communication property are from the protocol description and the properties of Protocol \ref{prtl:2pc}.\par
The protocol is $\epsilon$-sound based on a similar proof to the proof of Theorem \ref{thm:7.7}.
\appendix
\section{On Succinct Arguments for QMA}\label{sec:a}
We note that our work on SFS also leads to a new approach for constructing succinct arguments for QMA. In this section we roughly describe our approach. We do not formally work on this part so we put it in the appendix.
\subsection{Background}
Let's first review the notion of succinct argument in Section \ref{sec:3.3.4}. Suppose the server claims that ``I know a $w$ such that $f(w)=0$'', where $f$ is a public function. How efficiently could the client verify the server's claim? When $f$ is an efficient classical function, Kilian's protocol \cite{Kilian92} allows us to achieve this task with succinct client-side computation, that is, independent of both the function evaluation time and the witness size (size of $w$).\par
 So what if the server claims a QMA statement? That is, the server claims that it holds a quantum state $w$ such that $f(w)=0$, where $f$ denotes a (family of) quantum circuit. The \emph{succinct arguments for QMA} problem seeks for a protocol for this problem where the client side computation is succinct. A closely related problem is the classical verification of quantum computations (CVQC) problem, which seeks for a verification protocol where the client-side computation and the communication should be completely classical.\par
  The first and perhaps the most famous construction for CVQC is given by Mahadev \cite{MahadevVerification}; later a series of new constructions are proposed \cite{CCT,GVRSP,cvqcinlt,NZ23}. For the construction of succinct arguments for QMA, there are currently several existing constructions \cite{CCT,Faisal23,BKLMM22,MNZ24,GTNV24} and two of them are from standard assumptions \cite{MNZ24,GTNV24}. These succinct arguments constructions are based on (or highly related to) CVQC constructions and inherit the desirable property that the client is completely classical, which are thus called classical succinct arguments for QMA. In terms of the assumptions, \cite{GTNV24} is based on NTCF with distributional adaptive hardcore bit property, and \cite{MNZ24} is based on FHE and collapsing hash functions. In more detail:
\begin{itemize}
    \item \cite{GTNV24} constructs their succinct arguments for QMA protocol via a protocol called classical commitments to quantum states. This is a strengthening of the measurement protocol in Mahadev's CVQC work \cite{MahadevVerification}. Their results are based on an assumption called NTCF with distributional strong adaptive hardcore bit property (defined in their work).
    \item \cite{MNZ24} constructs their succinct arguments for QMA protocol based on the KLVY compiler \cite{KLVY22}, which is based on homomorphic encryption and nonlocal games. \cite{NZ23} constructs a CVQC protocol following the approach of \cite{KLVY22}; \cite{MNZ24} could be seen as a further step in this approach. \cite{MNZ24} assume a mild version of quantum homomorphic encryption \cite{MNZ24,Mahadev2017ClassicalHE,FHELWE} and collapsing hash functions \cite{Unruh16}.
\end{itemize}
The Mahadev's CVQC construction \cite{MahadevVerification} is under an assumption called the extended NTCF with the adaptive hardcore bit property. This assumption seems incomparable to both the assumptions in \cite{GTNV24} (although both could be seen as variants of NTCF) and the assumptions in \cite{MNZ24}.
        
\subsection{Our Approach For Classical Succinct Arguments for QMA}\label{sec:a.8}
Here we propose a different approach for constructing classical succinct arguments for QMA. We assume the existence of extended NTCF with the adaptive hardcore bit property and collapsing hash functions, which is different from previous works \cite{BKLMM22,MNZ24,GTNV24} and is closer to Mahadev's CVQC protocol. Our approach is to use our results and existing results to compile Mahadev's protocol. The compilation is relatively simple given our results and existing results.\par
We first give a brief review for the structure of Mahadev's protocol. The Mahadev's protocol roughly goes as follows:\par
\begin{enumerate}
    \item For each $i\in [L]$, the client samples $pk_i$ according to certain distribution and sends it to the server.
    \item The server commits to its witness state using these keys and sends back the commitment (which is a classical string).
    \item For each $i\in [L]$, the client randomly samples a challenge bit $c_i\in \{0,1\}$ and sends it to the server.
    \item The server does certain measurements on the committed state depending on the challenge bits and sends back the results; the client checks the results with certain relations.
\end{enumerate}
Then the compilation is as follows:
\begin{enumerate}
    \item The client could use the SFS protocol for sampling and sending $(pk_i)_{i\in [L]}$ in the first step. A PRG is used for generating the randomness.\par
    The calling to the $\fSFS$ protocol introduces quantum communication; we could further use an RSPV \cite{zhang24} to compile this step to a classical-channel protocol.
    \item The communication in steps 2-4 could be made succinct using the compiler in \cite{BKLMM22}.
\end{enumerate}
Note that existing works often consider the first step as the main obstacle for getting a succinct arguments protocol for QMA. \cite{BKLMM22,MNZ24} In our proposal this is relatively simple given our SFS protocol and existing RSPV protocol.
    \bibliographystyle{plain}
\bibliography{bib}

\begin{thebibliography}{10}

\bibitem{AMR22}
Navid Alamati, Giulio Malavolta, and Ahmadreza Rahimi.
\newblock Candidate trapdoor claw-free functions from group actions with applications to quantum protocols.
\newblock In {\em Theory of Cryptography: 20th International Conference, TCC 2022, Chicago, IL, USA, November 7–10, 2022, Proceedings, Part I}, page 266–293, Berlin, Heidelberg, 2022. Springer-Verlag.

\bibitem{GMP}
Alexander~Poremba Alexandru~Gheorghiu, Tony~Merger.
\newblock Quantum cryptography with classical communication: parallel remote state preparation for copy-protection, verification, and more.
\newblock 2022.

\bibitem{ATY17}
Anurag Anshu, Dave Touchette, Penghui Yao, and Nengkun Yu.
\newblock Exponential separation of quantum communication and classical information.
\newblock In {\em Proceedings of the 49th Annual ACM SIGACT Symposium on Theory of Computing}, STOC 2017, page 277–288, New York, NY, USA, 2017. Association for Computing Machinery.

\bibitem{AppleBaum2017}
Benny Applebaum.
\newblock Garbled circuits as randomized encodings of functions: a primer.
\newblock {\em Electron. Colloquium Comput. Complex.}, TR17, 2017.

\bibitem{AroraBarak}
Sanjeev Arora and Boaz Barak.
\newblock {\em Computational Complexity: A Modern Approach}.
\newblock Cambridge University Press, USA, 1st edition, 2009.

\bibitem{AJW12}
Gilad Asharov, Abhishek Jain, Adriana L{\'o}pez-Alt, Eran Tromer, Vinod Vaikuntanathan, and Daniel Wichs.
\newblock Multiparty computation with low communication, computation and interaction via threshold fhe.
\newblock {\em IACR Cryptol. ePrint Arch.}, 2011:613, 2012.

\bibitem{BCKM20}
James Bartusek, Andrea Coladangelo, Dakshita Khurana, and Fermi Ma.
\newblock {One-way functions imply secure computation in a quantum world.}
\newblock {\em Lect. Notes Comput. Sci.}, 12825:467, 2021.

\bibitem{BKLMM22}
James Bartusek, Yael~Tauman Kalai, Alex Lombardi, Fermi Ma, Giulio Malavolta, Vinod Vaikuntanathan, Thomas Vidick, and Lisa Yang.
\newblock Succinct classical verification of quantum computation.
\newblock In {\em IACR Cryptology ePrint Archive}, 2022.

\bibitem{BK}
James Bartusek and Dakshita Khurana.
\newblock Cryptography with certified deletion.
\newblock 2022.
\newblock \url{https://eprint.iacr.org/2022/1178}.

\bibitem{QRO}
Dan Boneh, {\"O}zg{\"u}r Dagdelen, Marc Fischlin, Anja Lehmann, Christian Schaffner, and Mark Zhandry.
\newblock Random oracles in a quantum world.
\newblock In Dong~Hoon Lee and Xiaoyun Wang, editors, {\em Advances in Cryptology -- ASIACRYPT 2011}, pages 41--69, Berlin, Heidelberg, 2011. Springer Berlin Heidelberg.

\bibitem{BCMVV}
Zvika Brakerski, Paul Christiano, Urmila Mahadev, Umesh~V. Vazirani, and Thomas Vidick.
\newblock A cryptographic test of quantumness and certifiable randomness from a single quantum device.
\newblock In {\em 59th {IEEE} Annual Symposium on Foundations of Computer Science, {FOCS} 2018, Paris, France, October 7-9, 2018}, pages 320--331, 2018.

\bibitem{BGKPV23}
Zvika Brakerski, Alexandru Gheorghiu, Gregory~D. Kahanamoku-Meyer, Eitan Porat, and Thomas Vidick.
\newblock Simple tests of quantumness also certify qubits.
\newblock 2023.

\bibitem{BKVV}
Zvika Brakerski, Venkata Koppula, Umesh Vazirani, and Thomas Vidick.
\newblock Simpler proofs of quantumness, 05 2020.

\bibitem{FHELWE}
Zvika Brakerski and Vinod Vaikuntanathan.
\newblock Efficient fully homomorphic encryption from (standard) lwe.
\newblock In {\em Proceedings of the 2011 IEEE 52Nd Annual Symposium on Foundations of Computer Science}, FOCS '11, pages 97--106, Washington, DC, USA, 2011. IEEE Computer Society.

\bibitem{BGK17}
Sergey Bravyi, David Gosset, and Robert K{\"o}nig.
\newblock Quantum advantage with shallow circuits.
\newblock {\em Science}, 362:308 -- 311, 2017.

\bibitem{BI20}
Anne Broadbent and Rabib Islam.
\newblock {Quantum encryption with certified deletion.}
\newblock {\em Lect. Notes Comput. Sci.}, 12552:92, 2020.

\bibitem{PV10}
Harry Buhrman, Nishanth Chandran, Serge Fehr, Ran Gelles, Vipul Goyal, Rafail~M. Ostrovsky, and Christian Schaffner.
\newblock Position-based quantum cryptography: Impossibility and constructions.
\newblock In {\em IACR Cryptology ePrint Archive}, 2010.

\bibitem{rorevisited}
Ran Canetti, Oded Goldreich, and Shai Halevi.
\newblock The random oracle methodology, revisited.
\newblock {\em J. ACM}, 51(4):557–594, July 2004.

\bibitem{CCT}
Nai-Hui Chia, Kai-Min Chung, and Takashi Yamakawa.
\newblock Classical verification of quantum computations with efficient verifier.
\newblock In Rafael Pass and Krzysztof Pietrzak, editors, {\em Theory of Cryptography}, pages 181--206, Cham, 2020. Springer International Publishing.

\bibitem{CMSZ21}
Alessandro Chiesa, Fermi Ma, Nicholas Spooner, and Mark Zhandry.
\newblock Post-quantum succinct arguments: Breaking the quantum rewinding barrier.
\newblock {\em 2021 IEEE 62nd Annual Symposium on Foundations of Computer Science (FOCS)}, pages 49--58, 2021.

\bibitem{qfactory}
Alexandru Cojocaru, L{\'{e}}o Colisson, Elham Kashefi, and Petros Wallden.
\newblock Qfactory: Classically-instructed remote secret qubits preparation.
\newblock In Steven~D. Galbraith and Shiho Moriai, editors, {\em Advances in Cryptology - {ASIACRYPT} 2019 - 25th International Conference on the Theory and Application of Cryptology and Information Security, Kobe, Japan, December 8-12, 2019, Proceedings, Part {I}}, volume 11921 of {\em Lecture Notes in Computer Science}, pages 615--645. Springer, 2019.

\bibitem{DFH12}
Ivan Damg{\aa}rd, Sebastian Faust, and Carmit Hazay.
\newblock Secure two-party computation with low communication.
\newblock In {\em IACR Cryptology ePrint Archive}, 2012.

\bibitem{ES20}
Edward Eaton and Fang Song.
\newblock {\em A Note on the Instantiability of the Quantum Random Oracle}, pages 503--523.
\newblock 04 2020.

\bibitem{Faisal23}
Islam Faisal.
\newblock Interactive oracle arguments in the qrom and applications to succinct verification of quantum computation.
\newblock Cryptology ePrint Archive, Paper 2023/421, 2023.
\newblock \url{https://eprint.iacr.org/2023/421}.

\bibitem{GentryFHE}
Craig Gentry.
\newblock Fully homomorphic encryption using ideal lattices.
\newblock In {\em Proceedings of the Forty-First Annual ACM Symposium on Theory of Computing}, STOC '09, page 169–178, New York, NY, USA, 2009. Association for Computing Machinery.

\bibitem{GVRSP}
Alexandru Gheorghiu and Thomas Vidick.
\newblock Computationally-secure and composable remote state preparation.
\newblock In David Zuckerman, editor, {\em 60th {IEEE} Annual Symposium on Foundations of Computer Science, {FOCS} 2019, Baltimore, Maryland, USA, November 9-12, 2019}, pages 1024--1033. {IEEE} Computer Society, 2019.

\bibitem{GB01}
Shafi Goldwasser and Mihir Bellare.
\newblock Lecture notes on cryptography.
\newblock 2001.

\bibitem{GLSV20}
Alex~B. Grilo, Huijia Lin, Fang Song, and Vinod Vaikuntanathan.
\newblock Oblivious transfer is in miniqcrypt.
\newblock In {\em Advances in Cryptology – EUROCRYPT 2021: 40th Annual International Conference on the Theory and Applications of Cryptographic Techniques, Zagreb, Croatia, October 17–21, 2021, Proceedings, Part II}, page 531–561, Berlin, Heidelberg, 2021. Springer-Verlag.

\bibitem{GTNV24}
Sam Gunn, Yael~Tauman Kalai, Anand Natarajan, and Agi Villanyi.
\newblock {Classical Commitments to Quantum States}.
\newblock 4 2024.

\bibitem{HW15}
Pavel Hubacek and Daniel Wichs.
\newblock On the communication complexity of secure function evaluation with long output.
\newblock In {\em Proceedings of the 2015 Conference on Innovations in Theoretical Computer Science}, ITCS '15, page 163–172, New York, NY, USA, 2015. Association for Computing Machinery.

\bibitem{asyao}
Zahra Jafargholi and Daniel Wichs.
\newblock Adaptive security of yao's garbled circuits.
\newblock In {\em Theory of Cryptography Conference}, 2016.

\bibitem{KLVY22}
Yael Kalai, Alex Lombardi, Vinod Vaikuntanathan, and Lisa Yang.
\newblock Quantum advantage from any non-local game.
\newblock In {\em Proceedings of the 55th Annual ACM Symposium on Theory of Computing}, STOC 2023, page 1617–1628, New York, NY, USA, 2023. Association for Computing Machinery.

\bibitem{KLtextbook}
Jonathan Katz and Yehuda Lindell.
\newblock {\em Introduction to Modern Cryptography, Second Edition}.
\newblock Chapman \& Hall/CRC, 2nd edition, 2014.

\bibitem{Kilian92}
Joe Kilian.
\newblock A note on efficient zero-knowledge proofs and arguments (extended abstract).
\newblock In S.~Rao Kosaraju, Mike Fellows, Avi Wigderson, and John~A. Ellis, editors, {\em Proceedings of the 24th Annual {ACM} Symposium on Theory of Computing, May 4-6, 1992, Victoria, British Columbia, Canada}, pages 723--732. {ACM}, 1992.

\bibitem{LinPass14}
Huijia Lin and Rafael Pass.
\newblock Succinct garbling schemes and applications.
\newblock {\em IACR Cryptol. ePrint Arch.}, 2014:766, 2014.

\bibitem{Mahadev2017ClassicalHE}
Urmila Mahadev.
\newblock Classical homomorphic encryption for quantum circuits.
\newblock In Mikkel Thorup, editor, {\em 59th {IEEE} Annual Symposium on Foundations of Computer Science, {FOCS} 2018, Paris, France, October 7-9, 2018}, pages 332--338. {IEEE} Computer Society, 2018.

\bibitem{MahadevVerification}
Urmila Mahadev.
\newblock Classical verification of quantum computations.
\newblock In Mikkel Thorup, editor, {\em 59th {IEEE} Annual Symposium on Foundations of Computer Science, {FOCS} 2018, Paris, France, October 7-9, 2018}, pages 259--267. {IEEE} Computer Society, 2018.

\bibitem{MNZ24}
Tony Metger, Anand Natarajan, and Tina Zhang.
\newblock { Succinct Arguments for QMA from Standard Assumptions via Compiled Nonlocal Games }.
\newblock In {\em 2024 IEEE 65th Annual Symposium on Foundations of Computer Science (FOCS)}, pages 1193--1201, Los Alamitos, CA, USA, October 2024. IEEE Computer Society.

\bibitem{Naorcomm}
Moni Naor.
\newblock Bit commitment using pseudo-randomness.
\newblock In {\em Proceedings of the 9th Annual International Cryptology Conference on Advances in Cryptology}, CRYPTO '89, page 128–136, Berlin, Heidelberg, 1989. Springer-Verlag.

\bibitem{NZ23}
Anand Natarajan and Tina Zhang.
\newblock Bounding the quantum value of compiled nonlocal games: from chsh to bqp verification.
\newblock 2023.

\bibitem{NielsenChuangs}
Michael~A. Nielsen and Isaac~L. Chuang.
\newblock {\em Quantum Computation and Quantum Information: 10th Anniversary Edition}.
\newblock Cambridge University Press, New York, NY, USA, 10th edition, 2011.

\bibitem{PR22}
Christopher Portmann and Renato Renner.
\newblock Security in quantum cryptography.
\newblock {\em Rev. Mod. Phys.}, 94:025008, Jun 2022.

\bibitem{QWW18}
Willy Quach, Hoeteck Wee, and Daniel Wichs.
\newblock Laconic function evaluation and applications.
\newblock In {\em 2018 IEEE 59th Annual Symposium on Foundations of Computer Science (FOCS)}, pages 859--870, 2018.

\bibitem{regevLWE}
Oded Regev.
\newblock On lattices, learning with errors, random linear codes, and cryptography.
\newblock {\em J. ACM}, 56(6), September 2009.

\bibitem{Sattath22}
Or~Sattath.
\newblock Uncloneable cryptography.
\newblock {\em Communications of the ACM}, 66:78 -- 86, 2022.

\bibitem{Unruh15}
Dominique Unruh.
\newblock Revocable quantum timed-release encryption.
\newblock {\em J. ACM}, 62(6), December 2015.

\bibitem{Unruh16a}
Dominique Unruh.
\newblock Collapse-binding quantum commitments without random oracles.
\newblock In {\em Proceedings, Part II, of the 22nd International Conference on Advances in Cryptology --- ASIACRYPT 2016 - Volume 10032}, page 166–195, Berlin, Heidelberg, 2016. Springer-Verlag.

\bibitem{Unruh16}
Dominique Unruh.
\newblock Computationally binding quantum commitments.
\newblock In {\em Proceedings, Part II, of the 35th Annual International Conference on Advances in Cryptology --- EUROCRYPT 2016 - Volume 9666}, page 497–527, Berlin, Heidelberg, 2016. Springer-Verlag.

\bibitem{watroustqi}
John Watrous.
\newblock The theory of quantum information.
\newblock 2018.

\bibitem{YZ22}
Takashi Yamakawa and Mark Zhandry.
\newblock Verifiable quantum advantage without structure.
\newblock {\em J. ACM}, 71(3), June 2024.

\bibitem{YaoGCOrigin}
A.~C. Yao.
\newblock How to generate and exchange secrets.
\newblock In {\em 27th Annual Symposium on Foundations of Computer Science (sfcs 1986)}, pages 162--167, Oct 1986.

\bibitem{Zhand17}
Mark Zhandry.
\newblock Quantum lightning never strikes the same state twice. or: Quantum money from cryptographic assumptions.
\newblock {\em Journal of Cryptology}, 34, 2017.

\bibitem{Zhanprf}
Mark Zhandry.
\newblock How to construct quantum random functions.
\newblock {\em J. ACM}, 68(5), August 2021.

\bibitem{cvqcinlt}
J.~Zhang.
\newblock Classical verification of quantum computations in linear time.
\newblock In {\em 2022 IEEE 63rd Annual Symposium on Foundations of Computer Science (FOCS)}, pages 46--57, Los Alamitos, CA, USA, nov 2022. IEEE Computer Society.

\bibitem{jiayu20}
Jiayu Zhang.
\newblock {\em Succinct Blind Quantum Computation Using a Random Oracle}, page 1370–1383.
\newblock Association for Computing Machinery, New York, NY, USA, 2021.

\bibitem{qafirstversion}
Jiayu Zhang.
\newblock {A Quantum Approach for Reducing Communications in Classical Cryptographic Primitives}, 10 2023.
\newblock \url{https://arxiv.org/abs/2310.05213v2}.

\bibitem{zhang24}
Jiayu Zhang.
\newblock Formulations and constructions of remote state preparation with verifiability, with applications.
\newblock 2023.

\end{thebibliography}

	\end{document}